%% file: llncs.tex
\documentclass{llncs}
\usepackage[utf8x]{inputenc}
\usepackage[english]{babel}
\usepackage{bussproofs}
\usepackage{amsmath}
\usepackage{amssymb}
\usepackage[disable]{todonotes}
\usepackage{stmaryrd}
\usepackage{xcolor}
\definecolor{linkblue}{rgb}{0.1,0,0.5}
\usepackage[pdfborder={0 0 0},colorlinks=true,linkcolor=linkblue,citecolor=linkblue,filecolor=linkblue,urlcolor=linkblue]{hyperref}
\usepackage{scrextend}
\usepackage{graphicx}
\usepackage{floatflt}
\usepackage{tikz}
\usetikzlibrary{calc}

\input{defs}

\def\version{EXTENDED}
\def\submission{FALSE}

\begin{document}

\title{Closing the Gap -- Formally Verifying Dynamically Typed Programs like Statically Typed Ones Using Hoare Logic\ifx \version\extended \newline -- Extended Version --\fi}

%\titlebanner{\ifx \version\extended extended version -- draft \else preprint \fi}        % These are ignored unless
%\preprintfooter{Practical Verification of Dynamically Typed Programs\ifx \version\extended -- Extended Version -- Draft\fi}   % 'preprint' option specified.

\ifx \submission\true
\author{-- omitted for submission --}
\institute{-- omitted for submission --}
\else
\author{Bj\"orn Engelmann \and Ernst-R\"udiger Olderog \and Nils Erik Flick}
\institute{University of Oldenburg, Germany \\ 
           \email{bjoern.engelmann@uni-oldenburg.de}
           \email{ernst.ruediger.olderog@informatik.uni-oldenburg.de}
           \email{flick@informatik.uni-oldenburg.de}}
\fi

\maketitle

\begin{abstract}
Dynamically typed object-oriented languages enable programmers to write elegant, reusable and extensible programs. However, with the current methodology for program verification, the absence of static type information creates significant overhead.
Our proposal is two-fold:

First, we propose a layer of abstraction hiding the complexity of dynamic typing when provided with sufficient type information. Since this essentially creates the illusion of verifying a statically-typed program, the effort required is equivalent to the statically-typed case.

Second, we show how the required type information can be efficiently derived for all type-safe programs by integrating a type inference algorithm into Hoare logic, yielding a semi-automatic procedure allowing the user to focus on those typing problems really requiring his attention. While applying type inference to dynamically typed programs is a well-established method by now, our approach complements conventional soft typing systems by offering formal proof as a third option besides modifying the program (static typing) and accepting the presence of runtime type errors (dynamic typing).
       
       %In fact, it is even possible to use conventional statically-typed assertion languages to reason about states of dynamically typed programs.

% In essence, the only thing additionally required when verifying a dynamically typed program using this
% novel method are small manual proofs safeguarding places where type safety cannot be established automatically.

% effort is only required
%        for programs expoiting the freedom offered by dynamic typing
% 
% While programs not containing any hard typing problems (the majority)
% can be verified just like in statically typed languages, harder cases require manual proof.
% 
% However, as the semi-automatic approach allows the user to concentrate on these cases, the
% amount of additional manual effort becomes
% proportional to the degree to which the program expoits the freedoms offered by dynamic typing.
\end{abstract}

\todo{use consequently in the whole paper: $\text{x}$ - program variable,
      $\text{@v}$ - instance variable, $x$ - logical variable, $\text{this}$ = program var,
      $\text{null}$ = expression (constant), $\Vnull$ - value, $S$ for statements, $T$ for union types, $\mathbb{T}$ for base types}
\todo{type-safe instead of typesafe}
\todo{is there a way to have PDF viewers see the tree-like structure of the document (sections, subsections, etc.)?}
\todo{some papers have much nicer citation-links \cite{Gardner12programlogicJS}}
\todo{think about: which results should be theorems, which ones lemmas?}
\todo{there are more...}

\section{Introduction}
  \label{sec:introduction}
  \input{sections/introduction}

\section{Overview / Motivation}
  \label{sec:overview}
  \input{sections/overview}

  \subsection{Static- vs. Dynamically-typed Hoare Logic}
  \label{sec:overview:stat-dyn-hoare-logic}
  \input{sections/overview_layer}
  % - Comparing Hoare Logics for Dynamically Typed- and Statically Typed Languages
  % - Layer of Abstraction allows SMT-solvers to solve Verification Conditions using efficient decision procedures for presburger/real arithmetic instead of reasoning about graph-like object structures.

  \subsection{Providing Type Information}
  \label{sec:overview:provide_type_info}
  \input{sections/overview_provide_type_info}

  \subsection{The Evaluator Example}
  \label{sec:overview:ev-example}
  \input{sections/overview_ev-example}

  \subsection{Semi-Automation}
  \label{sec:overview:semi-auto}
  \input{sections/overview_semi_automation}
  % - combining automation & completeness

\input{sections/overview_summing_up}

\section{Setting}

  \subsection{Model Languages: Dyn and Stat}
    \label{sec:setting:language}
    \input{sections/dyn_and_stat}

  \subsection{Hoare Logic}
    \label{sec:setting:hoare-logic}
    \label{sec:HL}  % old label
    \input{sections/hoare_logic.tex}

    % - sound and complete Hoare logic for dyn and stat is a reasonable assumption
  
% Core of the paper

\section{Layer of Abstraction}
  % Reducing Dynamically Typed Verification to Statically Typed Verification using Type Information
  \label{sec:layer}
  \input{sections/layer}
  
  \subsection{Type Safety Preconditions}
  \label{sec:layer:decomposition}
  \input{sections/layer_decomposition}

  \subsection{Mapping Objects to Values}
    \label{sec:layer:object_mapping}
    \input{sections/object_mapping}

  \subsection{Pure Expressions}
    \label{sec:layer:pure_expressions}
    \input{sections/pure_expressions}
\section{Deriving Type Information}
  \label{sec:deriv}
  
  \subsection{Type Information and Type Safety}
    \label{sec:deriv:type_inference_and_type_safety}
    \input{sections/type_inference_and_type_safety}

  \subsection{Type Safety Proofs are Type Information}
    \label{sec:deriv:extracting_type_information}
    \input{sections/extracting_type_information}

\section{Semi-Automation}
  \label{sec:auto}
  \input{sections/auto}

  \subsection{An Examplary Automatic Type Safety Verifier}
    \label{sec:auto:type_safety_verifier}
    \input{sections/example_type_inference}

  \subsection{Translating Typings into Proofs}
    \label{sec:auto:typings2proofs}
    \label{sec:TI2HL} % old label
    \input{sections/translating_typings_into_proofs}

  \subsection{Two-Layered Proofs}
    \label{sec:auto:two_layered_proofs}
    \label{sec:auto:proof_patches} % old label
    \label{sec:HL2TI} % old label
    \input{sections/two_layered_proofs}

\section{Verifying the Evaluator Example}
  \label{sec:verifying_ev-example}
  \input{sections/verifying_ev-example}

% \section{Evaluation}
% - implemented the example-TI with a proper refinement interface
% - manually proved type safety of X dynamically typed programs, translated to dyn
% - results in fig Y (hopefully) show that small, localized patches usually suffice to proof them typesafe
% \todo{add evaluation results}

\section{Related Work}
  \label{sec:related_work}
  \input{sections/related_work}

\section{Conclusion \& Future Work}
  \label{sec:conclusion}
  \input{sections/conclusion}

%%TODO{notes from NEF}
% Notes from NEF:
% x concurrency -> future work
% * use \mathit or \mathrm for text
% * remove to_ref()
% * Hoare gross
% * Úberschriften gross (am Anfang)
% * add missing labels
% * think about another letter for ``r''
%%TODO{notes from SL}

% The bibliography should be embedded for final submission.

% \begin{thebibliography}{}
% \softraggedright
% 
% \bibitem[Smith et~al.(2009)Smith, Jones]{smith02}
% P. Q. Smith, and X. Y. Jones. ...reference text...
% \end{thebibliography}

\todo{embed the bibliography}
\bibliographystyle{splncs03}
\bibliography{common}

\appendix
\section{Appendix: Identifying and Translating Pure Expressions}
  \label{app:pure_expressions}
  \input{appendices/app_pure_expressions}

\section{Appendix: Axiomatic Semantics for \textbf{dyn}}
  \label{app:hoare_logic}
  \input{appendices/app_hoare_logic}

%\section{Appendix: Type-level Hoare logic for \textbf{dyn}}
%  \label{app:type_hoare_logic}
%  \input{appendices/app_type_hoare_logic}

\section{Appendix: Automatic Type Safety Verifier $\VEx$}
  \label{app:verifier_VEx}
  \input{appendices/app_VEx}

\section{Appendix: Translation from Typings to Typing Proofs}
  \label{app:typing_to_proof}

\input{appendices/app_typing_to_proof}

\section{Appendix: Omitted Proofs}
  \label{app:proofs}
  \input{appendices/app_proofs}

\end{document}

%Questions
%- why no constants of object type in [253]?
%- does the mapping of constructors to new + init call cause completeness problems?
%- transformation of method calls:
%  - why skip as default?

%References:
% see bibliography.notes

%                       Revision History
%                       -------- -------
%  Date         Person  Ver.    Change
%  ----         ------  ----    ------

%  2013.06.29   TU      0.1--4  comments on permission/copyright notices

%% file: defs.tex
\def\true{TRUE}

\def\bottom{\bot}
\newcommand\nef[1]{\todo[color=cyan]{#1}}

\def\extended{EXTENDED}

\newcommand{\myparagraph}[1]{\newline\noindent\textbf{#1}}
\newcommand{\myparagraphB}[1]{\noindent\textbf{#1}}

\newcommand*\externallink[1]{\href{{#1}}{\includegraphics[scale=0.7]{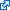}}}

\newcommand{\quantsep}{.\;}

\newcommand*\constraint{\sqsubseteq}
\newcommand*\constraintequal{=}
\newcommand{\constraintunion}{\sqcup}

\newenvironment{corners}{$\operatornamewithlimits{}_\llcorner^{\ulcorner}$}{$\operatornamewithlimits{}_\lrcorner^{\urcorner}$}

\newcommand{\precision}[1]{{\color{blue}#1}}

\newcommand{\mStrongPC}[1]{%
 \begin{tikzpicture}[baseline=(text.base)]%
   \node [draw,dashed,rounded corners,fill=red!20] (text) {$#1$};%
 \end{tikzpicture}}

\newcommand{\mTypesafePC}[1]{%
 \begin{tikzpicture}[baseline=(text.base)]%
   \node [draw,solid,rounded corners,fill=blue!20] (text) {$#1$};%
 \end{tikzpicture}}

\newcommand{\StrongPC}[1]{%
  \tikz[baseline=(text.base)]
    {\node [draw,dashed,rounded corners,fill=red!20] (text) {#1};}%
}
\newcommand{\TypesafePC}[1]{%
  \tikz[baseline=(text.base)]
    {\node [draw,solid,rounded corners,fill=blue!20] (text) {#1};}%
}

\newcommand{\mmybox}[1]{\tikzpicture[baseline=(text.base)]{\node [draw,rounded corners,fill=red!20] (text) {hihi};}}

\newcommand{\oC}{\rho}

\newcommand{\Var}{\mathcal{V}}
\newcommand{\Class}{\mathcal{C}}
\newcommand{\Method}{\mathcal{M}}
\newcommand{\Domain}{\mathcal{D}}
\newcommand{\Type}{\mathcal{T}}

\newcommand{\result}{\mathbf{r}}

\newcommand{\pCm}{\text{\u{p}}}
\newcommand{\qCm}{\text{\u{q}}}

\newcommand{\restricted}{\theta}

\newcommand{\deriv}{\vdash}
\newcommand{\free}{\mathit{free}}

\newcommand{\Vnull}{\mathit{null}}
\newcommand{\this}{\PL{this}}

\newcommand{\TInv}{\mathcal{I}}
\newcommand{\PL}[1]{\text{#1}}
\newcommand{\textPL}[1]{\textsl{#1}}

\newcommand{\Vector}[1]{\overrightarrow{#1}}

\newcommand{\pure}{\varepsilon}
\newcommand{\ppure}{\mathit{pure}}
\newcommand{\defined}{\stackrel{\Delta}{=}}
\newcommand{\safe}{\mathit{safe}}

\newcommand{\smin}{\typingX^{\dagger}}

\newcommand{\pre}[1]{{}^\downarrow #1}
\newcommand{\post}[1]{#1^\downarrow}
\newcommand{\Fversion}{\upsilon}
\newcommand{\typingassertion}{\Xi}
\newcommand{\translation}{\typingassertion}  % the same
\newcommand{\assertS}{S_{\mathit{TA}}}

\newcommand{\Fmethod}{\Method}
\newcommand{\Ex}[1]{#1_{\mathit{Ex}}}
\newcommand{\Verifier}{\mathcal{V}}
\newcommand{\VEx}{\Ex{\Verifier}}
\newcommand{\filter}{\mathit{pr}}
\newcommand{\projection}{\filter}

\newcommand{\typingX}{\mathit{ty}}
\newcommand{\typingY}{\mathit{ty'}}
\newcommand{\StmtX}{S}
\newcommand{\ExprX}{e}
\newcommand{\programX}{\pi}
\newcommand{\ClassX}{\PL{C}}
\newcommand{\MethodX}{\PL{m}}
\newcommand{\objectX}{o}
\newcommand{\VarLX}{\PL{u}}
\newcommand{\VarLXVirtual}{\hat{\VarLX}}
\newcommand{\VarIX}{\PL{@x}}

\newcommand{\LVarX}{v}
\newcommand{\VarX}{\PL{x}}
\newcommand{\VarXVirtual}{\hat{\VarX}}
\newcommand{\ProgVirtual}[1]{\widehat{#1}}
\newcommand{\UnionTypeX}{T}
\newcommand{\BaseTypeX}{\mathbb{T}}
\newcommand{\proofX}{\psi}

\newcommand{\locX}{L}

\newcommand{\Ntypeerror}{\tau}
\newcommand{\Ndivergence}{t}
\newcommand{\Nfailure}{f}
\newcommand{\Npost}{p}

\newcommand{\AsrtX}{p}
\newcommand{\AsrtY}{q}
\newcommand{\LExpX}{l}
\newcommand{\TAsrtX}{\tau}
\newcommand{\TAsrtY}{\tau'}
\newcommand{\TAsrtKX}{\nu}
\newcommand{\TLitX}{\mu}

\newcommand{\typevarIvarX}{\llbracket \ClassX.\VarIX \rrbracket}

\newcommand{\HLs}{$\mathit{HL}_s\;$}
\newcommand{\HLd}{$\mathit{HL}_d\;$}

%% file: sections/introduction.tex
\todo{what are the special requirements of ECOOP?}

Dynamically typed programming languages refrain from restricting their programs to ensure operations are only applied to suitable operands. While this allows experienced programmers to write more elegant, concise and reusable code, it has the obvious drawback that type errors may occur at runtime.

Recently, object-oriented dynamically typed languages like Python, Ruby and JavaScript are gaining popularity
also on the server-side (Ruby on Rails, node.js) and are used even for business- \cite{Nexedi} and safety-critical \cite{Frequentis} applications.

\todo{mention hybrid-typed languages}

Unfortunately, despite the growing need for correctness guarantees, the lack of type information causes a large overhead in formal methods like Hoare logic and severely decreases the effectiveness of automatic reasoning engines compared to the statically-typed setting (see Section~\ref{sec:overview:stat-dyn-hoare-logic}).

There are two ways to deal with this problem:

1) \textbf{Annotation}: Most contemporary approaches to verifying dynamically typed programs ask the user to manually supply the needed type information in loop invariants and method contracts \cite{Gardner12programlogicJS,QinEtAl2011,Swamy2013DijMonad,NguyenTH2013SoftContract}.
%\todo{cite other exemplary approaches, DJS?}
For larger programs, this induces significant overhead. We argue that manually supplying type information for all variables is not only tedious, but also often unnecessary, as most of this information could have been inferred automatically.

2) \textbf{Translation}: Obviously, translating the dynamically typed program into an equivalent statically typed version\footnote{After this translation, the static type system should be able to ensure the absence of type errors, unlike in the embeddings discussed in \cite{Harper2012}. Finding such an equivalent version is undecidable in general and hence requires manual effort (see Section~\ref{sec:overview:provide_type_info})} and then using a Hoare logic for statically-typed programs (like \cite{AptDeBoerOlderog2012,DeBoerPierik2003}) for verification is also possible. In such a translation process, type inference algorithms like \cite{JensenMollerThiemann2009,Furr2009StaticTypeInferenceForRuby} are usually of significant help. Note, however, that gradual typing \cite{SiekTaha2007,BiermanAbadi2014UTypeScript} it not useful in this context, as such Hoare logics require the the entire program to be well-typed prior to verification. Additionally, this approach removes any benefit of dynamic typing since it is equivalent to verifying a statically typed language with type inference.

We propose to get the best of both worlds by integrating an automatic type safety verifier with Hoare logic into a semi-automatic procedure and using the derived type information to reduce overhead and enable effective automated reasoning about dynamically typed programs just like with statically typed ones. In the context of soft typing \cite{Cartwright91softTyping}, our approach can also be understood as offering proofs of type safety as a third option besides rewriting the program (static typing) and runtime-checks (dynamic typing).

Concretely, in this paper we describe two components:

1) A layer of abstraction that, given suitable type information, abstracts from the complexities of dynamic typing and hence reduces the verification of dynamically typed programs to that of statically typed ones. This also works with partial type information on a per-expression basis (see Section~\ref{sec:overview:stat-dyn-hoare-logic}).

2) A construction for complementing a Hoare logic with an automatic type safety verifier, yielding a semi-automatic procedure for deriving type information with the following properties (see Section~\ref{sec:overview:semi-auto}):
\begin{itemize}
 \item Automation -- only typing problems beyond the reach of the automatic verifier require manual intervention.
 \item Completeness relative to the Hoare logic -- if the Hoare logic is complete, then type information can be derived for all typesafe programs (see Section~\ref{sec:deriv:extracting_type_information}).
 \item Bidirectional exchange of results -- automatically derived type information can be used in Hoare logic proofs and vice versa, proof results are used by the automatic verifier to increase precision.
\end{itemize}

Together, these two components form a novel verification system that makes the effort additionally required to verify a dynamically typed program proportional to the total complexity \nef{amount?} of hard typing problems in this program. Unlike in annotation-based-approaches, programs with only trivial typing problems require no additional effort and unlike in translation-based-approaches, all typesafe programs can be verified.
%For this system, the following equation hence holds for all programs and can also be read effort-wise:
%
%\[ \text{dynamically typed correctness} = \text{statically typed correctness} + \text{type safety} \]

This paper constitutes our first step towards connecting the (relative) complete Hoare logics \cite{AptDeBoerOlderog2012,DeBoerPierik2003} and advanced reasoning engines \todo{cite Boogie? look up references of Soft-Contract-Verification paper} developed for statically-typed object-oriented languages with the advancing automatic type safety verifiers for dynamically typed languages \cite{NguyenTH2013SoftContract,Swamy2013DijMonad,Chugh2012NestedRef,JensenMollerThiemann2009}.
\ifx \version\extended In this extended version, proofs for all theorems and lemmas can be found in Appendix~\ref{app:proofs}. 

\myparagraphB{Notation}
$\Vector{p}$ is a sequence $p_1,...,p_n$ where $n$ is obvious from context or does not matter, $\{\Vector{p}\}$ the smallest set containing all its elements and $\Vector{a} \Vector{b}$ sequence concatenation. $\defined$ means ``is defined as'' and $\mathbb{N}_n \defined \{0,...,n\}$, $\mathbb{N}^1_n \defined \{1,...,n\}$.

%% file: sections/overview.tex
We will first discuss how correctness proofs can be simplified using sufficient type information and then how this information can be derived.

% We divide the problem into two parts: Generating sufficient type information and using it to simplify the correctness proof.  Our discussion will start with the latter.
% 
% The problem has two parts:  and using it to .  Our discussion will start with the latter.
% 
% We will first discuss simplifying the correctness proof using sufficient type information and then how to generate this information

%% file: sections/overview_layer.tex
Apart from the additional need to establish type safety, there are other differences between Hoare logic for dynamically typed- and statically typed languages (\HLd and \HLs). The latter (like \cite{AptDeBoerOlderog2012,DeBoerPierik2003}), usually share a type system between programming- and assertion language: the assertion $\PL{x} > 8$ denotes the set of states where the value of a numeric program variable $\PL{x}$ is larger than 8.
In \HLd (like \cite{Gardner12programlogicJS}) however, as types are not statically known, all variables are of type $\mathbb{O}$ (object). The assertion $\PL{x} > 8$ is hence meaningless as $>$ is not defined for type $\mathbb{O}$. In this setting, a similar set of states can be denoted by the assertion\footnote{the precise meaning of $\mathbb{N}(\PL{x},i)$ will be explained in Subsection~\ref{sec:layer:object_mapping}.} $\exists i \quantsep \mathbb{N}(\PL{x},i) \wedge i > 8$ which can be automatically derived from $\PL{x} > 8$ given sufficient type information (the fact that the object referenced by \PL{x} always represents a number).

Furthermore, \HLs usually include side-effect-free (pure) program expressions ($\ExprX$) into the assertion language, allowing efficient reasoning using proof rules like
\[ \{\AsrtY[\VarLX := \ExprX]\} \VarLX := \ExprX \{\AsrtY\} \]
Here, $q[\VarLX := e]$ denotes the substitution of all occurrences of a variable $\VarLX$ by $e$ in the assertion $q$. This rule allows directly deducing weakest preconditions over assignments like $\{\PL{x} + 5 > 8\} \PL{x} := \PL{x} + 5 \{\PL{x} > 8\} \;\; (1)$ by letting the expression $e$ traverse the boundary between program and logic. In \HLd, this is not possible since program expressions could have side-effects. While a subset of side-effect-free methods can be defined, identifying such pure expressions requires type information. Without it, establishing a property equivalent to $(1)$ requires $\ge 6$ rule applications.

This observation is given significance by the fact that usually most expressions only involve immutable data types like numbers and strings. Regarding them as side-effecting operations on general object-structures not only complicates proofs, but also significantly decreases the effectiveness of automated reasoning engines \todo{cite some examples}. For instance, assertions can often be efficiently established by SMT solvers over Presburger arithmetic while no similar decision procedure exists for arbitrary operations over general object-structures.
\todo{paragraph can be made simpler, but loses elegance}

% Against this background, the need for type information to simplify both the assertions by removing
% unnecessary type-mapping-predicates and the proofs by allowing effective identification of pure expressions seems rather obvious.
Section~\ref{sec:layer} will show how type information can be used to counter these problems and create the illusion of proving a statically typed program.

% creates the illusion to actually prove a statically typed program by automatically hiding mapping predicates, identifying pure expressions and removing type-safety preconditions from the rules for method calls, conditionals and while-loops, which also require significant effort to establish.

%% file: sections/overview_provide_type_info.tex
% The previous section demonstrated the importance of type information for verification. Next, let us discuss how to derive that information for dynamically typed programs.

% The problem of type safety  finding consistent type information for a dynamically typed program (that is, type information that is sound, thus contains all possible runtime types for each expression and still precise enough to establish the absence of runtime type errors) is undecidable in general.
Sufficient type information for dynamically typed programs is uncomputable in general (Section~\ref{sec:deriv:type_inference_and_type_safety}). However, a number of good approximations exist \cite{JensenMollerThiemann2009,Furr2009StaticTypeInferenceForRuby} \todo{cite some more TI algos} that we will refer to as \emph{automatic type safety verifiers}.

It is known \todo{find citation} that many dynamically typed programs only occasionally diverge from what would also be possible in static typing disciplines\footnote{Advanced dynamic features like mixins, traits, method update and dynamic class hierarchies increase the complexity of type inference. However, in this paper we aim to study the problem of dynamic typing in isolation and leave them as future work.} and consequently, that 
the output of such algorithms is usually sufficient for typing most of their subexpressions \cite[Section~5]{JensenMollerThiemann2009}\cite[Section~6]{Furr2009StaticTypeInferenceForRuby}.

% usually only a small percentage of typing problems in a given program (subexpressions used as receivers or conditions) require precision beyond what automatic algorithms can currently offer .
%
\todo{cite Gradual Typing for First-Class Classes in footnote?}
If the entire program can be typed by a sound automatic verifier, then \HLs could be applied. However, the whole point of dynamic typing is the possibility to go beyond the limits of such automatic procedures (type systems). Approaches to verifying these languages
%(in contrast to type-inferred statically typed languages)
thus must also be able to operate under less ideal circumstances. The following example will illustrate this point.

%% file: sections/overview_ev-example.tex
\todo{create figure of data-structure}
\begin{figure}
\verb|class Evaluator {|

\verb|  method parse(str) { ... }|

\verb||

\verb|  method calc(env, tree = @tree) {|

\verb|    if tree[0] = VALUE then tree[1]|

\verb|    elseif tree[0] = VAR then env[tree[1]]|

\verb|    elseif tree[0] = OP then|

\verb|      if tree[1] = ADD then calc(env, tree[2]) + calc(env, tree[3])|

\verb|      elseif ...|

\verb|      else nil|

\verb|      fi|

\verb|    else nil|

\verb|    fi|

\verb|  }|

\verb|}|

\verb|new Evaluator().parse(input).calc(ENV)|
  \caption{\label{fig:ev-example} Relevant part of the evaluator example source code}
\end{figure}

Figure~\ref{fig:ev-example} depicts a dynamically typed program evaluating arithmetic expressions. While crafted to provide a hard typing problem, its use of ad-hoc data structures is not uncommon in Ruby, Python or Javascript.

The class \verb|Evaluator| has two methods \verb|parse()| and \verb|calc()|. The former parses a string and stores the resulting parse tree in the instance variable \verb|@tree|, while the latter evaluates a given parse tree (defaulting to \verb|@tree|) over a given environment (a mapping from variable names (strings) to integers).

The example is hard to type because the parse trees are represented as ad-hoc constructions of nested lists. Numeric constants \verb|VALUE|, \verb|VAR| and \verb|OP| in the first element distinguish value-, variable- and operation nodes. The types of the remaining list elements depend on these node types: the second element is numeric (the value) for value-nodes, a string (the variable name to be looked up in the environment) for var-nodes and numeric (representing the operation to be performed) for op-nodes. Only op-nodes use nesting: further list elements are sub-parse-trees that are to be recursively evaluated to operands.

Typing this example requires deducing precise types for heterogeneous lists from propositions (like \verb|tree[0] = VALUE|) about their first element. To the best of our knowledge there is no automatic procedure able to establish such implications. Also note that the typing problem can be made even harder: allowing an arbitrary number of operands in op-nodes, returning strings instead of \verb|null|, etc. This example will be used to demonstrate our technique.

%% file: sections/overview_semi_automation.tex
%Figure \todo{create Figure} depicts the general concept for semi-automatically deriving type information.
%
\begin{figure}
  {\def\svgwidth{350px}
    \input{images/img-concept1}
  }
\caption{\label{fig:concept} Overview of the concept}
\end{figure}
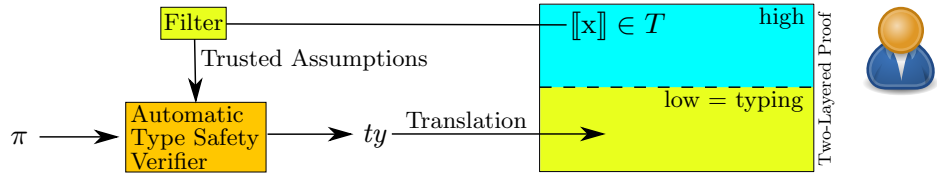
\todo{use less shiny colors}
In the concept depicted in Figure~\ref{fig:concept}, the correctness proof is split into two ``layers'' (see Section~\ref{sec:auto:two_layered_proofs}). While the user (supported by a theorem prover) derives his proof in the higher layer, the lower layer contains type information and is created and modified solely by the automatic type safety verifier. For this purpose, the typings ($\typingX$) derived for the program $\programX$ by the verifier are translated into proofs (see Section~\ref{sec:auto:typings2proofs}). While the information contained in this lower layer proof is already useful for supporting the user's higher-layer proof (see Section~\ref{sec:layer}), the user may at any time decide to refine it by deriving more precise type information in the higher layer. This information is filtered to make it interpretable for the verifier and then supplied as trusted assumptions to refine the lower-layer type information (see Sections~\ref{sec:auto:type_safety_verifier}, \ref{sec:auto:two_layered_proofs}).

%% file: images/img-concept1.tex
%% Creator: Inkscape inkscape 0.48.4, www.inkscape.org
%% Accompanies image file 'concept1.pdf' (pdf, eps, ps)
\begingroup%
  \makeatletter%
  \providecommand\color[2][]{%
    \errmessage{(Inkscape) Color is used for the text in Inkscape, but the package 'color.sty' is not loaded}%
    \renewcommand\color[2][]{}%
  }%
  \providecommand\transparent[1]{%
    \errmessage{(Inkscape) Transparency is used (non-zero) for the text in Inkscape, but the package 'transparent.sty' is not loaded}%
    \renewcommand\transparent[1]{}%
  }%
  \providecommand\rotatebox[2]{#2}%
  \ifx\svgwidth\undefined%
    \setlength{\unitlength}{1007.39941406bp}%
    \ifx\svgscale\undefined%
      \relax%
    \else%
      \setlength{\unitlength}{\unitlength * \real{\svgscale}}%
    \fi%
  \else%
    \setlength{\unitlength}{\svgwidth}%
  \fi%
  \global\let\svgwidth\undefined%
  \global\let\svgscale\undefined%
  \makeatother%
  \begin{picture}(1,0.18189097)%
    \put(0,0){\includegraphics[width=\unitlength]{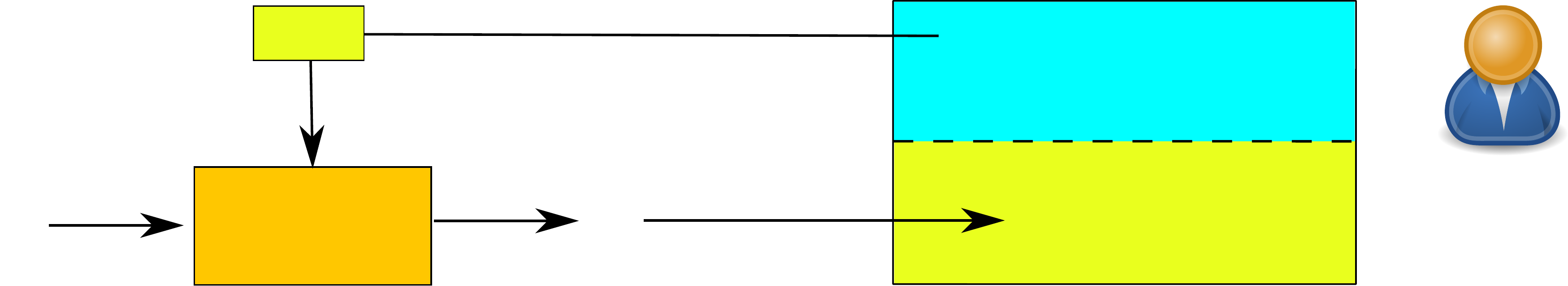}}%
    \put(0.12891684,0.05412188){\color[rgb]{0,0,0}\makebox(0,0)[lb]{\smash{Automatic}}}%
    \put(0.12891684,0.02912188){\color[rgb]{0,0,0}\makebox(0,0)[lb]{\smash{Type Safety}}}%
    \put(0.12891684,0.00412188){\color[rgb]{0,0,0}\makebox(0,0)[lb]{\smash{Verifier}}}%
    \put(0.0,0.03053823){\color[rgb]{0,0,0}\makebox(0,0)[lb]{\smash{\scalebox{1.2}{$\programX$}}}}%
    \put(0.37823033,0.03319948){\color[rgb]{0,0,0}\makebox(0,0)[lb]{\smash{\scalebox{1.2}{$\typingX$}}}}%
    \put(0.42500627,0.04750236){\color[rgb]{0,0,0}\makebox(0,0)[lb]{\smash{Translation}}}%
    \put(0.20441883,0.11503059){\color[rgb]{0,0,0}\rotatebox{0.04453833}{\makebox(0,0)[lb]{\smash{Trusted Assumptions}}}}%
    \put(0.60161422,0.15054116){\color[rgb]{0,0,0}\rotatebox{-0.0016154}{\makebox(0,0)[lb]{\smash{\scalebox{1.2}{$\llbracket \VarX \rrbracket \in \UnionTypeX$}}}}}%
    \put(0.70366368,0.0702944){\color[rgb]{0,0,0}\rotatebox{-0.0016154}{\makebox(0,0)[lb]{\smash{low = typing}}}}%
    \put(0.80562554,0.1604097){\color[rgb]{0,0,0}\rotatebox{-0.0016154}{\makebox(0,0)[lb]{\smash{high}}}}%
    \put(0.88301896,0.00682917){\color[rgb]{0,0,0}\rotatebox{90.00395054}{\makebox(0,0)[lb]{\smash{\scalebox{0.75}{Two-Layered Proof}}}}}%
    \put(0.16499889,0.15388489){\color[rgb]{0,0,0}\makebox(0,0)[lb]{\smash{Filter}}}%
  \end{picture}%
\endgroup%

%% file: sections/overview_summing_up.tex
Note that deriving type information and using the layer of abstraction is not a strict 2-step process. The latter requires the former only on a per-expression basis, allowing an interleaving of steps. Concretely, the layer of abstraction applies to all expressions proven type-safe (see Section~\ref{sec:layer}). Expressions with open typing problems may be included at any time by proving them type-safe. This interleaving is possible as our refinements are monotonic (see Section~\ref{sec:auto:two_layered_proofs}).

% 
% Concretely, given a dynamically typed program containing some number of independent hard typing problems, say $n$, then
% our approach allows directly verifying it like a statically typed one, only resorting to dynamically typed proof rules for the $n$ type-unsafe expressions (see Section~\ref{sec:layer}). Also, any one of these expressions may be included into the layer of abstraction at any time by proving it type-safe.
% 

% \begin{itemize}
% %  \item Statically typed Hoare logic would not be directly applicable as its well-typedness assumption is violated.
%  \item Translation-based approaches would first have to find an equivalent, statically typeable version of the program, before verification can start.
%  \item Annotation-based approaches would be directly applicable, but require all type information to be specified manually.
% \end{itemize}

%% file: sections/dyn_and_stat.tex
%\section{Dynamically Typed Languages}

%\subsection{Dynamically Typed Model Language: dyn}

To explain our methodology in a setting facilitating formal proof we introduce a pair of minimalistic programming languages that differ only in the fact that one is dynamically typed (\textbf{dyn}) while the other uses a static type system with type inference (\textbf{stat}). Like their real-world siblings, the two are imperative, class-based object-oriented languages including inheritance, method renaming, dynamic dispatch and constructors. However, they do not support advanced dynamic features like a dynamic class hierarchy, method update or eval().
\begin{figure}
\underline{Syntax of \textbf{dyn}:}

$\mathit{Prog_d} \ni \pi ::= \Vector{\mathit{class}} \PL{ } \StmtX$ \hfill $\PL{u} \in \Var_L, \VarIX \in \Var_I, \PL{C} \in \Class, \PL{m} \in \Method$

$\mathit{Class_d} \ni \mathit{class} ::= \PL{class C $<$ C }\{ \Vector{\mathit{meth}} \}$

$\mathit{Meth_d} \ni \mathit{meth} ::= \PL{method m}(\Vector{\PL{u}}) \{ \StmtX \} \mid \PL{rename m m}$, \hfill $\mathit{Stmt_d} \ni \StmtX ::= \StmtX;\StmtX \mid e$

$\mathit{Expr_d} \ni \ExprX ::= \PL{null} \mid \PL{u} \mid \VarIX \mid \this{} \mid \ExprX == \ExprX \mid \ExprX \PL{ is\_a? \ClassX} \mid \ExprX.\PL{m}(\Vector{\ExprX}) \mid \PL{new }\ClassX(\Vector{\ExprX})$

\hspace*{2cm}$\mid \PL{u} := \ExprX \mid \VarIX := \ExprX \mid \PL{if } \ExprX \PL{ then } \StmtX \PL{ else } \StmtX \PL{ fi}\mid \PL{while } \ExprX \PL{ do } \StmtX \PL{ od}$

\underline{Syntax of \textbf{stat}:} -- coincides with \textbf{dyn}, except for

$\mathit{Stmt_s} \ni \StmtX ::= \StmtX;\StmtX \mid \ExprX_{\pure} \mid \PL{u} := \ExprX_{\pure}.\PL{m}(\Vector{\ExprX_{\pure}}) \mid \PL{u} := \PL{new C}(\Vector{\ExprX_{\pure}}) \mid \PL{u} := \ExprX_{\pure} \mid  \VarIX := \ExprX_{\pure}$

\hspace*{2cm}$\mid \PL{if } \ExprX_{\pure} \PL{ then } \StmtX \PL{ else } \StmtX \PL{ fi}\mid \PL{while } \ExprX_{\pure} \PL{ do } \StmtX \PL{ od}$

$\mathit{Expr_s} \ni \ExprX_{\pure} ::= \PL{null} \mid \PL{u} \mid \VarIX \mid \this{} \mid \ExprX_{\pure} == \ExprX_{\pure}$ \hfill $c_{\pure} \in \mathit{Cnst_s}, \mathit{op}(\Vector{\ExprX_{\pure}}) \in \mathit{Op_s}$

\hspace*{2cm}$\mid \ExprX_{\pure} \PL{ is\_a? C} \mid c_{\pure} \mid \mathit{op}(\Vector{\ExprX_{\pure}}) \mid \PL{if } \ExprX_{\pure} \PL{ then } \ExprX_{\pure} \PL{ else } \ExprX_{\pure} \PL{ fi}$

\begin{minipage}{0.55\textwidth}
\underline{Syntactic sugar in \textbf{dyn}:} 

$\ExprX_1 \oplus \ExprX_2 \defined \ExprX_1.\MethodX_{\oplus}(\ExprX_2)$

$\PL{if } \ExprX \PL{ then } \StmtX \PL{ fi} \defined \PL{if } \ExprX \PL{ then } \StmtX \PL{ else null fi}$

$\mathit{false} \defined \PL{new bool(null)}$, $\mathit{true} \defined \mathit{false}.\PL{not()}$

$0 \equiv \PL{new num(null)}$, $n \defined \PL{$(n-1)$.succ()}$

% $"" \equiv \PL{new string(null, null)}$,
% 
% $\PL{"...a"} \equiv \PL{"..."}.\PL{addchar}(n_{\PL{a}})$
% where $n_{\PL{a}} \in \mathbb{N}$ is the ASCII-code of character $\PL{a}$.
\end{minipage}
\begin{minipage}{0.45\textwidth}
\underline{Basic data types in \textbf{stat}:}

$\mathit{true}, \mathit{false} \in \mathit{Cnst_s}$, $\neg: \mathbb{B} \mapsto \mathbb{B} \in \mathit{Op_s}$

$\wedge, \vee, \rightarrow: \mathbb{B} \times \mathbb{B} \mapsto \mathbb{B}  \in \mathit{Op_s}$

$0,1,... \in \mathit{Cnst_s}$, $=: \mathbb{N} \times \mathbb{N} \mapsto \mathbb{B} \in \mathit{Op_s}$

$+,*,div: \mathbb{N} \times \mathbb{N} \mapsto \mathbb{N} \in \mathit{Op_s}$

%$"", "a", ... \in \mathit{Cnst_s}$, $+: \mathbb{S} \times \mathbb{S} \mapsto \mathbb{S} \in \mathit{Op_s}$
%
% $s[n]: \mathbb{S} \times \mathbb{N} \mapsto \mathbb{N} \in \mathit{Op_s}$
% 
% $|s|: \mathbb{S} \mapsto \mathbb{N} \in \mathit{Op_s}$

\end{minipage}
%
% \underline{Type lattices of \textbf{dyn} (left) and \textbf{stat} (right):}
% 
%  \input{images/typesystem-dyn}
%  \input{images/typesystem-stat}
% \label{fig:dyn_stat_typesystem}  % old label

\caption{\label{fig:lang_syntax}Syntax of \textbf{dyn} and \textbf{stat}}
\end{figure}
\todo{$e==e$ and $e \PL{ is\_a? C}$ were added. Also add to semantics, etc.}
\todo{extend is\_a? - syntax to allow checking for union types?}

\myparagraphB{Syntax:}
The syntax of both \textbf{dyn} and \textbf{stat} is depicted in Figure~\ref{fig:lang_syntax}.
In \textbf{dyn}, method bodies consist of statements ($\StmtX$) which in contrast to expressions ($\ExprX$)
can contain sequential composition.
Expressions are composed of \textPL{null}, the only constant, local- and instance variables (prefixed with $\PL{@}$), the self-reference \textPL{this}, operators for object identity and dynamic type checks,
method- and constructor calls, assignments, conditionals and while loops. Note that equality ($\PL{=}$) is desugared to a (class-specific) method call, while object identity ($\PL{==}$) is a build-in operation yielding $\mathit{true}$ iff the two expressions refer to the same object (We stipulate $\PL{null == null}$ yields $\mathit{true}$).

Each class except the predefined class \textPL{object} must specify a parent class whose methods are inherited. The inheritance relation must be acyclic. Every class thus transitively inherits from \textPL{object}. Inherited methods may be overwritten or renamed (using \textPL{rename}). Like in actual dynamically typed languages, inheritance is mere code reuse and can be removed using an automatic expansion step \cite{PalsbergS91TypeInfernceOO}. Furthermore, we will assume this step to be completed and not concern ourselves any further with inheritance or renaming.

% Each class $C$ has a domain $\mathcal{D}_C$ disjoint from all others and $\Vnull$ is the only element of
% $\mathcal{D}_{Nullclass}$%TODO{is this compatible with our translation?}. The entire value domain of
% type object is thus given by
% $\mathcal{D} = \bigcup_{C \in \Class} \mathcal{D}_C$.
% We use the function $Class: \mathcal{D} \mapsto \Class$ to distinguish values of different classes
% \[Class(v) = C \text{ iff } v \in \mathcal{D}_C.\]

\myparagraphB{Semantics:}
\label{sec:dyn_semantics}
%
% For a given program, we denote the set of all methods as $\Method$, the set of all classes as $\Class$ and the set of all variables as $\Var = \Var_L \uplus \Var_I$ where $\Var_L$ and $\Var_I$ are the sets of local- and instance variables respectively.
%
Both \textbf{dyn} and \textbf{stat} programs consist of a main statement $\StmtX$ and sets of classes $\Class$, methods $\Method$ and variables $\Var = \Var_L \uplus \Var_I$ where $\Var_L$ and $\Var_I$ are the sets of local- and instance variables respectively. While each class $\ClassX \in \Class$ has a subset of method declarations $\Method_C \subseteq \Method$ and instance variables $\Var_C \subseteq \Var_I$, every method $\ClassX.\MethodX \in \Method$ has a subset of local variables $\Var_{\ClassX.\MethodX} \subseteq \Var_L$ used in its method body $S_{\ClassX.\MethodX}$.
$\Var_S = \{this, \result\} \subset \Var_L$ is a set of special variables. While $this$ references the current object and is not allowed to be assigned to in programs, $\result$ holds the result of the last evaluated expression and cannot be used in programs.

%define S_{PL{C.m}}
%define \Var_\PL{C.m}
%define \Method(L)?

% for $C \in \Class$. $\Var_C$ includes all instance variables used in methods of $C$.

\textbf{Dyn}'s value domain is the set of all objects $\Domain_{d} = \Domain_{\mathbb{O}}$ and its type system is the lattice of union types represented as sets of class names $\{\ClassX_1,...,\ClassX_n\} \in \Type_d = 2^{\Class}$ with the subset-ordering $\subseteq$ (see Figure~\ref{fig:dyn_stat_typesystem}). The \textPL{null} value is contained in every such type. \textbf{Stat} on the contrary distinguishes basic data types $\Type_s = \{\mathbb{O}, \mathbb{N}, \mathbb{B}, \mathbb{S}, \mathbb{L}, \mathbb{M}, ...\}$ and its value domain $\Domain_{s} \cong \biguplus_{\BaseTypeX}\Domain_{\BaseTypeX}$
% where $\Domain_{\mathbb{L}} = \Domain_{\mathbb{N}} \mapsto \Domain_{\mathbf{stat}}$ and $\Domain_{\mathbb{M}} = \Domain_{\mathbf{stat}} \mapsto \Domain_{\mathbf{stat}}$. \textbf{Stat} values thus
includes objects, numbers, booleans, strings, lists and finite maps. \todo{how do statically typed maps handle the event of an unknown key?}\todo{how do statically types lists handle index out of bounds?}
We omit definitions of states, state update etc. as they are standard. To keep track of instance-class relationships we use class references and for every class $\ClassX$ introduce a distinct object $\oC_{\ClassX}$ as well as a special instance variable $\PL{@}\mathbf{c}$ such that $\objectX.\PL{@}\mathbf{c} = \oC_{\ClassX}$ iff $\objectX$ is an instance of class $\ClassX$. Using $\PL{@}\mathbf{c}$ in programs is not permitted.
\begin{figure}
 \input{images/img-typesystem-dyn}
 \input{images/img-typesystem-stat}
\caption{\label{fig:dyn_stat_typesystem}Type lattices of \textbf{dyn} (left) and \textbf{stat} (right)}
\end{figure}
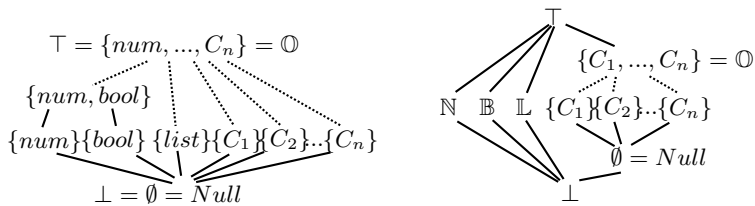
\myparagraph{Comparing Dyn with Stat:}
\textbf{Dyn} is a pure object-oriented language (objects are the only values) while \textbf{stat} has basic data types. However, both provide the same constants and pure (i.e. side-effect-free) operations on them. Dyn desugars them to constructor and method calls (see Figure~\ref{fig:lang_syntax}), while \textbf{stat} (like usual in statically typed languages) provides them build-in ($c_{\pure}$ and $\mathit{op}(\Vector{\ExprX_{\pure}})$ in Figure~\ref{fig:lang_syntax}).

Also, \textbf{stat} expressions are pure. Side-effects are only allowed in statements, which must only have pure subexpressions. This is not a restriction, as every \textbf{dyn}-expression can be transformed into a sequence of \textbf{stat} statements by recursively (and in the order of evaluation) replacing subexpressions $\ExprX$ by fresh local variables $\VarLX$ and prepending the assignment $\VarLX := \ExprX;$.

% Hoare logic rules for languages with side-effecting expressions (like the ones
%      in Section~\ref{sec:HL}) in essence emulate this prepending. While the rules look a bit different,
%      this is merely a cosmetic change and causes no problems during verification.

Every \textbf{stat} program is also a \textbf{dyn} program that evaluates to (an object-oriented version of) the same result. The only reason that the opposite direction does not hold is the language restriction imposed by \textbf{stat}'s static type system.

\myparagraphB{Type Errors:}
Contrary to \textbf{stat}, which rejects programs deemed unsafe at compile time, \textbf{dyn} allows every syntactically correct program to be executed and raises type errors at runtime when
\begin{itemize}
 \item a method call is not supported by its receiver (in this arity) or
 \item a condition of a conditional or while loop is not boolean
\end{itemize}
While ``message not understood''-errors are fundamentally linked to type-checking in class-based OO-languages, dynamically typed languages often allow conditions to be of arbitrary type. Nevertheless, the second error condition models a common error class where a built-in operation supports a fixed set of types.

Many dynamically typed languages raise type errors when accessing variables prior to assignment. We will leave this as future work and consider all local (instance) variables to be initialized to $\Vnull$ prior to method executions (on instantiation). Also, type errors are often treated as exceptions, allowing interception and handling. For simplicity, we will consider them as fatal.

%% file: images/img-typesystem-dyn.tex
%% Creator: Inkscape inkscape 0.48.5, www.inkscape.org
%% Accompanies image file 'typesystem-dyn.pdf' (pdf, eps, ps)
\begingroup%
  \makeatletter%
  \providecommand\color[2][]{%
    \errmessage{(Inkscape) Color is used for the text in Inkscape, but the package 'color.sty' is not loaded}%
    \renewcommand\color[2][]{}%
  }%
  \providecommand\transparent[1]{%
    \errmessage{(Inkscape) Transparency is used (non-zero) for the text in Inkscape, but the package 'transparent.sty' is not loaded}%
    \renewcommand\transparent[1]{}%
  }%
  \providecommand\rotatebox[2]{#2}%
  \ifx\svgwidth\undefined%
    \setlength{\unitlength}{159.73538733bp}%
    \ifx\svgscale\undefined%
      \relax%
    \else%
      \setlength{\unitlength}{\unitlength * \real{\svgscale}}%
    \fi%
  \else%
    \setlength{\unitlength}{\svgwidth}%
  \fi%
  \global\let\svgwidth\undefined%
  \global\let\svgscale\undefined%
  \makeatother%
  \begin{picture}(1,0.40415853)%
    \put(0,0){\includegraphics[width=\unitlength]{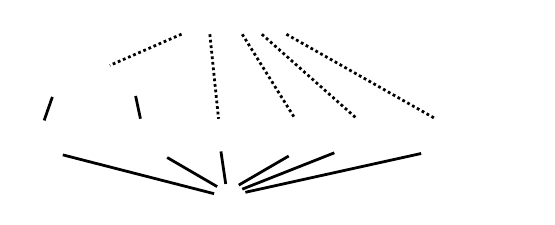}}%
    \put(0.47662211,0.1408685){\makebox(0,0)[lb]{\smash{$\{C_1\}$}}}%
    \put(0.58755481,0.14243245){\makebox(0,0)[lb]{\smash{$\{C_2\}$}}}%
    \put(0.74159712,0.14148737){\makebox(0,0)[lb]{\smash{$\{C_n\}$}}}%
    \put(0.69614356,0.14337712){\makebox(0,0)[lb]{\smash{$...$}}}%
    \put(0.03858983,0.24372876){\makebox(0,0)[lb]{\smash{$\{num,bool\}$}}}%
    \put(0.20114862,0.01098077){\makebox(0,0)[lb]{\smash{$\bot=\emptyset=Null$}}}%
    \put(0.09374181,0.36405035){\makebox(0,0)[lb]{\smash{$\top=\{num,...,C_n\}=\mathbb{O}$}}}%
    \put(-0.00401323,0.13835948){\makebox(0,0)[lb]{\smash{$\{num\}$}}}%
    \put(0.16954848,0.13976372){\makebox(0,0)[lb]{\smash{$\{bool\}$}}}%
    \put(0.33719748,0.140268){\makebox(0,0)[lb]{\smash{$\{list\}$}}}%
  \end{picture}%
\endgroup%

%% file: images/img-typesystem-stat.tex
%% Creator: Inkscape inkscape 0.48.4, www.inkscape.org
%% Accompanies image file 'typesystem-stat.pdf' (pdf, eps, ps)
\begingroup%
  \makeatletter%
  \providecommand\color[2][]{%
    \errmessage{(Inkscape) Color is used for the text in Inkscape, but the package 'color.sty' is not loaded}%
    \renewcommand\color[2][]{}%
  }%
  \providecommand\transparent[1]{%
    \errmessage{(Inkscape) Transparency is used (non-zero) for the text in Inkscape, but the package 'transparent.sty' is not loaded}%
    \renewcommand\transparent[1]{}%
  }%
  \providecommand\rotatebox[2]{#2}%
  \ifx\svgwidth\undefined%
    \setlength{\unitlength}{105.00910645bp}%
    \ifx\svgscale\undefined%
      \relax%
    \else%
      \setlength{\unitlength}{\unitlength * \real{\svgscale}}%
    \fi%
  \else%
    \setlength{\unitlength}{\svgwidth}%
  \fi%
  \global\let\svgwidth\undefined%
  \global\let\svgscale\undefined%
  \makeatother%
  \begin{picture}(1,0.73578382)%
    \put(0,0){\includegraphics[width=\unitlength]{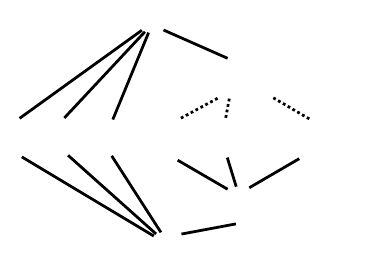}}%
    \put(-0.00236256,0.32931191){\makebox(0,0)[lb]{\smash{$\mathbb{N}$}}}%
    \put(0.13916599,0.32931191){\makebox(0,0)[lb]{\smash{$\mathbb{B}$}}}%
    \put(0.27088965,0.32726654){\makebox(0,0)[lb]{\smash{$\mathbb{L}$}}}%
    \put(0.37287003,0.32795878){\makebox(0,0)[lb]{\smash{$\{C_1\}$}}}%
    \put(0.54161543,0.33033778){\makebox(0,0)[lb]{\smash{$\{C_2\}$}}}%
    \put(0.77593805,0.32890018){\makebox(0,0)[lb]{\smash{$\{C_n\}$}}}%
    \put(0.70679603,0.33177478){\makebox(0,0)[lb]{\smash{$...$}}}%
    \put(0.48719494,0.4998498){\makebox(0,0)[lb]{\smash{$\{C_1,...,C_n\}=\mathbb{O}$}}}%
    \put(0.60931765,0.1492338){\makebox(0,0)[lb]{\smash{$\emptyset = Null$}}}%
    \put(0.43179488,0.01877453){\makebox(0,0)[lb]{\smash{$\bot$ 
}}}%
    \put(0.37082459,0.65230869){\makebox(0,0)[lb]{\smash{$\top$ 
}}}%
  \end{picture}%
\endgroup%

%% file: sections/hoare_logic.tex
The presentation of \textbf{dyn} and \textbf{stat}'s program logics closely follows \cite{AptDeBoerOlderog2012,AptDeBoerOlderogBook2009}. We start by introducing the assertion language (Figure~\ref{fig:dyn_ass_lang}). Essentially, it is weak second order logic, extended with the same constants $c_{\pure}$, operations $op(\Vector{\LExpX})$ and types used in \textbf{stat}. It will be used to reason about both \textbf{dyn} and \textbf{stat}, however.

Assertions contain typed logical expressions ($\LExpX$). Such expressions consist of typed logical variables, local/instance program variables $\VarLX$/$\LExpX.\VarIX$ (of type $\mathbb{O}$ in \textbf{dyn} and of some type $\BaseTypeX \in \Type_s$ in \textbf{stat} / same, with $\LExpX$ being of type $\mathbb{O}$) including \textPL{this}, typed constants and typed operations. Contrary to program expressions, logical expressions can access instance variables of objects other than \textPL{this}.

Logical expressions may only occur as parts of well-typed equations. Following \cite{DeBoerPierik2003}, \todo{in this paper, undefined operations yield $\bottom$. Find the other one!} undefined operations like dereferencing a $\Vnull$ value or accessing a sequence with an index out of bounds ($l[n]$ with $n \ge |l|$) yield a $\Vnull$ value and equality is non-strict with respect to such values ($\Vnull = \Vnull$ is $\mathit{true}$) to ensure a two-valued logic. Assertions are boolean combinations of such equations allowing quantification over finite sequences of elements of basic types.

We also introduce the following abbreviation for making reasoning about runtime types more convenient:

\vspace{0.1cm}
\centerline{$\mathit{Asrt} \ni \AsrtX,\AsrtY ::= \llbracket l \rrbracket \in T, \;\;\;\; T \in \Type$}
\centerline{$\llbracket l \rrbracket \in \{\PL{C}_1,...,\PL{C}_n\} \;\; \defined \;\; l \not= \Vnull \rightarrow [ l.\PL{@}\mathbf{c} = \oC_{\PL{C}_1} \vee ... \vee l.\PL{@}\mathbf{c} = \oC_{\PL{C}_n} ]$}
\vspace{0.1cm}

The reader may convince himself/herself that the following implications hold:

\vspace{0.1cm}
\begin{tabular}{llll}
\hspace{0.5cm}$\llbracket l \rrbracket \in \UnionTypeX_1 \wedge \llbracket l \rrbracket \in \UnionTypeX_2$ & $\rightarrow \llbracket l \rrbracket \in \UnionTypeX_1 \sqcap \UnionTypeX_2$ \hspace{0.5cm} & $\llbracket l \rrbracket \not\in \UnionTypeX$ & $\rightarrow \llbracket l \rrbracket \in \top \setminus \UnionTypeX$ \\
\hspace{0.5cm}$\llbracket l \rrbracket \in \UnionTypeX_1 \vee \llbracket l \rrbracket \in \UnionTypeX_2$ & $\rightarrow \llbracket l \rrbracket \in \UnionTypeX_1 \sqcup \UnionTypeX_2$ & $l_1 = l_2$ & $\rightarrow \exists \UnionTypeX \quantsep \llbracket l_1 \rrbracket \in \UnionTypeX \wedge \llbracket l_2 \rrbracket \in \UnionTypeX$
\end{tabular}
\vspace{0.1cm}

%
%
%
% The substitution $[\PL{u} := \PL{new}]$ is equivalent to the substitution $[\PL{u}/\PL{new}]$
% defined in \cite{DeBoerPierik2002CompAided}. It computes the weakest precondition of an object creation
% statement $\PL{u} := \PL{new}$, thereby handling references to the newly created
% object as well as quantification over objects.
%
%Its signature can, however, be extended.
%%TODO{mention conditions for purity?}
%
% To link the objects of the programming language with the values in our assertions,
% we define the following predicates:
% %
% \begin{align*}
%  \mathit{num}(o) \equiv & o \not= \Vnull \wedge o.@\mathbf{c} = \oC_{num} \\
%  \mathit{num}(o,n) \equiv & \mathit{num}(o) \wedge \\
%  & n = 0 \rightarrow o.@\mathrm{pred} = \Vnull \wedge \\
%  & n > 0 \rightarrow \mathit{num}(o.@\mathrm{pred},n-1)
% \end{align*}
% for all $o$ of type $\mathbf{object}$, $n$ of type $\mathbb{N}$.
% %
% \begin{align*}
%  bool(o,b) \equiv (o \not= \Vnull \wedge o.@\mathbf{c} = \oC_{bool}) \rightarrow (b \leftrightarrow o.@\mathrm{to\_ref} \not= \Vnull) 
% \end{align*}
% for all $o$ of type $\mathbf{object}$, $b$ of type $\mathbb{B}$.
%
%
% Also, for logical expressions $e \in L$, we extend the state-access to $\sigma(e)$ in the
% canonical way.\todo{also give a semantics for the assertions! at least: $L.@v$}
%
Selected differences between the Hoare-style axiomatic semantics for \textbf{dyn} and \textbf{stat} are contrasted in Figure~\ref{fig:hoare_logic}. While the semantics for \textbf{stat} are standard\footnote{they closely follow other Hoare logics for statically typed languages \cite{AptDeBoerOlderog2012,DeBoerPierik2003,AptDeBoerOlderogBook2009}}, the rules for \textbf{dyn} were modeled after \cite{Gardner12programlogicJS}. Omitted rules are listed in Appendix~\ref{app:hoare_logic}. In Hoare triples $\{p\} S \{q\}$, the special variable $\result$ is only allowed in the postcondition $q$ and denotes the return value of $S$. The rules will be analyzed in the next section.

\todo{typical assertions should be $p,q$}
\begin{figure}[p]
$\AsrtX,\AsrtY \in \mathit{Asrt} ::= \; \LExpX = \LExpX \mid \neg \AsrtX \mid \AsrtX \wedge \AsrtX \mid \exists v : \BaseTypeX^* \quantsep \AsrtX$ \hspace{0.5cm} $\BaseTypeX \in \Type_s$

$\LExpX \in \mathit{LExp} ::= \LVarX \mid \VarLX \mid \LExpX.\VarIX \mid \Vnull \mid \this{} \mid$ if $\LExpX$ then $\LExpX$ else $\LExpX$ fi $\mid \LExpX = \LExpX \mid |\LExpX| \mid \LExpX[\LExpX] \mid c_{\pure} \mid op(\Vector{\LExpX})$

with the usual abbreviations: $\AsrtX \vee \AsrtY \defined \neg(\neg \AsrtX \wedge \neg \AsrtY)$, $\AsrtX \rightarrow \AsrtY \defined \neg \AsrtX \vee \AsrtY$, $\AsrtX \leftrightarrow \AsrtY \defined \AsrtX \rightarrow \AsrtY \wedge \AsrtY \rightarrow \AsrtX$, $\exists \LVarX : \BaseTypeX \quantsep \AsrtX \defined \exists \LVarX' : \BaseTypeX^* \quantsep |\LVarX'| = 1 \wedge \AsrtX[\LVarX/\LVarX'[0]]$, $\forall \LVarX : \BaseTypeX \quantsep \AsrtX \defined \neg \exists \LVarX : \BaseTypeX \quantsep \neg \AsrtX$

\caption{\label{fig:dyn_ass_lang} Syntax of the assertion language}
\end{figure}
\begin{figure}[p]
\begin{center}
  Hoare logic rules for
\end{center}
\begin{minipage}{0.5\textwidth}
  \begin{center}
   \textbf{dyn}
  \end{center}
\end{minipage}
\begin{minipage}{0.5\textwidth}
  \begin{center}
   \textbf{stat}
  \end{center}
\end{minipage}

\vspace{0.2cm}
\noindent RULE: Assignment (ASGN)
\vspace{0.2cm}

\begin{minipage}{0.5\textwidth}
  \begin{center}
  \AxiomC{$\{p\} e \{q[\VarLX := \result] \}$}
  \UnaryInfC{$\{p\} \VarLX := e \{q\}$}
  \DisplayProof
  \end{center}
\end{minipage}
\begin{minipage}{0.5\textwidth}
  \begin{center}
  \AxiomC{$\{p[\VarLX := e_{\pure}]\} \VarLX := e_{\pure} \{p\}$}
  \DisplayProof
  \end{center}
\end{minipage}

\vspace{0.2cm}
\noindent RULE: Conditional (COND) (\TypesafePC{type-safe} partial correctness)
\vspace{0.2cm}

\begin{minipage}{0.5\textwidth}
  \begin{center}
  \AxiomC{$\begin{matrix}
	    \{p\} e \{r \wedge \mathit{bool\_test} \} \\
	    \{r \wedge b\} S_1 \{q\} \\
	    \{r \wedge \neg b\} S_2 \{q\}
	  \end{matrix}$}
  \UnaryInfC{$\{p\}$ if $e$ then $S_1$ else $S_2$ fi $\{q\}$}
  \DisplayProof
  \end{center}
\end{minipage}
\begin{minipage}{0.5\textwidth}
  \begin{center}
  \AxiomC{$\{p \wedge b_{\pure}\} S_1 \{q\}$}
  \AxiomC{$\{p \wedge \neg b_{\pure}\} S_2 \{q\}$}
  \BinaryInfC{$\{p\}$ if $b_{\pure}$ then $S_1$ else $S_2$ fi $\{q\}$}
  \DisplayProof
  \end{center}
\end{minipage}
\hspace{0.1cm}
where $b$ is a predicate and $\mathit{bool\_test} \defined \mTypesafePC{\llbracket \result \rrbracket \in \{bool\}} \wedge \mathbb{B}(\result,b)$
\newline
\hspace{0.2cm}
\noindent RULE: Method Call (METH)
\vspace{0.2cm}

\begin{minipage}{0.5\textwidth}
  \begin{center}
  \AxiomC{$\begin{matrix}
	    \{p_i\} e_i \{p_{i+1}[\PL{u}_i := \result] \}\text{ for }i \in \mathbb{N}_n \\
	    \{p_{n+1}\} \PL{u}_0.\PL{m}(\PL{u}_1,...,\PL{u}_n) \{q\}
	  \end{matrix}$}
	
  \UnaryInfC{$\{p_0\} e_0.\PL{m}(e_1,...,e_n) \{q\}$}
  \DisplayProof
  \end{center}
\end{minipage}
\begin{minipage}{0.5\textwidth}
  \begin{center}
  \AxiomC{$\{p\} \VarLX := \VarLX_0.\PL{m}(\VarLX_1,...,\VarLX_n) \{q\}$}
  \UnaryInfC{$\begin{matrix}
                \{p[\VarLX_0,...,\VarLX_n := e_{\pure 0},...,e_{\pure n}]\}\text{\hspace{2cm}} \\
                \text{\hspace{2cm}}\VarLX := e_{\pure 0}.\PL{m}(e_{\pure 1},...,e_{\pure n}) \{q\}
              \end{matrix}$}
  \DisplayProof
  \end{center}
\end{minipage}

\vspace{0.2cm}
where $\PL{u}_i \in \Var_L \text{ fresh}, \PL{u}_i \not\in var(e_j) \cup change(e_j)$ for all $i,j \in \mathbb{N}_n$.

\vspace{0.2cm}
\noindent RULE: Recursion (REC) (\textbf{dyn} and \textbf{stat}) (\TypesafePC{type-safe} partial correctness)

\begin{center}
\AxiomC{$\begin{matrix}
          A \vdash \{p\} S \{q\}, \\
          A \vdash \{p_i\} \PL{begin local this}, \Vector{\PL{u}_i} := \PL{v}'_i,\Vector{\PL{v}_i}; S_i \text{ end} \{q_i\}, i \in \mathbb{N}^1_n \\
          \mTypesafePC{p_i \rightarrow \llbracket \PL{v}'_i \rrbracket \in \{\PL{C}_i\}, i \in \mathbb{N}^1_n}
         \end{matrix}$}
       
\UnaryInfC{$\{p\} S \{q\}$}
\DisplayProof
\end{center}
where method $\PL{m}_i(\Vector{\PL{u}_i}) \{ S_i \} \in \Method_{\PL{C}_i}$, $A \equiv \{p_1\} \PL{v}'_1.\PL{m}_1(\Vector{\PL{v}_1}) \{q_1\}, ..., \{p_n\} \PL{v}'_n.\PL{m}_n(\Vector{\PL{v}_n}) \{q_n\}$.

% \begin{center}
% \AxiomC{$\begin{matrix}
%           A \vdash \{p\} S \{q\}, \\
%           A \vdash \{p_i\} \PL{begin local this}, \Vector{u_i} := l_{\pure i},\Vector{e_{\pure i}}; S_i \text{ end} \{q_i\}, i \in \mathbb{N}^1_n \\
%           p_i \rightarrow (\mStrongPC{l_{\pure i} \not= \Vnull} \wedge \mTypesafePC{l_{\pure i} \not= \Vnull \rightarrow l_{\pure i}.@\mathbf{c} = \oC_{C_i}}), i \in \mathbb{N}^1_n
%          \end{matrix}$}
%        
% \UnaryInfC{$\{p\} S \{q\}$}
% \DisplayProof
% \end{center}
% %
% where $m_i(\Vector{u_i}) \{ S_i \} \in \Method_{C_i}$ and $A \equiv \{p_1\} l_{\pure 1}.m_1(\Vector{e_{\pure 1}}) \{q_1\}, ..., \{p_n\} l_{\pure n}.m_n(\Vector{e_{\pure n}}) \{q_n\}$.

\caption{\label{fig:hoare_logic} Comparison of dynamically typed and statically typed Hoare logic rules}
\end{figure}

%% file: sections/layer.tex
Let us compare the proof rules given in Figure~\ref{fig:hoare_logic}. Obviously, the \textbf{dyn} rules are more complicated than their \textbf{stat} counterparts. Analyzing their differences, one can identify three core reasons why reasoning about dynamically typed programs is more complex than reasoning about statically typed ones.

\myparagraphB{1. Type safety:}
% In dynamically typed languages, type errors are runtime events. Like divergence (failures) do with total (strong) correctness, they give rise to a notion of correctness excluding them that we will call \emph{typesafe partial correctness}.
In Figure~\ref{fig:hoare_logic}, the parts ensuring type safety are \TypesafePC{marked}. Such \emph{type safety preconditions} are unnecessary in statically typed languages.

\myparagraphB{2. Mapping objects to values:}
Hoare logic for dynamically typed languages often uses predicates to map between program objects and logical values. For instance, the COND rule has to use the predicate $\mathbb{B}()$ to establish a correspondence between the program expressions $\ExprX$ and the logical expression $b$ of type $\mathbb{B}$. This additional layer of indirection not only reduces readability, but also hinders substitutions for pure expressions (see next paragraph).

\myparagraphB{3. Side-effecting expressions:}
In the \textbf{stat}-rules ASGN and COND, pure program expressions $e_{\pure}$ and $b_{\pure}$ are directly used in logical assertions. Here, the clever design choice of a shared type system pays off. Unfortunately, dynamic typing forces us to relinquish this benefit, as the types of expressions are not statically known and impure expressions are ill-suited for logical reasoning. Observe also how \textbf{dyn}'s METH rule models the evaluation order using a sequence of intermediate predicates $p_i$, which would not be necessary for pure expressions. However, since \textbf{dyn} treats operations as method calls, the METH rule needs to be applied even for pure operations like $+, <, \wedge,$ etc, making properties of assignments and conditionals even more tedious to derive.

The following sections will explain how the layer of abstraction mitigates these issues.

%% file: sections/layer_decomposition.tex
Like already mentioned, the fact that type errors are runtime events in dynamically typed languages gives rise to the following notion of correctness:
\begin{definition}[Type-safe partial correctness]
$\{p\} \StmtX \{q\}$ holds in the sense of \emph{type-safe partial correctness} (written $\models_{\Ntypeerror \Npost} \{p\} \StmtX \{q\}$) if every non-diverging, non-failing computation starting in a state satisfying $p$ does not abort with a type error, but ends in a state satisfying $q$.
\end{definition}
The preconditions particular to proof rules for type-safe partial correctness are called \emph{type safety preconditions}. Being orthogonal to preventing divergence ($\Ndivergence$) and failures ($\Nfailure$) (calling a method on a null value) as well as ensuring the post-condition ($\Npost$), these correctness notions may be freely combined. Total correctness would hence be denoted $\Ntypeerror \Ndivergence \Nfailure \Npost$. In the proof rules given in Figure~\ref{fig:hoare_logic} (and Appendix~\ref{app:hoare_logic}), $\Ntypeerror$-preconditions are \TypesafePC{marked}.

% Figure~\ref{fig:proof-systems} arranges such combinations as pairs of type-unsafe and type-safe variants. 
% Clearly, these preconditions can also be combined with other notions of correctness yielding for instance ``typesafe strong partial correctness''.
%
% \begin{tabular}{l|l}
%  type-unsafe (no $\tau$) & type-safe ($\tau$) \\
%  \hline
%  partial correctness ($p$) & type-safe partial correctness ($\tau p$) \\
%  strong partial correctness ($s p$) & type-safe strong partial correctness ($\tau s p$) \\
%  type-unsafe total correctness ($t s p$) & total correctness ($\tau t s p$)
% \end{tabular}

In \HLs, type-safety preconditions are unnecessary. Regarding such preconditions, correctness proofs in statically-typed languages resemble those in dynamically typed languages for type-unsafe correctness notions. Omitting these preconditions hence is a first step in proving dynamically typed programs like statically-typed ones. This can be achieved by treating type safety issues seperately from other correctness issues.
\begin{definition}[Decomposition]
\label{def:decomposition}
The following rule is added to the proof system for a type-safe notion of correctness ($\tau X$):
%\begin{center}
\begin{floatingfigure}[l]{3cm}
\AxiomC{$\begin{matrix}
          \vdash_{X}\{p\} \StmtX \{q\} \\
          \vdash_{\Ntypeerror \Npost}\{p\} \StmtX \{\mathit{true}\}
         \end{matrix}$}
       
\UnaryInfC{$\{p\} \StmtX \{q\}$}
\DisplayProof
\end{floatingfigure}
%\end{center}
%
\noindent where $\vdash_{X}$ refers to the corresponding type-unsafe variant of the proof system while $\vdash_{\Ntypeerror \Npost}$ always refers to the proof system for type-safe partial correctness.
\end{definition}
\vspace{0.2cm}

Correctness of the decomposition rule follows directly from the semantic definition for type-safe partial correctness. Intuitively, it states that whenever $\models_{\Ntypeerror \Npost} \{p\} \StmtX \{\mathit{true}\}$ and the precondition $p$ have been established for some statement $\StmtX$, we can omit type safety preconditions when reasoning about $\StmtX$, although our program is dynamically-typed.

%% file: sections/object_mapping.tex
Mapping predicates are a further peculiarity of \HLd. However, when the types of all used variables are known, those predicates can be generated automatically. We will now provide a ``virtual'' variable $\hat{\PL{u}}$ of the corresponding base type for each object variable $\PL{u}$ that can be safely mapped.

First, a subset of ``pure'' (i.e. immutable) classes $\Class_{\pure} \subseteq \Class$ along with a function $\Psi$ mapping classes from $\Class_{\pure}$ to corresponding base types $\BaseTypeX \in \Type_s$ of the assertion language must be defined. For \textbf{dyn}, this mapping is

\vspace{0.1cm}
$\Psi(num) = \mathbb{N}$, $\Psi(bool) = \mathbb{B}$, $\Psi(list) = \mathbb{L}$, ...
\vspace{0.1cm}

The mapping can be extended to union types $\UnionTypeX \in \Type_d$ by defining

\vspace{0.1cm}
$\Psi(\{\}) = \mathit{Null}$, $\Psi(\{\ClassX\}) = \Psi(\ClassX)$ for $\ClassX \in \Class_{\pure}$ and $\Psi(\UnionTypeX) = \mathbb{O}$ otherwise.
\vspace{0.1cm}

% $\Psi(T) = \begin{cases}
%             \mathit{Null} & \text{if } T = \{\}, \\
%             \mathbb{N} & \text{if } T = \{\mathit{num}\}, \\
%             \mathbb{B} & \text{if } T = \{\mathit{bool}\}, \\
%             \mathbb{L} & \text{if } T = \{\mathit{list}\}, \\
%             \mathbb{O} & \text{else}
%            \end{cases}
% $
For each type $\BaseTypeX \in \mathcal{T}_s$, there is usually already a mapping predicate $\BaseTypeX(o,v): \mathbb{O} \times \Domain_\BaseTypeX \mapsto \mathbb{B}$ for mapping objects to values
as well as a safety predicate $\safe_\BaseTypeX(o): \mathbb{O} \mapsto \mathbb{B}$ defining under what condition this mapping is safe. For $\mathbb{N}$ these are\footnote{Expressing them using quantification over sequences instead of recursion is possible, but less readable.}

$\mathbb{N}(o,n) \defined \mathit{safe}_{\mathbb{N}}(o) \rightarrow (o.\PL{@pred} = \PL{null} \rightarrow n = 0 \; \wedge$

\hspace{3.35cm} $o.\PL{@pred} \not= \PL{null} \rightarrow \mathbb{N}(o.\PL{@pred}, n - 1))$ \hspace{0.5cm} and

$\mathit{safe}_{\mathbb{N}}(o) \defined o \not= null \wedge \llbracket o \rrbracket \in \{num\}$

We then introduce a new assertion language $\mathit{Asrt}_\Upsilon$ allowing the use of automatically mapped virtual variables $\VarXVirtual$. Its semantics is defined in terms of a mapping $\Upsilon: \mathit{Asrt}_\Upsilon \mapsto \mathit{Asrt}$ to the old assertion language.
\begin{definition}[Automatic Variable Mapping]
{
\renewcommand{\PL}[1]{\emph{#1}} \todo{check if this still looks right in the final version of the paper!}

Let $\VarX_1,...,\VarX_n$ be a sequence of variables that can be safely mapped to types $\BaseTypeX_1,...,\BaseTypeX_n$ and for which $\VarXVirtual_i$ occur free in $\AsrtX$ for $i \in \mathbb{N}^1_n$. Also, let $\LVarX_{\VarXVirtual_1},...,\LVarX_{\VarXVirtual_n}$ be a corresponding sequence of logical variables of types $\BaseTypeX_1,...,\BaseTypeX_n$. Then,

$\Upsilon(\AsrtX) \defined \exists \LVarX_{\VarXVirtual_1} : \BaseTypeX_1,...,\LVarX_{\VarXVirtual_n} :\BaseTypeX_n \quantsep \Upsilon_S(\AsrtX) \wedge \Upsilon_M(\AsrtX)$

$\Upsilon_S(\AsrtX) \defined \AsrtX[\VarXVirtual_1,...,\VarXVirtual_n := \LVarX_{\VarXVirtual_1},...,\LVarX_{\VarXVirtual_n}]$, $\Upsilon_M(\AsrtX) \defined \BaseTypeX_1(\VarX_1,\LVarX_{\VarXVirtual_1}) \wedge ... \wedge \BaseTypeX_n(\VarX_n,\LVarX_{\VarXVirtual_n})$

% The automatic variable mapping $\Upsilon(\AsrtX) \defined \Upsilon_S(\AsrtX) \wedge \Upsilon_M(\AsrtX)$ consists of two components: the substitution $\Upsilon_S(\AsrtX)$ and the mapping $\Upsilon_M(\AsrtX)$. Both are defined by induction over the structure of the $\Upsilon$-assertion $\AsrtX$. The only interesting cases are
% 
% $\Upsilon_S(\VarLXVirtual) \defined \LVarX_{\VarLXVirtual}$, $\Upsilon_S(\this{}.\VarIXVirtual) \defined \LVarX_{\VarIXVirtual}$, \hfill $\Upsilon_S(op(a_1,...,a_n)) \defined op(\Upsilon_S(a_1), ..., \Upsilon_S(a_n))$
% 
% $\Upsilon_M(\VarLXVirtual) \defined T(\VarLX, \LVarX_{\VarLXVirtual})$, $\Upsilon_M(\this{}.\VarIXVirtual) \defined T(\this{}.\VarIX, \LVarX_{\VarIXVirtual})$
% 
% where $\VarLX$ ($\VarIX$) is a local (instance) variable with $\mathit{safe}_T(\VarLX)$ ($\mathit{safe}_T(\this{}.\VarIX)$) and $\LVarX_{\VarLXVirtual}$ ($\LVarX_{\VarIXVirtual}$) is a fresh logical variable of type $T$. For $\Upsilon_S$, all other cases are preserving and for $\Upsilon_M$, they are conjunctions of sub-terms or $\mathit{true}$.
%
%$(\AsrtX) \defined \AsrtX[\VarLXVirtual_1,...,\VarLXVirtual_n,\LExpX_1.\VarIXVirtual_1,...,\LExpX_k.\VarIXVirtual_k := \VarLXAux_1,...,\VarLXAux_n,\VarIXAux_{\LExpX_1.\VarIXVirtual_1},...,\VarIXAux_{\LExpX_k.\VarIXVirtual_k}]$

% $\Upsilon'(p)$ is defined by induction over the structure of the $\Upsilon$-assertion $p$. The only non-identity cases are
% 
% $\Upsilon(l_1 = l_2) \defined \Upsilon(l_1) = \Upsilon(l_2) \wedge \Upsilon'(l_1) \wedge \Upsilon'(l_2)$, \hfill $\Upsilon'(l_1 = l_2) \defined \mathit{true}$
}
\end{definition}

The precise definition of which variables can be ``safely mapped'' depends on the type information available. For the verifier that will be discussed in Section~\ref{sec:auto:type_safety_verifier}, the $\VarX_i$ may be local variables $\VarLX$ or instance variables of the current object $\this{}.\VarIX$. Note that $\mathit{Asrt}_\Upsilon$ conservatively extends $\mathit{Asrt}$, as any assertion $\AsrtX \in \mathit{Asrt}$ is mapped to itself. We hence assume $\Upsilon$ to be implicitly applied to all assertions, enabling the pervasive use of automatic object mapping. For instance, assuming that $\mathit{safe}_{\mathbb{N}}(\PL{u})$ could be established in the lower layer, the $\Upsilon$-assertion $\VarLXVirtual < 5$ can be used instead of the equivalent $\exists \LVarX_{\VarLXVirtual} : \mathbb{N} \quantsep \LVarX_{\VarLXVirtual} < 5 \wedge \mathbb{N}(\VarLX,\LVarX_{\VarLXVirtual})$. To formally show that the automatic object mapping allows us to trivially map \textbf{stat} assertions into \textbf{dyn} assertions, we need a mapping $\Theta$ between their states.

\myparagraphB{Translating States:} $\Theta(\sigma_s) \defined \sigma_d$ where $\sigma_d$ is derived from $\sigma_s$ by introducing for every base type $\BaseTypeX \in \Type_s \setminus \{\mathbb{O},\mathit{Null}\}$ a (possibly infinite) set of objects $\{ o_v \mid v \in \Domain_\BaseTypeX \wedge \BaseTypeX(o_v,v) \}$ and substituting every variable $\VarX$ of base type $\BaseTypeX$, holding the value $v \in \Domain_\BaseTypeX$ by a variable $\VarX$ of type $\mathbb{O}$, referencing the object $o_v$. Furthermore, for each base type $\BaseTypeX \in \Type_s \setminus \{\mathbb{O},\mathit{Null}\}$, we identify any two objects $o_1,o_2$ iff $\BaseTypeX(o_1,v_1)$, $\BaseTypeX(o_2,v_2)$ and $v_1 = v_2$. We lift this equivalence to \textbf{dyn} states in the natural way. 
%Note that it is possible to keep $\sigma_d$ finite by introducing only those $o_v$ actually referenced.

\myparagraphB{Translating Assertions:} $\Theta(\AsrtX) \defined \AsrtX[\VarX_1,...,\VarX_n := \VarXVirtual_1,...,\VarXVirtual_n]$ where $\VarX_i$ are all variables that can be safely mapped and occur free in $\AsrtX$.
\begin{theorem}\label{thm:layer_assertion_homomorphism}
 For all assertions $p$ and \textbf{stat} states $\sigma$: \hfill $\sigma \models p \text{ iff } \Theta(\sigma) \models \Theta(p)$.
\end{theorem}
The automatic mapping requires safety predicates to be pre-established in the lower layer, which requires both type information and tracking of null values. \todo{state somewhere that we regard these both as type information}

%% file: sections/pure_expressions.tex
\HLs allow highly effective reasoning by including (syntactically identified) pure program expressions into their logical assertions. In this section, we will show that assuming the availability of type information in the lower layer, this concept is also applicable to dynamically-typed languages.

To define a pure subset of \textbf{dyn} expressions, one complements the set of ``pure'' classes $\Class_{\pure}$ with a set of ``pure'' (i.e. side-effect-free) methods $\Method_{\pure} \subseteq \Method$ and extends the function $\Psi$ to also map method- and constructor calls to corresponding logical expressions. Such an expression $\LExpX \in \mathit{LExp}$ of type $\BaseTypeX$ with free variables $\LVarX_0,...,\LVarX_n$ of types $\BaseTypeX_0, ..., \BaseTypeX_n$ can be interpreted as a function $f_\LExpX: \BaseTypeX_0  \times ... \times \BaseTypeX_n \mapsto \BaseTypeX$. We hence denote its type as $\mathit{LExp}(\BaseTypeX_0  \times ... \times \BaseTypeX_n \mapsto \BaseTypeX)$. The extension of the mapping $\Psi$ can then be written as follows:

For every pure operation \PL{m} of arity $n$:

\hspace{1cm} $\Psi: (\BaseTypeX_0.\PL{m}(\BaseTypeX_1,...,\BaseTypeX_n) \rightarrow \BaseTypeX) \mapsto \mathit{LExp}(\BaseTypeX_0 \times ... \times \BaseTypeX_n \mapsto \BaseTypeX)$

For every pure constructor $\PL{new C}$ of arity $n$:

\hspace{1cm} $\Psi: (\Psi(C).\PL{init}(\BaseTypeX_1,...,\BaseTypeX_n) \rightarrow \BaseTypeX) \mapsto \mathit{LExp}(\BaseTypeX_1 \times ... \times \BaseTypeX_n \mapsto \BaseTypeX)$

For the type $\mathbb{N}$ these are

\hspace{1cm} $\Psi(\mathbb{N}.\PL{init}(\mathit{Null})) = 0$, $\Psi(\mathbb{N}.\PL{init}(\mathbb{N})) = v_1 + 1$,

\hspace{1cm} $\Psi(\mathbb{N}.\PL{add}(\mathbb{N})) = v_0 + v_1$, $\Psi(\mathbb{N}.\PL{succ}()) = v_0 + 1$.

It is then possible to define a predicate $\ppure(\ExprX)$ automatically identifying pure expressions given type information for all variables used. $\Psi$ can be extended to map such pure program expressions to typed logical expressions. We denote the type of a pure expression by $\tau(e)$. Then, after establishing that

\vspace{0.1cm}
\centerline{$ \{\AsrtX[\hat{\result} := \Psi(\BaseTypeX_0.\MethodX(\BaseTypeX_1,...,\BaseTypeX_n)\rightarrow \BaseTypeX)]\} \VarLX_0.\MethodX(\VarLX_1,...,\VarLX_n) \{\AsrtX\} $}
\vspace{0.1cm}

with $\BaseTypeX_i = \tau(\VarLX_i)$ for all $i \in \mathbb{N}_n$ holds for all methods in $\Method_\pure$, the following axiom can be established by induction over the structure of $\ExprX$

\vspace{0.1cm}
\noindent AXIOM: PURE EXPR: \hspace{1cm} $\{\AsrtX[\hat{\result} := \Psi(\ExprX)]\} \ExprX \{\AsrtX\}$ where $\ppure(\ExprX)$
\vspace{0.1cm}

Combining the axiom with \textbf{dyn}-specific proof rules yields simplified rules for pure expressions that closely resemble those for \textbf{stat}. For instance:

\vspace{0.2cm}

\noindent
\begin{minipage}{0.4\textwidth}
  \noindent AXIOM: PURE ASGN

  \vspace{0.1cm}
  \centerline{$\{\AsrtX[\VarXVirtual := \Psi(\ExprX)]\} \VarX := \ExprX \{\AsrtX\}$}
  \vspace{0.1cm}

  where $\ppure(\ExprX)$, $\tau(\ExprX) \sqsubseteq \tau(\VarX)$.
\end{minipage}
\begin{minipage}{0.6\textwidth}
\noindent RULE: PURE COND
\vspace{0.1cm}

\centerline{
  \AxiomC{$\{p \wedge \Psi(e)\} S_1 \{q\}$}
  \AxiomC{$\{p \wedge \neg \Psi(e)\} S_2 \{q\}$}
  %\RightLabel{where $e$ is pure}
  \BinaryInfC{$\{p\}$ if $e$ then $S_1$ else $S_2$ fi $\{q\}$}
  \DisplayProof
}

  where $\ppure(e)$ and $\tau(e) = \mathbb{B}$.
\end{minipage}
\vspace{0.1cm}
\todo{think about: can the automatically inserted safety conditions break anything?}

% Note that it is important that the classes mapped in this way are immutable. Since object creation
% is a side-effect, we rely on the fact that these classes are immutable to logically regard all distinguishable instances of these classes as ``already created'' and equate them with their logical counterpart.

Definitions for $\ppure(e)$, $\Psi: \mathit{Expr}_d \mapsto \mathit{LExp}$ and $\tau(e)$ as well as omitted rules and soundness proofs can be found in Appendix~\ref{app:pure_expressions}. Finally, we are able to state the main theorem of this section: in combination with decomposition and automatic object mapping, above rules allow verification just like in statically typed languages. This follows from the fact that \textbf{stat} proofs closely resemble \textbf{dyn} proofs using these techniques.

\myparagraphB{Translating Programs:} Since \textbf{stat} $\subset$ \textbf{dyn}, we simply have $\Theta(S) \defined S$.

\myparagraphB{Translating Proofs:} $\Theta(\phi) = \varphi$ is defined inductively over the structure of the proof $\phi$ in Hoare logic for \textbf{stat}.
Applications of the rules ASGN,COND,LOOP and METH need to be substituted for applications of PURE ASGN, PURE COND, PURE LOOP and PURE METH + PURE ASGN respectively. Note that this is always possible as \textbf{stat} expressions are pure and well-typed pure assignments preserve safety predicates. \todo{do we need to show this?}
Applications of all other rules can be preserved, as they are identical for \textbf{dyn} and \textbf{stat}.
%
% $\phi \equiv$
% \AxiomC{$\phi_1$}
% \AxiomC{$...$}
% \AxiomC{$\phi_n$}
% \RightLabel{\scriptsize(X)}
% \TrinaryInfC{$\{p\} S \{q\}$}
% \DisplayProof
%
%
% $\varphi \equiv$
% \AxiomC{$\Theta(\phi_1)$}
% \AxiomC{$...$}
% \AxiomC{$\Theta(\phi_n)$}
% \RightLabel{\scriptsize(X).}
% \TrinaryInfC{$\{\Theta(p)\} S \{\Theta(q)\}$}
% \DisplayProof
%
%
\begin{theorem}
 \label{thm:translation_stat_to_dyn}
 For every \textbf{stat} program $S$ and every correctness proof $\phi$ of a property $\{p\} S \{q\}$ in Hoare logic for a particular correctness notion of \textbf{stat} programs, $\Theta(\phi)$ is a valid proof of the property $\{\Theta(p)\} S \{\Theta(q)\}$ in Hoare logic for a corresponding type-unsafe correctness notion of \textbf{dyn} programs.
\end{theorem}
Furthermore, since types for \textbf{stat} programs can be inferred, their type-safety proofs can be constructed automatically (see Section~\ref{sec:auto:typings2proofs}). Applying the decomposition rule (Definition~\ref{def:decomposition}) then yields a proof for type-safe correctness. It follows that for statically typable programs, deriving a proof in \textbf{dyn} (using the layer of abstraction) does not require any more effort than deriving it in \textbf{stat}. The remainder of this paper will discuss how the layer of abstraction can be applied to arbitrary dynamically typed programs by deriving the necessary type information. In Section~\ref{sec:verifying_ev-example}, we will demonstrate this point by proving the evaluator example correct.

%% file: sections/type_inference_and_type_safety.tex
% In this paper, we regard a \emph{type error} as a fatal runtime event occurring when a method is called
% on a receiver that does not support it (message not understood) or when the condition of a
% conditional or while loop does not evaluate to a boolean value.

A program $\pi$ is called \emph{type-safe} if no execution of $\pi$ can result in a type error. \emph{Type safety} is the problem of deciding whether a given program is type-safe. Since type errors can be regarded as a form of output, type safety is a nontrivial semantic property and hence undecidable for Turing complete languages by Rice's theorem.\todo{check the exact formulation}

A \emph{type} $\UnionTypeX$ is an element of a complete lattice $(\Type, \sqsubseteq) = (2^\Class,\subseteq)$.
%
% \emph{Type abstraction}
% means to abstract from the infinite domain of values $\mathcal{D}$ to the finite domain of types
% $\mathcal{T}$ using a Galois
% connection (a pair of a monotone abstraction function $\Fabstract: 2^\mathcal{D} \mapsto \mathcal{T}$ and a
% monotone concretization function $\Fconcrete: \mathcal{T} \mapsto 2^\mathcal{D}$, such that
% $\Fabstract \circ \Fconcrete \sqsubseteq id_{\mathcal{T}}$ and
% $\Fconcrete \circ \Fabstract \sqsubseteq id_{\mathcal{D}}$).
% \todo{is type abstraction used anywhere?}
%
A \emph{typing} $\typingX$ of a program $\pi$ is an arbitrary data structure giving rise to a mapping $\typingX(\StmtX): \mathit{Stmt} \mapsto \mathcal{T}$ from sub-statements $\StmtX$ of $\pi$ to types. It is important to stress that a sub-statement occurring multiple times in $\pi$ is treated as multiple different statements. One can think of statements as represented by their parse tree nodes.

% Matching the definition of a type error, it is required that for all method calls $e_0.\PL{m}(\Vector{e})$ in $\pi$, $e_0 \in dom(s)$ and for every conditional or while loop with a condition $e$ in $\pi$, $e \in dom(s)$.

A typing $\typingX$ for a program $\pi$ is called \emph{sound} iff in every execution of $\pi$, whenever a sub-statement $\StmtX$ is evaluated to a value $v$, then $v$ is of a type $\UnionTypeX \sqsubseteq \typingX(\StmtX)$. A typing $\typingX$ is at least as \emph{precise} as another typing $\typingY$, written $\typingX \sqsubseteq \typingY$, iff for all statements $\StmtX$ it holds that $\typingX(\StmtX) \sqsubseteq \typingY(\StmtX)$.

For a program $\programX$, the \emph{least precise type-safe typing} $\smin_\pi$ is a typing where for every method call $e_0.\PL{m}(e_1,...,e_n)$, $\smin_\programX(e_0) = \{ \ClassX \in \Class \mid \ClassX \text{ supports method $\PL{m}$ }\allowbreak{} \text{of }\allowbreak{} \text{arity }\allowbreak{} \text{$n$} \}$ and for every conditional or while loop with condition $e$, $\smin_\pi(e) = \{\mathit{bool}\}$ and for all other sub-statements $\StmtX$, $\smin_\programX(\StmtX) = \top$. By definition, a program $\programX$ is type-safe iff it has a sound\footnote{if a method call, conditional or while loop is unreachable, sound typings may assign the type $\bot$ to its receiver / condition.} typing $\typingX$ that is precise enough to establish type safety ($\typingX \sqsubseteq \smin_\programX$).

Type safety verifiers (type inference algorithms) derive a typing for a given program by over-approximating its behavior. A verifier is \emph{sound} iff the typings it derives are.

% Note that type inference cannot possibly be complete as it would then find consistent
% typings for all typesafe programs and hence decide type safety.

% Note however, that the typings are not necessarily consistent.
% Firstly the program could not be typesafe and secondly deriving consistent typings
% is equivalent to deciding type safety.

Note that given a typing $\typingX$ for a program $\programX$, it is straightforward to decide $\typingX \sqsubseteq \smin_\programX$. In fact, this would even be possible if $\typingX$, like $\smin_\programX$, would assign $\top$ to all non-receiver, non-condition subexpressions. However, deciding soundness usually requires more information. For this reason, sound type safety verifiers usually
a) assign types to all subexpressions and
b) provide a set of inference rules (commonly called a ``type system'')
allowing to check safety of their derived typings using this additional type information. A soundness proof for these rules with respect to the semantics of the programming language is a crucial part of proving such algorithms sound.

Some algorithms (e.g. context-sensitive ones) even assign multiple types to each statement (one for each context). While $\typingX(\StmtX)$ in this case yields the union of all types assigned to $\StmtX$, the soundness proof may differentiate these types.

As type safety verifiers differ, so do their typings. We associate with each verifier $\Verifier_X$ a kind of typing capturing its respective format and restrictions. Between such kinds, It is possible to translate in both directions. However, as the precision achievable with a verifier $\Verifier_X$ varies, so does the precision expressible using $\Verifier_X$-typings. For instance, while it is usually possible to translate path-insensitive typings into path-sensitive ones by assigning the same types to each path, the reverse direction entails merging paths and thus a loss of precision.

% In Section~\ref{sec:type_inference}, we will exemplary discuss the base case of
% simple algorithm based on a flow-insensitive, path-insensitive and context-insensitive flow-analysis.
% Section~\ref{sec:Extensions} will then discuss more powerful analyses.

% -- not important:
% Even our simple base-case algorithm already stores one type for every subexpression $S$ of the given
% program $\pi$ in its typing. Flow-sensitivity would require additionally storing one type for each
% assignment to a variable. Path-sensitivity would require one version of each expressions type for each
% path to it from the beginning of its method (which can be exponentially many in the size of the method).
% Context-sensitivity would require one version of a method's subexpression's types for each of its contexts.
% Here, the number of contexts varies greatly between the number of callsites (program locations where the
% method is called) \todo{cite Palsberg}, the number of execution paths that call the method \todo{cite callstring
% method} and potentially unbounded in cases of recursive polymorphism.\todo{mention?}

%% file: sections/extracting_type_information.tex
A \emph{type safety proof} for a statement $\StmtX$ is a proof of the property $\{\AsrtX\} \StmtX \{\mathit{true}\}$ for some precondition $\AsrtX$ in Hoare logic for type-safe partial correctness. When run from a state satisfying $\AsrtX$, it ensures type-safety of $\StmtX$ by establishing all type-safety preconditions.

Such proofs constitute a kind of typing as their assertions contain type information that is by definition sufficient to establish type safety. Soundness of these typings can be validated using the proof rules of Hoare logic. Before discussing how to extract type information from a Hoare logic proof, one should state that this information needs to be compatible with the type safety verifier to be useful for our purpose. We hence define  \emph{typing assertions} $\mathit{TAsrt} \subset \mathit{Asrt}$ as a subset of the assertion language modeling the capabilities of this verifier.

For instance, the verifier $\VEx$ that will be presented in Section~\ref{sec:auto:type_safety_verifier} is based on a flow-sensitive, path-sensitive data flow analysis. As usual, only local variables of the current method and instance variables of the current object are tracked flow-sensitively. The remainder of the heap is abstracted into a finite number of type variables $\typevarIvarX$ -- one for each instance variable $\VarIX$ of each class $\ClassX$.

Logically, $\VEx$ establishes a \emph{global typing invariant} of the form 
\[ \Ex{\TInv}(\typingX) \;\; \defined \;\; \forall o \quantsep \bigwedge\limits_{\ClassX \in \Class} \left( \llbracket o \rrbracket \in \{\ClassX\} \rightarrow \bigwedge\limits_{\VarIX \in \Var_{\ClassX}} \llbracket o.\VarIX \rrbracket \in \typingX(\typevarIvarX) \right) \]
for a $\VEx$-typing $\typingX$, stating the fact that the types assigned to the type variables $\typevarIvarX$ in $\typingX$ are over-approximating the actual types of those instance variables. Also, automatic verifiers provide for each program location the return type of the previously executed expression as well as the types of all variables tracked flow-sensitively. Logically, those can be regarded as a conjunction of typing literals (see below). Additionally, path sensitivity allows differentiating different paths leading to a program location and hence requires expressing alternatives, leading us to a disjunctive normal form of typing literals. Hence, only the literals allowed in typing assertions are verifier-specific\footnote{Adding the literals $\VarLX = \Vnull$ and $\this{}.\VarIX = \Vnull$ allows tracking null values.}. For $\VEx$ we define

\todo{insert changes from the talk}
\vspace{0.2cm}
$\mathit{TAsrt} \ni \TAsrtX ::= \TLitX \mid \TAsrtX \vee \TAsrtX \mid \TAsrtX \wedge \TAsrtX$, \hfill $\Ex{\mathit{TLit}} \ni \TLitX ::= \neg \TLitX \mid \llbracket \VarLX \rrbracket \in \UnionTypeX \mid \llbracket \this{}.\VarIX \rrbracket \in \UnionTypeX$
\vspace{0.2cm}

\todo{check if $\restricted$-Hoare logic also works with the typing assertions containing disjunctions and typing literals about instance variables of the current object}
% 
% old definition:
%
% Now, let us define the restricted set of assertions that is sufficient for proofs of type safety:
% %
% \todo{formally define allowed conjunctive refinements}
% \todo{allow conjunctive refinements}
% \begin{definition}
%  \emph{Typing assertions} are assertions of the form
%  \[ p_\restricted \equiv \TInv \wedge l_1 \wedge ... \wedge l_n \]
%  where $n \in \mathbb{N}$ and $l_i \equiv \llbracket v_i \rrbracket \in T_i$ with
%  $v_i \in \Var_L, T_i \in 2^{\Class}$ for $i \in \mathbb{N}^1_n$ are positive
%  typing literals.
% \end{definition}
% 
\todo{check if these things also hold with the new typing assertions}
% Typing assertions thus imply any subset of their literals,
% imply $\llbracket \this{}\PL{.@v} \rrbracket \in T$ iff $\llbracket \this{} \rrbracket \in \{C_1,...,C_n\}$
% with $\bigsqcup_{i = 1}^n T_{C_i.@v} \sqsubseteq T$ is among their literals and
% can for any $v \in \Var_L$ be trivially decomposed into
% $p_\restricted \equiv p'_\restricted \wedge \llbracket v \rrbracket \in T$ where
% $v \not\in \free(p'_\restricted)$.
%
%
%
%
% For each program location
% 
% 2 kinds of type information to extract
% 
% 1. flow-sensitive type combinations of local variables and instance variables of the current object
% 
% 2. return-types of expressions \& statements
% 
% 
% Model for type safety verifier:
% - flow-sensitive tracking of a combination of
%   - local variables
%   - instance variables of the currect object
%   - current expression's return type
%   => disjunctive normal form of typing literals

We will now define how to extract type information from Hoare logic proofs. In such a proof, each postcondition may contain flow-sensitive type information about variables as well as the return value $\result$ of the previous expression. Given an assertion $\AsrtX$, one extracts this information by first converting $\AsrtX$ into disjunctive normal form, treating typing literals, equations and quantifiers as literals and then applying a projection $\projection_X: \mathit{Asrt} \mapsto \mathit{TAsrt}_X$ that preserves $\wedge, \vee, \mu$ while mapping all literals $\not\in \mathit{TLit}_X$ to $\mathit{true}$ ($\defined \llbracket \this{} \rrbracket \in \top$).
Every assertion $\AsrtX$ thus implies $\projection_X(\AsrtX)$. Note that depending on the structure of $\AsrtX$, there might be a significant loss of precision. This is unproblematic, however, as supplying type information is in the user's interest. The following theorems show that sufficiently precise type information can always be supplied for type-safe programs.
\begin{lemma}\label{lem:extractability}
 For every assertion $\AsrtX$ and every $\Verifier_X$-typing assertion $\TAsrtX$ such that $\AsrtX \rightarrow \TAsrtX$, there exists an equivalent assertion $\AsrtX' \leftrightarrow \AsrtX$ such that $\projection_X(\AsrtX') = \TAsrtX$.
\end{lemma}
A $\Verifier_X$-typing assertion $\TAsrtX$ is \emph{most precise} for an assertion $\AsrtX$ iff $\AsrtX \rightarrow \TAsrtX$ and for all $\Verifier_X$-typing assertions $\TAsrtX'$, $\AsrtX \rightarrow \TAsrtX'$ implies $\TAsrtX \rightarrow \TAsrtX'$.
\begin{theorem}\label{thm:extractability}
 For every verifier $\Verifier_X$, each type safety proof $\psi$ has an equivalent proof $\psi'$, such that for every assertion $\AsrtX'$ in $\psi'$, $\projection_X(\AsrtX')$ is most precise for $\AsrtX'$.
\end{theorem}
Furthermore, one can define a projection $\projection^\VarX_X$ further projecting $\Verifier_X$ typing assertions to summary types for the variable $\VarX$ such that for all assertions $\AsrtX$, all variables $\VarX$ and all verifiers $\Verifier_X$ we have $\AsrtX \rightarrow \llbracket \VarX \rrbracket \in \projection_X^{\VarX}(\AsrtX)$. For $\VEx$:

$\Ex{\projection}^{\VarX}(\llbracket \VarX \rrbracket \in \UnionTypeX) \defined \UnionTypeX$, \hfill
$\Ex{\projection}^{\VarX}(\llbracket \VarX' \rrbracket \in \UnionTypeX) \defined \top$ with $\VarX' \not= \VarX$,
\hfill $\Ex{\projection}^{\VarX}(\neg \mu) \defined \top \setminus \Ex{\projection}^{\VarX}(\mu)$

%$\Ex{\projection}^{\VarX}(\llbracket \VarX' \rrbracket \not\in \UnionTypeX) \defined 

$\Ex{\projection}^{\VarX}(\TAsrtX \wedge \TAsrtX') \defined \Ex{\projection}^{\VarX}(\TAsrtX) \sqcap \Ex{\projection}^{\VarX}(\TAsrtX')$, \hfill $\Ex{\projection}^{\VarX}(\TAsrtX \vee \TAsrtX') \defined \Ex{\projection}^{\VarX}(\TAsrtX) \sqcup \Ex{\projection}^{\VarX}(\TAsrtX')$

We extend $\projection_X^{\VarX}$ to assertions by defining $\projection_X^{\VarX} \defined \projection_X^{\VarX} \circ \projection_X$. Using it, every type safety proof $\proofX$ gives rise to a $\Verifier_X$-typing $\typingX_\proofX$ assigning every sub-statement $\StmtX$ the type $\projection_X^{\result}(q_1 \wedge ... \wedge q_k)$ where $q_i$ are the postconditions of all Hoare triples of the form $\{\AsrtX_i\} \StmtX \{\AsrtY_i\}$ in $\proofX$.

\begin{theorem}[Completeness relative to Hoare logic]
\label{thm:relative_completeness}
 Given completeness of the Hoare logic, for every type-safe program $\pi$ there exists a type safety proof $\psi$ such that $\typingX_\psi$ is sound and precise enough to establish type safety: $\typingX_\psi \sqsubseteq \smin_\pi$.
\end{theorem}

It follows that no (sound) typing can be more precise than a type safety proof. It is hence possible to translate from all other kinds of typings into them without incurring any loss of precision.

%% file: sections/auto.tex
% There are two things we gain from using type safety proofs as typings:
% %
% \begin{enumerate}
%  \item Completeness -- when using a complete Hoare logic, it is possible to derive type information for every type-safe program.
%  \item Precision -- in order to achieve completeness, the assertions used in Hoare logic must be able to precisely characterize every set of states possibly occurring in a program. Since types are but an abstraction of concrete values, these assertions must therefore also be able to express even the most precise type information. It is hence possible to translate from all other kinds of typings into type safety proofs without incurring any loss of precision.
% \end{enumerate}
%
% Those benefits, however, come at a price: like with general Hoare logic proofs, finding a type safety proof for a given program is undecidable. There will hence always be cases in which an automatic type safety verifier will not be able to derive the proof automatically.

As type safety is undecidable and full automation hence only achievable at the expense of completeness, we instead aim at semi-automation by integrating a suitable automatic type safety verifier into Hoare logic. Such verifiers have to satisfy the following requirements:
\begin{itemize}
 \item {\bf Soundness} - it safely over-approximates the actual program behavior. % When no type errors are found, the program is type-safe.
 \item {\bf Monotonicity} - increasing precision cannot create type errors.
 \item {\bf Refinements} - provides an interface to supply trusted assumptions for increasing precision (Section~\ref{sec:auto:type_safety_verifier}). These must be treated
 \begin{itemize}
  \item Flow sensitively - assumptions should have an associated program location and affect only data flows from that location onward.
  \item Path sensitively - assumptions should be able to use disjunctions ($\vee$) to express alternatives. The verifier should treat these alternatives like different paths reaching the associated program location.
 \end{itemize}
 \item {\bf Termination} - terminates on all inputs (programs).\footnote{Potentially non-terminating analyses like \cite{Henglein91typeinference} must be performed iteratively by first generating a base result and then refining it towards higher precision (similar to \cite{SridharanBodik2006}): thus, they can be interrupted any time and yield the most precise result reached.}
\end{itemize}
Note that flow- and path sensitivity are required only for refinements\footnote{Flow- or path-insensitive assumptions would increase the annotation-burden.}, not for the verifier itself. Also, the chosen Hoare logic must be powerful enough to express the verifier's reasoning. While we consider these requirements to be modest and can hardly imagine an analysis not amendable using this approach, proving this is difficult. In this paper we will therefore concentrate on automatic verifiers based on flow-analysis which are known to contain quite powerful analyses \cite{KastrinisSmaragdakis2013,XuRountev2008,SridharanBodik2006}.
% Section~\ref{sec:auto:type_safety_verifier} will give an exemplary context-insensitive analysis and Section~\ref{sec:Extensions} will discuss extensions along the usual dimensions.

%% file: sections/example_type_inference.tex
In this section we will introduce an exemplary automatic type safety verifier $\VEx$ to in the following complement our abstract discussion with concrete examples using $\VEx$-typings.

To also shed some light on the minimum requirements above, we chose a minimalistic one exactly fulfilling these criteria. $\VEx$ is based on a sound, flow-sensitive data flow analysis and resembles the work of Palsberg et al. \cite{PalsbergS91TypeInfernceOO}. As required, we allow specifying a set of path-sensitive trusted assumptions. However, the analysis does not introduce path-sensitivity by itself.

\myparagraphB{Flow Sensitivity}
Intra-procedurally, local variables as well as instance variables of the current object are tracked flow-sensitively. As usual, this is realized by converting all statements to static single-assignment form (SSA) with respect to these variables prior to analysis. A $\VEx$-typing $\typingX$ hence assigns one type $\typingX(\llbracket \VarX \rrbracket, \locX)$ per program location $\locX$ to each such variable $\VarX$.

\myparagraphB{Path Sensitivity}
To realize intra-procedural path-sensitivity, for each path $j \in \mathit{path}(\post{\StmtX})$ from the start of a method to each program location $\post{\StmtX}$ in the method, the previous sub-statement $\StmtX$ and each flow-sensitively tracked variable $\VarX$ are assigned distinct types $\typingX(\llbracket \StmtX \rrbracket, j)$ and $\typingX(\llbracket \VarX \rrbracket, \post{\StmtX}, j)$ respectively. Program locations are denoted by $\pre{S}$ (or $\post{S}$) for the beginning (end) of a sub-statement $S$ of $\pi$.
%We use $\pre{S}$ ($\post{S}$) to denote the program location before (after) statement $S$.

\myparagraphB{Null Pointers}
Although \cite{PalsbergS91TypeInfernceOO} is a pure type analysis, we use the same algorithm to also perform  null-pointer analysis. This is realized by defining the value $\Vnull$ as the only
instance of a class $\mathit{Null}$ and furthermore explicitly inserting the class $\mathit{Null}$ into all types whose expression may evaluate to $\Vnull$ instead of implicitly allowing $\Vnull$ to be element of every type's domain. We hence define

$\llbracket \VarX \rrbracket \in \{\mathit{Null},\ClassX_1,...,\ClassX_n\} \equiv \VarX = \Vnull \vee \llbracket \VarX \rrbracket \in \{\ClassX_1,...,\ClassX_n\}$

% The type system is based on union types represented as sets of classes with $\subseteq$ as subtyping relation ($(\mathcal{T},\sqsubseteq) = (2^{\Class},\subseteq)$). Class types of the form $[\PL{m}(T_1,...,T_n)\rightarrow R]$ where $T_i$ and $R$ are type variables are used to express capabilities (to support a method called $\PL{m}$ of arity $n$) for structural subtyping. As usual $\Vnull$ is contained in every type.
%\todo{move to the discussion of types in Setting-Section}

The interested reader may find a more detailed account of algorithm $\VEx$ in Appendix~\ref{app:verifier_VEx}.

\begin{figure}
 {\def\svgwidth{350px}
   \input{images/img-typing-assertions}
 }
\caption{\label{fig:logics} Overview of the logics involved and the mappings between them}
\end{figure}
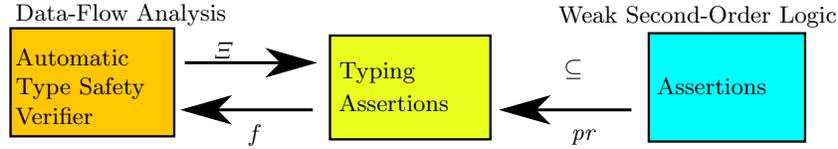

%We will now (and in the next section) introduce a number of mappings between different representations of type information as depicted in Figure~\ref{fig:logics}.  As a logical model for the automatic verifier, typing assertions are well-suited as an intermediary between assertion language and automatic verifier. The same colors are used in Figure~\ref{fig:concept} to indicate the representation used in each step. %Typing assertions hence serve both as a logical representation of trusted assumptions and as an assertion language for the lower layer.

\myparagraphB{Typings to Logic}
The function $\typingassertion_X$ maps the type information contained in a flow-sensitive, path-sensitive $\Verifier_X$-typing $\typingX$ for each program location $\locX$ into a typing assertion. $\VEx$-typings $\typingX$ are flow-sensitive, as $\typingX(\llbracket \VarX \rrbracket,\locX)$, the type of the variable $\VarX$ at program location $\locX$ takes strong updates into account. Hence, for $\VEx$-typings, the function $\Ex{\typingassertion}$ can be defined as:

$\Ex{\typingassertion}(\typingX,\pre{S}) \defined \bigvee_{j \in \mathit{path}(\pre{S})} \bigwedge_{\PL{x} \in \Var_{\Fmethod(\pre{S})}} \llbracket \PL{x} \rrbracket \in \typingX(\llbracket \PL{x} \rrbracket, \pre{S},j )$

$\Ex{\typingassertion}(\typingX,\post{S}) \defined \bigvee_{j \in \mathit{path}(\post{S})} \left( \llbracket \result \rrbracket \in \typingX(\llbracket S \rrbracket,j) \wedge \bigwedge_{\PL{x} \in \Var_{\Fmethod(\post{S})}} \llbracket \PL{x} \rrbracket \in \typingX(\llbracket\VarX\rrbracket,\post{S},j) \right)$

where $S$ is a sub-statement of $\pi$, $\Fmethod(L)$ denotes the method a program location $L$ belongs to and $\Var_{\PL{C.m}} \defined \mathit{var}(S_\PL{C.m}) \cup \mathit{change}(S_{\PL{C.m}})$.

\begin{definition}[Refinement of Typings]
 Let $\typingX$ be a $\VEx$-typing derived for a program $\programX$. Then a \emph{conjunctive refinement step} of $\typingX$ using the trusted assumption $\tau$ at program location $L$ is a quadruple $(\typingX,\tau,L,\typingY)$, written $\typingX \Ex{\stackrel{\tau,L}{\rightarrow}}  \typingY$ with the $\VEx$-typing $\typingY$ being derived for a program $\programX'$ resulting from $\programX$ by inserting the Statement $\mathcal{R}_\TAsrtX$ just before $L$ and $\mathcal{R}_\tau$ being defined inductively as
 
 $\mathcal{R}_{\llbracket \VarX \rrbracket \in \UnionTypeX} \defined \VarX := \VarX \sqcap \UnionTypeX$\footnote{the type filter $\VarX \sqcap \UnionTypeX$ is defined in Appendix~\ref{app:verifier_VEx}}

 $\mathcal{R}_{\TLitX \wedge \TLitX'} \defined \mathcal{R}_{\TLitX}; \mathcal{R}_{\TLitX'}$
 
 $\mathcal{R}_{\TAsrtKX \vee \TAsrtKX'} \defined \PL{if} ... \PL{ then} \mathcal{R}_{\TAsrtKX} \PL{else }\mathcal{R}_{\TAsrtKX'} \PL{ end}$\footnote{the condition does not matter as $\Ex{\Verifier}$ will treat conditionals non-deterministically anyway.}
 
\end{definition}

\begin{theorem}\label{thm:refinement_monotonicity}
 For all conjunctive refinements $\typingX \Ex{\stackrel{\tau,L}{\rightarrow}} \typingY$, \hspace{0.5cm} $\typingY \sqsubseteq \typingX$ holds.
\end{theorem}

%% file: images/img-typing-assertions.tex
%% Creator: Inkscape inkscape 0.48.5, www.inkscape.org
%% Accompanies image file 'typing-assertions.pdf' (pdf, eps, ps)
\begingroup%
  \makeatletter%
  \providecommand\color[2][]{%
    \errmessage{(Inkscape) Color is used for the text in Inkscape, but the package 'color.sty' is not loaded}%
    \renewcommand\color[2][]{}%
  }%
  \providecommand\transparent[1]{%
    \errmessage{(Inkscape) Transparency is used (non-zero) for the text in Inkscape, but the package 'transparent.sty' is not loaded}%
    \renewcommand\transparent[1]{}%
  }%
  \providecommand\rotatebox[2]{#2}%
  \ifx\svgwidth\undefined%
    \setlength{\unitlength}{649.72797852bp}%
    \ifx\svgscale\undefined%
      \relax%
    \else%
      \setlength{\unitlength}{\unitlength * \real{\svgscale}}%
    \fi%
  \else%
    \setlength{\unitlength}{\svgwidth}%
  \fi%
  \global\let\svgwidth\undefined%
  \global\let\svgscale\undefined%
  \makeatother%
  \begin{picture}(1,0.14985689)%
    \put(0,0){\includegraphics[width=\unitlength]{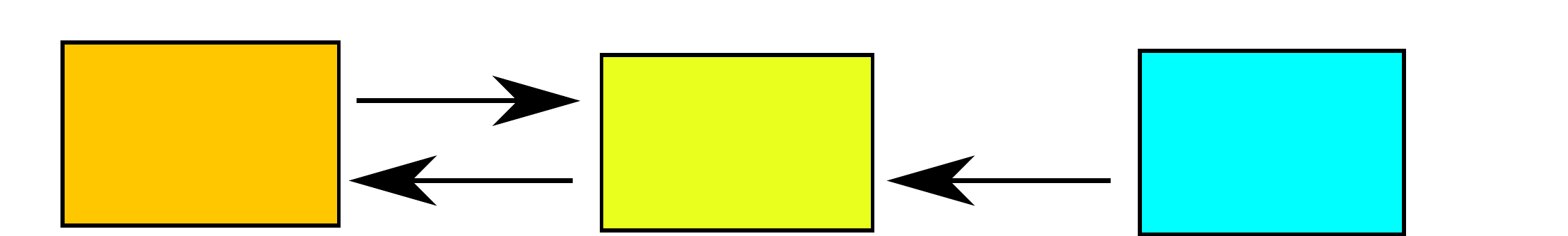}}%
    \put(0.04527516,0.01980432){\color[rgb]{0,0,0}\makebox(0,0)[lb]{\smash{Verifier}}}%
    \put(0.39341424,0.03352624){\color[rgb]{0,0,0}\makebox(0,0)[lb]{\smash{Assertions}}}%
    \put(0.73572287,0.05178628){\color[rgb]{0,0,0}\makebox(0,0)[lb]{\smash{Assertions}}}%
    \put(0.62979612,0.12935392){\color[rgb]{0,0,0}\makebox(0,0)[lb]{\smash{Weak Second-Order Logic}}}%
    \put(0.63546389,0.07159371){\color[rgb]{0,0,0}\makebox(0,0)[lb]{\smash{$\subseteq$}}}%
    \put(0.28015336,0.09070742){\color[rgb]{0,0,0}\makebox(0,0)[rb]{\smash{$\translation$}}}%
    \put(0.64383825,0.00248135){\color[rgb]{0,0,0}\makebox(0,0)[lb]{\smash{$\projection$}}}%
    \put(0.31047326,0.00123755){\color[rgb]{0,0,0}\makebox(0,0)[rb]{\smash{$f$}}}%
    \put(0.39372089,0.06835223){\color[rgb]{0,0,0}\makebox(0,0)[lb]{\smash{Typing
}}}%
    \put(0.04527516,0.08323736){\color[rgb]{0,0,0}\makebox(0,0)[lb]{\smash{Automatic
}}}%
    \put(0.04558181,0.05152084){\color[rgb]{0,0,0}\makebox(0,0)[lb]{\smash{Type Safety
}}}%
    \put(0.04527516,0.1305012){\color[rgb]{0,0,0}\makebox(0,0)[lb]{\smash{Data-Flow Analysis
}}}%
  \end{picture}%
\endgroup%

%% file: sections/translating_typings_into_proofs.tex
\begin{definition}[Typing Proof]
 A \emph{typing proof} $\psi$ for a $\Verifier_X$-typing $\typingX$ of a statement $\StmtX$ is a minimal\footnote{all Hoare triples in $\psi$ must contribute to establishing the conclusion.} proof of the property $\{p\} S \{\mathit{true}\}$ for some precondition $p$ in Hoare logic for partial correctness that for every sub-expression $\ExprX$ of $\StmtX$ contains a Hoare triple $\{p_\ExprX\} \ExprX \{q_\ExprX\}$ with $\projection_X^\result(q_\ExprX) \sqsubseteq \typingX(\ExprX)$.
\end{definition}
Technically, $\psi$ only establishes soundness of the typing $\typingX$ (by being a Hoare logic proof and $\typingX_\psi \sqsubseteq \typingX$). However, when $\typingX_\psi \sqsubseteq \smin_S$, $\psi$ can be turned into a type safety proof by changing the proof system to Hoare logic for type-safe partial correctness and trivially establishing the type safety preconditions. Hence, typing proofs are well-suited as intermediate steps towards type safety proofs.

Recall that $\Verifier_X$-typings can be checked for soundness using a $\Verifier_X$-specific inference system. It is hence possible to extend $\typingassertion_X$ to mechanically derive a typing proof $\psi = \translation_X(\programX,\typingX)$ for a sound $\Verifier_X$-typing $\typingX$ for a program $\programX$ by translating the rules of this inference system into Hoare logic and establishing $\TInv_X(\typingX)$ as a global invariant.
In such a proof, each assertion $\AsrtX$ at program location $L$ is exactly the typing assertion $\typingassertion_X(\typingX,L)$. We hence write $\typingassertion_X(\psi,L) \defined \typingassertion_X(\typingX,L)$.

The interested reader may find an examplary translation for $\VEx$-typings in Appendix~\ref{app:typing_to_proof}.
Note that $\translation_X$ allows using automatically derived type information in Hoare logic proofs even in theorem proving environments trusting only propositions that they verified a proof for.

% notes on the proof:
% - do all the rules fit together?
% - side conditions satisfied?
%   - VAR: p -> [v] < T
%   - IVAR: p -> [this.@v] < T
%   - ASSIGN: [E] < [v], [E] = [v := E]
%   - COND: [E_1] < [if..end], [E_2] < [if..end]
%   - SEQ: [E_2] = [E_1;E_2]
%   - LOOP: [while...done] < {}
%   - METH: [p_i] < [q_i] for i \in {1,...,n}, [r] < [E.m(...)]

% rules:
% 
% CONSEQUENCE-RULE with
% $T_1 = T_2 \leftrightarrow T_1 \subseteq T_2 \wedge T_2 \subseteq T_1$,
% $T_1 \subseteq T_2 \rightarrow T_1 \subseteq T_2 \cup T_3$ for all $T_3 \in 2^{\Class}$
% 
% INVARIANCE-RULE

% null|v|@v|this|v := E|@v := E|if|E;E|while
% E.m(E)|new C(E)|block|recursion

% $s(\llbracket E \rrbracket) = \{C_1,...,C_n\}$
% 
% $\{\} E \{ \mathbf{r.c} = o_1 \vee ... \mathbf{r.c} = o_n \}$

%% file: sections/two_layered_proofs.tex
A \emph{two-layered proof} is a Hoare logic proof for a type-safe notion of correctness of a dynamically-typed program, in which every assertion has the form $\TAsrtX \wedge \AsrtX$ for a typing assertion $\TAsrtX$ and an assertion $\AsrtX$. While $\AsrtX$ is user-editable, $\TAsrtX$ is meant to be created and modified solely by an automated type safety verifier. We refer to $\TAsrtX$ as the ``lower layer'' and $\AsrtX$ as the ``higher layer'' of the proof/assertion.
\begin{theorem}[Two-Layered Proof Construction]\label{thm:two-layered-proof-construction}
Given a typing proof $\phi_l$ and a proof $\phi_h$ for the same program $\pi$, it is always possible to construct a two-layered-proof $\phi$ with $\phi_l$ as lower and $\phi_h$ as higher layer.
\end{theorem}
Starting from a typing proof $\translation_X(\pi,\typingX)$ in the lower layer and only $\mathit{true}$ in the higher layer, proofs in the higher layer are supported by type information from the lower layer (Section~\ref{sec:layer}). The type information may also be refined:
% and interleaving proof steps for type safety and other correctness properties.
%
\begin{definition}[Refinement of Typing Proofs]
Let $\psi = \translation_X(\pi,\typingX)$ be a typing proof generated by a typing $\typingX$ of a program $\pi$. Then each conjunctive refinement step $\typingX \stackrel{\tau,L}{\rightarrow} \typingY$ gives rise to a \emph{conjunctive proof refinement step} $\psi \stackrel{\tau,L}{\rightarrow} \psi'$ with $\psi' = \translation_X(\pi,\typingY)$.
\end{definition}
Let $\psi_l = \translation_X(\pi,\typingX)$ be the lower layer of a two-layered proof $\psi$. Then whenever a typing literal appears within the higher layer $\AsrtX$ of an assertion at program location $L$ in $\psi$, the lower-layer proof $\psi_l$ is substituted by the result $\psi'_l$ of the conjunctive proof refinement step $\psi_l \stackrel{\projection_X(\AsrtX),L}{\rightarrow} \psi'_l$.
In such refinements $\typingassertion_X(\psi'_l,L')\allowbreak{ }\rightarrow\allowbreak{ }\typingassertion_X(\psi_l,L')$ holds for all $L'$ due to Theorem~\ref{thm:refinement_monotonicity}. Higher layer proof steps depending on lower layer information hence remain valid.

%% file: sections/verifying_ev-example.tex
To demonstrate how the techniques developed enable the convenient verification of dynamically typed programs despite hard typing problems, we will proof the evaluator example both type-safe and correct.
Figure~\ref{fig:correctness-ev-example} shows all annotations\footnote{Again, all recursive predicates can instead be expressed using quantification over sequences, at the expense of readability} necessary to prove that \verb|calc()| derives a given term's value.
\begin{figure}
\verb|def |$\mathit{valuetree2}(t,n) \defined \mathbb{N}(t[0],\widehat{\PL{VALUE}}) \wedge \mathbb{N}(t[1],n),$ \hfill $\mathit{valuetree1}(t) \defined \exists n \quantsep \mathit{valuetree2}(t,n)$

\verb|def |$\mathit{vartree3}(t,e,x) \defined \mathbb{N}(t[0],\ProgVirtual{\PL{VAR}}) \wedge \mathbb{S}(t[1],x) \wedge e[x] \not= \Vnull$

\verb|def |$\mathit{vartree2}(t,e) \defined \exists x \quantsep \mathit{vartree3}(t,e,x)$

\verb|def |$\mathit{optree5}(t,e,\mathit{op},l,r) \defined \mathbb{N}(t[0],\ProgVirtual{\PL{OP}}) \wedge \mathbb{N}(t[1],\mathit{op}) \wedge \mathbb{L}(t[2],l) \wedge \mathit{parsetree2}(l,e) \wedge \mathbb{L}(t[3],r) \wedge \mathit{parsetree2}(r,e)$

\verb|def |$\mathit{optree2}(t,e) \defined \exists \mathit{op},l,r \quantsep \mathit{optree5}(t,e,\mathit{op},l,r)$

\verb|def |$\mathit{parsetree2}(t,e) \defined \mathit{valuetree1}(t) \vee \mathit{vartree2}(t,e) \vee \mathit{optree2}(t,e)$

%\verb||

{\color{blue}
\verb|def |$\mathit{treeval3}(\mathit{tree},\mathit{env},n) \defined$

\verb|     |$\mathit{valuetree2}(\mathit{tree},n) \;\; \vee$

\verb|     |$\exists x \quantsep \mathit{vartree3}(\mathit{tree},\mathit{env},x) \wedge \mathit{env}[x] = n \;\; \vee$

\verb|     |$\exists \mathit{op},l,r,n_l,n_r \quantsep \mathit{optree5}(\mathit{tree},\mathit{env},\mathit{op},l,r) \wedge \mathit{treeval3}(l,\mathit{env},n_l) \wedge \mathit{treeval3}(r,\mathit{env},n_r) \; \wedge$

\verb|       |$\mathit{op} = \ProgVirtual{\PL{ADD}} \rightarrow n = n_l + n_r \;\; \wedge ...$}

%\verb|       |$...$}

%\verb||

\verb|inv  |$\forall t,e \quantsep \mathit{parsetree2(t,e)} \rightarrow$

\verb|       |$t[0] = \ProgVirtual{\PL{VALUE}} \rightarrow \mathit{valuetree1(t)} \wedge t[0] = \ProgVirtual{\PL{VAR}} \rightarrow \mathit{vartree2(t,e)} \wedge t[0] = \ProgVirtual{\PL{OP}} \rightarrow \mathit{optree2(t,e)}$

\verb||

\verb|  |$\{ \mathit{parsetree2}(\ProgVirtual{\PL{tree}}, \ProgVirtual{\PL{env}}) \}$

\verb| 1  method calc(env, tree = @tree) {|

\verb| 2    if tree[0] = VALUE then | $\{\mathit{valuetree1(\ProgVirtual{\PL{tree}})}\}$ \verb| tree[1]|

\verb| 3    elseif tree[0] = VAR then | $\{\mathit{vartree2}(\ProgVirtual{\PL{tree}},\ProgVirtual{\PL{env}})\}$ \verb| env[tree[1]]|

\verb| 4    elseif tree[0] = OP then | $\{\mathit{optree2}(\ProgVirtual{\PL{tree}},\ProgVirtual{\PL{env}})\}$

\verb| 5      if tree[1] = ADD then calc(env, tree[2]) + calc(env, tree[3])|

\verb| 6      elseif ...|

\verb| 7      else nil|

\verb| 8      fi|

\verb| 9    else nil|

\verb|10    fi|

\verb|11  }|

\verb|  |{\color{blue}$\{ \mathit{treeval3}(\ProgVirtual{\PL{tree}}, \ProgVirtual{\PL{env}}, \ProgVirtual{\result}) \}$}
  \caption{\label{fig:correctness-ev-example} Correctness proof for the evaluator example}
\end{figure}
\todo{make sure the example matches its initially presented version}
\myparagraph{Type safety:} The given invariant enables deriving the assertions on lines 2-4 and hence a proper typing of the remaining program. As a property of the ad-hoc data structure it must be established in the (omitted) method \verb|parse()|. With the types of \verb|env| ($\mathbb{S} \mapsto \mathbb{N}$) and $\result$ ($\mathbb{N}$) known, their mapping can be automated. The complex ad-hoc data structure \verb|tree| is given the (imprecise) type $\mathbb{N} \mapsto \mathbb{O}$ and its elements hence need manual mapping. These mappings are encapsulated in predicates ($\mathit{valuetree2}, \mathit{vartree3}, \mathit{optree5}$) and furthermore ignored.

\myparagraphB{Correctness {\color{blue}(marked)}:} The lower layer information allows identifying numerous pure expressions (\verb|tree[1]|, \verb|env[tree[1]]|, \verb|tree[1] == ADD|, etc.). Establishing the specified property for the first two branches then only requires applying PURE EXPR. The conditional in line 5 can be handled by PURE COND. Since all arguments to the recursive method calls in that line are also pure, $\{\mathit{optree2}(\ProgVirtual{\PL{tree}},\ProgVirtual{\PL{env}}) \wedge \ProgVirtual{\PL{tree}}[1] = \ProgVirtual{\PL{ADD}} \} \allowbreak \verb| calc(env, tree[2]) + calc(env, tree[3])| \allowbreak \{\mathit{treeval3(\ProgVirtual{\PL{tree}},\ProgVirtual{\PL{env}},\result)}\}$ can be derived automatically. Note that all implications can be handled by SMT solvers with theories for Presburger arithmetic and lists for which effective decision procedures are known \todo{cite them} and do not require reasoning about graph-like object structures \todo{check}.

%% file: sections/related_work.tex
% Automated \cite{Cartwright91softTyping,NguyenTH2013SoftContract,Guha2011TypingLocalControl,Chugh2012NestedRef},
% Semi-Automated \cite{SiekTaha2007,BiermanAbadi2014UTypeScript,Chugh2012DJS,Swamy2013DijMonad},
% Non-Automated \cite{Lerner2013TeJaS,Gardner12programlogicJS,QinEtAl2011}
% 
% Type Safety \cite{Cartwright91softTyping,SiekTaha2007,BiermanAbadi2014UTypeScript,Guha2011TypingLocalControl},
% Some Properties \cite{Lerner2013TeJaS},
% Correctness in General \cite{NguyenTH2013SoftContract,Chugh2012NestedRef,Chugh2012DJS,Swamy2013DijMonad,Gardner12programlogicJS,QinEtAl2011}
% 
% Imperative,
% Functional
% 
% Symbolic \cite{Chugh2012NestedRef,Chugh2012DJS,Swamy2013DijMonad,Gardner12programlogicJS,QinEtAl2011},
% Abstraction-Based \cite{Cartwright91softTyping,SiekTaha2007,BiermanAbadi2014UTypeScript,Guha2011TypingLocalControl,Lerner2013TeJaS}
% Combination \cite{NguyenTH2013SoftContract}
% 
% Soundness 

There are several threads of related work regarding dynamically typed programs. In each, we can only discuss those works most closely related to ours.
% - Type Safety
%   - soft typing
%   - gradual typing
% - Verification of dynamically typed programs
%   - automatic verifiers
%   - program logics
% - Combining Static Analysis with Program Logics

% For quite some time, formal methods targeting
% dynamically typed languages have been scarce. With the notable exception of Boyer and Moore's ACL2 system \cite{BoyerMoore73}
% % and NuPRL ?!?
% almost all approaches developed were tailored to languages exhibiting a static type system.
% 
% ACL2 uses the language LISP itself as an assertion language, which is a very elegant and
% extensible approach but unfortunately seems only applicable to terminating subsets of functional
% (side-effect-free) languages.

\myparagraphB{Type Safety:}
%Out of numerous approaches for bringing type safety guarantees to dynamically typed languages, we will highlight only the most related:
Cartwright \cite{Cartwright91softTyping} pioneered a strand of work called ``soft typing'', applying automated type safety verifiers to dynamically typed languages with the aim of improving performance.
Another line of work is ``gradual typing'' \cite{SiekTaha2007,BiermanAbadi2014UTypeScript}, letting the user decide which parts of the program should be statically checked for type errors, while dynamically typing the remaining program. \footnote{The two may also be combined \cite{Rastogi2012GradualTypeInference}}

Both soft- and gradual typing require rewriting of program parts in exchange for type safety guarantees. On the contrary, our approach is able to provide such guarantees also for parts that are not statically typable. Also, the user is free to omit type safety from the specification (dynamic typing) and may still rewrite the program to allow automatic checking (static typing), both on a per-expression basis (gradual). With respect to correctness, our approach hence subsumes both soft and gradual typing. However, it does not (yet) increase performance.

% These systems, however, were never extended beyond establishing the preconditions of primitive operations \cite[Section~5]{NguyenTH2013SoftContract}. In the construction presented in this paper, sound approaches of this kind fit well into the role of the automated type safety verifier.
%- Typed Scheme, Typed Racket (subsumed by Soft Contract Verification?)
%- Static Type Inference for Ruby
%? GateKeeper

%\todo{what is the difference between earlier work, Henglein and gradual typing?}
%\todo{what about other approaches like Hybrid Typing, Occurrence Typing, Like Typing?}

Others \cite{Guha2011TypingLocalControl,Lerner2013TeJaS}, have extended such abstraction-based verifiers to handle many ideoms common in dynamically typed languages.

% Guha et al. \cite{Guha2011TypingLocalControl} show that a type system and a flow analysis can be fruitfully combine to reason about local control and state. Many idioms commonly found in JavaScript and other dynamically typed languages can be handled this way. Unfortunately, the authors defer type inference as future work and hence there is no automated type safety verifier available for us to extend.
% 
% Another example is TeJaS \cite{Lerner2013TeJaS}, a traditional, but parameterizable type system for JavaScript. While TeJaS as a whole is known to be undecidable, a decidable subset of it could form the basis of an
% automatic type safety verifier.

\myparagraphB{Correctness:}
%\myparagraphB{Program Logics:}
To the best of our knowlegdge, \cite{Gardner12programlogicJS} currently is the only\footnote{\cite{QinEtAl2011} only treat partial correctness. Also, they restrict the programming language to allow a form of (type-unsafe) pure expressions}
axiomatic semantics for a type-safe notion of correctness of a dynamically typed language.
Like discussed in Section~\ref{sec:overview:stat-dyn-hoare-logic} it uses type safety preconditions, considers all variables to be of object type and does not use pure expressions and would thus benefit from our approach.
\todo{find another \HLd}

%\myparagraphB{Automated Program Verifiers:}
% Nguyến 
Nguyen
%\todo{there are 2 accents over the e}
et al. \cite{NguyenTH2013SoftContract} proposed an automatic contract verifier for untyped higher-order functional languages based on symbolic execution inserting run-time checks for contracts it cannot statically guarantee. Since they use a mechanism similar to widening to enforce termination, their apporach also combines abstraction-based and symbolic reasoning.
%Unfortunately it only works in functional languages.
%\todo{can this not easily be turned into a semi-automatic procedure?}

Drawing on their work on the verification of untyped higher-order functional programs \cite{Chugh2012NestedRef}, Chugh et al. \cite{Chugh2012DJS,Chugh2013PhD} provide a dependent-type system for an untyped functional ``core calculus'' $\lambda_{JS}$ JavaScript programs can be translated into. No soundness is demonstrated for their system.
%\todo{type safety preconditions? all variables of type object? pure expressions?}

%- DijMonad <<<
Swamy et al. \cite{Swamy2013DijMonad} semi-automatically reason about a wide range of JavaScript ideoms by translating into the dependently-typed functional language $F^*$ and using its SMT-based reasoning engine. %Their approach is mainly symbolic and hence requires method contracts and loop invariants to be given manually.
They also noticed that the type information generated by an {ab\-strac\-tion}-based type safety verifier (GateKeeper in their case) are useful to improve the effectiveness of automatic reasoning engines. However, they did not feed the symbolically derived proof results back into GateKeeper and did not use the type information to ease the annotation burden for their users.
Since their main focus lies on a novel encoding of Dijkstra's predicate transformer semantics allowing $F^*$'s dependent type inference to effectively reason about imperative programs in a style similar to Hoare logic, we consider the approaches to be largely complementary.

In general, all fully automatic approaches \cite{Cartwright91softTyping,NguyenTH2013SoftContract,Guha2011TypingLocalControl,Chugh2012NestedRef} are necessarily incomplete. They can however be used as automatic type safety verifiers.
Furthermore, all purely symbolic approaches \cite{Chugh2012NestedRef,Chugh2012DJS,Swamy2013DijMonad,Gardner12programlogicJS,QinEtAl2011} require all type information to be manually specified in method contracts and loop invariants.

%\todo{necessary?}
%Also, while our focus lies on providing a convenient environment for manual verification steps,
%Swamy et al. did not even provide a way to annotate invariants in JavaScript programs.
%Note that removing the necessity for the user to understand the entire verification
%tool chain can be quite tricky in such translation-based approaches: f.i. DJS requires while loops to be
%annotated with Heap-type annotations as they are translated into functions.

% Recently there have been numerous approaches to verifying dynamically typed languages. We categorize them as follows:
% 
% 1) fully automatic:  -> necessarily incomplete \todo{cite Nested Refinements}
%   -> (potentially) useful as an automated type safety verifier in our approach
% 
%   2) (potentially) complete -> require manual effort to ensure type safety of programs that could be automatically
%   handled by abstraction-based approaches. \todo{check: TeJaS, DJS, Dijkstra-Monad}

\todo{add Dminor (Functional untyped languages)}

Both the idea and the term ``Layer of abstraction'' are
inspired by the work of Gardner et al. \cite{Gardner12programlogicJS} % \todo{cite Towards a Program Logic for JavaScript \& JuS} 
on reasoning about JavaScript. However, their work abstracts from the peculiarities of
the JavaScript variable store, while ours abstracts from the complexity of dynamic typing and
is applicable to virtually any dynamically typed language.

The decomposition rule used to establish the layer of abstraction is inspired by similar constructions in \cite{AptDeBoerOlderogBook2009}.

Some tools for verification of statically typed imperative programs \cite{DarvasLeino2007} allow using a ``pure'' subset of the programming language (that is side-effect-free and guaranteed to terminate) within assertions. The ability of our layer of abstraction to allow the use of well-typed ``pure'' program expressions in assertions can be seen as an extension of this idea to dynamically typed programs.

\myparagraphB{Combining Static Analysis with Program Logics:}
There has been a considerable amount of work on integrating algorithmic decision procedures (mostly
model checking) and deductive methods for program verification (See \cite{Uribe2000} for pointers).
Due to the deep connection between data flow analysis and model checking
\cite{SchmidtSteffen1998}, many of these techniques can be considered as related.

Note that our conjunctive refinement differs from abstraction refinement since it is not the abstraction that is refined, but the analysis result.

Also, translations from typings (f.i. from type systems for information-flow properties) to program logics are commonly used in the Proof-Carrying Code Community \cite{Hamid2004Interfacing} to avoid the need for property-specific proof-checkers. Although PCC is a completely different application area, their aim was also to integrate results derived by different inference systems into one common representation -- and incidently they also chose a program logic as their ``lingua franca''.

A closely related proposal also integrating symbolic with abstraction-based reasoning is MIXY\cite{PhangKhoo2010Mixing}, a framework for mixing symbolic execution with type checking. In their system, the user divides his/her program into s-blocks and t-blocks. While s-blocks are analysed using symbolic execution, type analysis is applied to t-blocks. The results of both analyses are bidirectionally exchanged using so-called MIX-rules: Type analysis results are translated into a matching start environment for symbolic execution and types ensured by (exhaustive!) symbolic execution can be used for type analysis. Also, the aim is related: What Phang et. al called ``balancing precision vs. efficiency'' is the same as ``combining automation with completeness'', although Phang et al. do not proof their system complete. Our approach could most likely be integrated into their framework as ``Hoare-Logic blocks'' $\{_h \; e \; {}_h\}$ with typing $\varGamma \vdash \{_h \; e \; {}_h\}: \tau$ for which a Hoare-triple $\{p_e\} e \{q_e\}$ must be derived where $p_e \defined \typingassertion_X(\varGamma)$ and $\projection^\result_X(q_e) = \tau$ for some verifier $\Verifier_X$.

%% file: sections/conclusion.tex
The approach presented allows verifying dynamically typed programs just like statically typed ones, requiring manual assistance only on hard typing problems. If a program is statically-typable, there is no difference. Otherwise, the user may freely choose whether going beyond the limits of the type system is worth the verification effort.
While the stated requirements for automated verifiers allow conveniently using the technique, more powerful verifiers can be expected to significantly increase the degree of automation.
Being gradually applicable like gradual typing and automated like soft typing, the approach allows deriving type safety guarantees also for code parts that cannot be startically typed.
%Offering formal proof as a third option besides modifying the program (static typing) and accepting the possibility of runtime type errors (dynamic typing), the technique integrates well with concepts like pluggable type systems.
Should it turn out to be practically usable, it would suggest dynamic typing as a serious alternative to static typing for verifiable languages that have all necessary infrastructure readily available. Until then, there are several useful ways to extended this work:
\myparagraph{Completeness:}
Currently, there is no (relative) complete Hoare logic for a dynamically typed language. While the approach is also applicable to incomplete program logics, no completeness guarantee can be provided in this case.
\myparagraph{Other Program Logics:}
\cite{Chugh2012DJS,Swamy2013DijMonad} are both based on refinement types. Recently, it has been shown \cite{Unno2013Automated} how to extend such systems to provide (relative) completeness like Hoare logic. It would be interesting to investigate if our approach is also applicable to such program logics.
\myparagraph{Other Program Analyses:} The formalization of semi-automation suggests that it could be generalizable to arbitrary data flow analyses.
\myparagraph{Performance:} The derived type information could be used to omit run-time checks and generate more efficient binaries.
\myparagraph{Features:}
The current concept excludes optional variables and type errors as exceptions. Also, closures and advanced dynamic features like method update, dynamic type hierarchies and eval() should be studied.
\myparagraph{Implementation:}
An implementation would allow evaluating the practical usefulness of the concept.
\todo{functional programs: type inference + refinement types}
\todo{extend to other program analyses (termination or alias/points-to-analysis)}
\todo{use the precise type information also to improve performance by eliminating runtime type checks. Also: intervals can be used to choose the right representations of numerals, fixed-size lists / strings can be represented as arrays,...}
%
%
%%TODO{make parallel}
%%TODO{TI algorithms supporting small proofs}
%
%
\paragraph{Acknowledgements}
\ifx\submission\true
-- omitted for submission --
\else
We thank Sven Linker, Martin Hilscher and Eike Best for insightful discussions and useful comments on drafts of this paper.
\fi

%% file: appendices/app_pure_expressions.tex
\paragraph{Identifying pure expressions} ($\ppure: \mathcal{T}_s \times \mathit{Expr_d} \mapsto \mathbb{B}$)

$\ppure(e) \defined \tau(e) \text{ defined}$, \hfill $\tau(e) = \mathit{inf}(\{ \BaseTypeX \mid \ppure(\BaseTypeX,e) \})$

$\ppure(\mathit{Null}, \PL{null}) \defined \mathit{true}$, \hfill $\forall e \in \mathit{Expr}_d \quantsep \ppure(\mathit{Null}, e) \rightarrow \ppure(\mathbb{O}, e)$

$\ppure(\BaseTypeX, \VarLX) \defined \llbracket \VarLX \rrbracket \in \UnionTypeX \wedge \Psi(\UnionTypeX) = \BaseTypeX \wedge \safe_\BaseTypeX(\VarLX)$ (includes the case of $\this{}$)

$\ppure(\BaseTypeX, \VarIX) \defined \llbracket \this{}.\VarIX \rrbracket \in \UnionTypeX \wedge \Psi(\UnionTypeX) = \BaseTypeX \wedge \safe_\BaseTypeX(\VarIX)$

$\ppure(\mathbb{B}, e \PL{ is\_a? C}) \defined \ppure(\mathbb{O}, e)$

$\ppure(\mathbb{B}, e_1 == e_2) \defined \exists \BaseTypeX \in \mathcal{T}_s \quantsep \ppure(\BaseTypeX, e_1) \wedge \ppure(\BaseTypeX,e_2)$

$\ppure(\BaseTypeX, e_0.\PL{m}(e_1,...,e_n)) \defined$

\hspace{1cm} $\exists \BaseTypeX_0,...,\BaseTypeX_n \quantsep \ppure(\BaseTypeX_i, e_i)$ for $i \in \mathbb{N}_n \wedge \Psi(\BaseTypeX_0.\PL{m}(\BaseTypeX_1,...,\BaseTypeX_n) \rightarrow \BaseTypeX)$ defined

$\ppure(\BaseTypeX, \PL{new C}(e_1,...,e_n)) \defined$

\hspace{1cm} $\exists \BaseTypeX_1,...,\BaseTypeX_n \quantsep \ppure(\BaseTypeX_i, e_i)$ for $i \in \mathbb{N}^1_n \wedge \Psi(\Psi(\{C\}).\PL{init}(\BaseTypeX_1,...,\BaseTypeX_n) \rightarrow \BaseTypeX)$ defined

$\ppure(\BaseTypeX, \PL{u} := e) \defined \mathit{false}$

$\ppure(\BaseTypeX, \PL{@v} := e) \defined \mathit{false}$

$\ppure(\BaseTypeX, \PL{while } \ExprX \PL{ do } \StmtX \PL{ od}) \defined \mathit{false}$

$\ppure(\BaseTypeX, \PL{if } \ExprX \PL{ then } \StmtX_1 \PL{ else } \StmtX_2 \PL{ fi}) \defined \ppure(\mathbb{B},\ExprX) \wedge \StmtX_i \equiv \ExprX_i \wedge \ppure(\BaseTypeX, \ExprX_i)$ for $i \in \{1,2\}$

\paragraph{Translation of pure expressions into logical expressions} ($\Psi: \mathcal{T}_s \times \mathit{Expr}_d \mapsto \mathit{LExp}$)

$\Psi(e) \defined \Psi_{\tau(e)}(e)$, $\Psi_\BaseTypeX(\VarX) \defined \VarXVirtual$

$\Psi_{\mathit{Null}}(\PL{null}) \defined \PL{null}$, $(\Psi_{\mathit{Null}}(e) = p) \rightarrow (\Psi_{\mathbb{O}}(e) = p)$

$\Psi_\BaseTypeX(e_0.\PL{m}(e_1,...,e_n)) \defined l[v_0,...,v_n := \Psi_{\BaseTypeX_0}(e_0), ..., \Psi_{\BaseTypeX_n}(e_n)]$

where $\ppure(\BaseTypeX_i, e_i)$ for $i \in \mathbb{N}_n$ and $\Psi(\BaseTypeX_0.\PL{m}(\BaseTypeX_1,...,\BaseTypeX_n) \rightarrow \BaseTypeX) = l$.

$\Psi_\BaseTypeX(\PL{new C}(e_1,...,e_n)) \defined l[v_1,...,v_n := \Psi_{\BaseTypeX_1}(e_1), ..., \Psi_{\BaseTypeX_n}(e_n)]$

where $\ppure(\BaseTypeX_i, e_i)$ for $i \in \mathbb{N}^1_n$ and $\Psi(\Psi(\{C\}).\PL{init}(\BaseTypeX_1,...,\BaseTypeX_n) \rightarrow \BaseTypeX) = l$.

$\Psi_{\mathbb{B}}(e_1 == e_2) \defined \Psi_\BaseTypeX(e_1) = \Psi_\BaseTypeX(e_2)$ where $\ppure(\BaseTypeX, e_i)$ for $i \in \{1,2\}$.

$\Psi_{\mathbb{B}}(e \PL{ is\_a? C}) \defined \llbracket \Psi_{\mathbb{O}}(e) \rrbracket \in \{\PL{C}\}$

$\Psi_\BaseTypeX(\PL{if } e \PL{ then } e_1 \PL{ else } e_2 \PL{ fi}) \defined \PL{if } \Psi_{\mathbb{B}}(e) \PL{ then } \Psi_\BaseTypeX(e_1) \PL{ else } \Psi_\BaseTypeX(e_2) \PL{ fi}$

\paragraph{Proof Rules for pure expressions}
\hspace{0.1cm}\newline

\noindent RULE: PURE LOOP (\StrongPC{strong} \TypesafePC{type-safe} partial correctness)
\begin{center}
\AxiomC{$\{p \wedge \Psi(e) \} S \{p\}$}
\RightLabel{where $\ppure(e)$, $\tau(e) = \mathbb{B}$}
\UnaryInfC{$\{p\}$ while $e$ do $S$ od $\{p \wedge \neg \Psi(e) \wedge \result = \Vnull\}$}
\DisplayProof
\end{center}

\noindent RULE: PURE METH
\begin{center}
\AxiomC{$\begin{matrix}
%          \{\AsrtX_i\} \ExprX_i \{\AsrtX_{i+1}[\VarLX_i := \result] \}\text{ for }i \in \mathbb{N}_n \\
          \{\AsrtX\} \VarLX_0.\MethodX(\VarLX_1,...,\VarLX_n) \{\AsrtY\}
         \end{matrix}$}
       
\UnaryInfC{$\{\AsrtX[\VarLXVirtual_0,...,\VarLXVirtual_n := \Psi(\ExprX_0),...,\Psi(\ExprX_n)]\} \ExprX_0.\MethodX(\ExprX_1,...,\ExprX_n) \{\AsrtY\}$}
\DisplayProof
\end{center}
where $\VarLX_i \in \Var_L$ fresh and $\ppure(e_i)$ for all $i \in \mathbb{N}_n$.

\paragraph{Soundness}

\begin{proof}
 The Axiom PURE EXPR can be established by induction over the structure of the pure expression $\ExprX$, using the guarantees provided by $\ppure(\ExprX)$. In the cases for variables, $\ppure(\VarX)$ implies $\mathit{safe}_\BaseTypeX(\VarX)$ for some type $\BaseTypeX$. In the case for method calls, we assume
 \[ \{\AsrtX[\hat{\result} := \Psi(\BaseTypeX_0.\MethodX(\BaseTypeX_1,...,\BaseTypeX_n)\rightarrow \BaseTypeX]\} \VarLX_0.\MethodX(\VarLX_1,...,\VarLX_n) \{\AsrtX\} \]
 with $\BaseTypeX_i = \tau(\VarLX_i)$ for all $i \in \mathbb{N}_n$ to be established for all methods in $\Method_{\pure}$ (which is precisely the meaning of ``correspondence between methods and operations with respect to the mapping $\Psi$'' in Section~\ref{sec:layer:pure_expressions}).

 The rules PURE ASGN, PURE COND, PURE LOOP and PURE METH can then be derived by combining the axiom PURE EXPR with the \textbf{dyn} rules for ASGN, COND, LOOP and METH respectively. \hfill $\qed$
\end{proof}

%% file: appendices/app_hoare_logic.tex
\todo{adapt rules to new naming of syntax elements}

\vspace{0.1cm}
\noindent AXIOM: CONST

\begin{center}
$\{p[\result := \Vnull]\} \PL{null} \{p\}$
\end{center}

\noindent AXIOM: VAR

\begin{center}
$\{p[\result := \PL{u}]\} \PL{u} \{p\}$
\end{center}

Note: includes the case of $\PL{u} \equiv \this{}$.

\noindent AXIOM: IVAR

\begin{center}
$\{p[\result := \this{}\PL{.@v}]\} \PL{@v} \{p\}$
\end{center}

\noindent RULE: ASGN (both normal and instance variables)
\begin{center}
\AxiomC{$\{p\} e \{q[\PL{u} := \result] \}$}
\RightLabel{where $\PL{u} \in \Var$}
\UnaryInfC{$\{p\} \PL{u} := e \{q\}$}
\DisplayProof
\end{center}
\ifx \version\extended
\noindent RULE: SEQ

\begin{center}
\AxiomC{$\{p\} S_1 \{r\}$}
\AxiomC{$\{r\} S_2 \{q\}$}
\BinaryInfC{$\{p\} S_1; S_2 \{q\}$}
\DisplayProof
\end{center}
\fi

\noindent RULE: Conditional (COND) (\StrongPC{strong} \TypesafePC{type-safe} partial correctness)

  \begin{center}
  \AxiomC{$\begin{matrix}
	    \{p\} e \{r \wedge \mathit{bool\_test} \} \\
	    \{r \wedge b\} S_1 \{q\} \\
	    \{r \wedge \neg b\} S_2 \{q\}
	  \end{matrix}$}
  \UnaryInfC{$\{p\}$ if $e$ then $S_1$ else $S_2$ fi $\{q\}$}
  \DisplayProof
  \end{center}
\hspace{0.1cm}
where $b$ is a predicate and $\mathit{bool\_test} \equiv \mStrongPC{\result \not= \Vnull} \wedge \mTypesafePC{\llbracket \result \rrbracket \in \{bool\}} \wedge \mathbb{B}(\result,b)$

\noindent RULE: LOOP (\StrongPC{strong} \TypesafePC{type-safe} partial correctness)
\begin{center}
\AxiomC{$\begin{matrix}
          \{p\} e \{p \wedge r \wedge \mathit{bool\_test} \} \\
          \{p \wedge r \wedge b \} S \{p\}
         \end{matrix}$}
\UnaryInfC{$\{p\}$ while $e$ do $S$ od $\{p \wedge \neg b \wedge \result = \Vnull\}$}
\DisplayProof
\end{center}
where $b$ is a predicate and $\mathit{bool\_test}$ is defined as in COND.

\noindent RULE: CONS

\begin{center}
\AxiomC{$p \rightarrow p_1, \{p_1\} S \{q_1\}, q_1 \rightarrow q$}
\UnaryInfC{$\{p\} S \{q\}$}
\DisplayProof
\end{center}

\noindent RULE: PASGN
\begin{center}
$\{p[\Vector{\PL{u}} := \Vector{t}]\} \Vector{\PL{u}} := \Vector{t} \{p\}$
\end{center}
where $\{\Vector{\PL{u}}\} \subseteq \Var_L$ and $\{\Vector{t}\} \subseteq \Var_L \cup \{\PL{null}\}$.

\noindent RULE: BLCK
\begin{center}
\AxiomC{$\{p\} \Vector{\PL{u}}\Vector{\overline{\PL{u}}} := \Vector{t}\Vector{\PL{null}}; S \{q\}$}
\RightLabel{when $(\Var_L \setminus \{\result\}) \cap \free(q) = \emptyset$}
\UnaryInfC{$\{p\}$ begin local $\Vector{\PL{u}} := \Vector{t}; S \text{ end} \{q\}$}
\DisplayProof
\end{center}
where $\{\Vector{\PL{u}}\} \subseteq \Var_L$, $\{\Vector{t}\} \subseteq \Var_L \cup \{\PL{null}\}$,
$\{\Vector{\overline{\PL{u}}}\} = \Var_L \setminus (\{\Vector{\PL{u}}\} \cup \Var_S)$ and $|\Vector{\PL{null}}| = |\Vector{\overline{\PL{u}}}|$.

\todo{think about: either only allow local variables instead of $t$ in rules BLCK and PASGN or also allow ``terms'' in REC}

\noindent RULE: METH
\begin{center}
\AxiomC{$\begin{matrix}
          \{\AsrtX_i\} \ExprX_i \{\AsrtX_{i+1}[\VarLX_i := \result] \}\text{ for }i \in \mathbb{N}_n \\
          \{\AsrtX_{n+1}\} \VarLX_0.\MethodX(\VarLX_1,...,\VarLX_n) \{\AsrtY\}
         \end{matrix}$}
       
\UnaryInfC{$\{\AsrtX_0\} \ExprX_0.\MethodX(\ExprX_1,...,\ExprX_n) \{\AsrtY\}$}
\DisplayProof
\end{center}
where $\VarLX_i \text{ fresh}, \VarLX_i \in \Var_L, \VarLX_i \not\in var(\ExprX_j) \cup change(\ExprX_j)$ for all $i,j \in \mathbb{N}_n$.

\noindent RULE: REC
(\StrongPC{strong} \TypesafePC{typesafe} partial correctness)
\begin{center}
\AxiomC{$\begin{matrix}
          A \vdash \{p\} S \{q\}, \\
          A \vdash \{p_i\} \PL{begin local this}, \Vector{\PL{u}_i} := \PL{v}'_i,\Vector{\PL{v}_i}; S_i \text{ end} \{q_i\}, i \in \mathbb{N}^1_n \\
          p_i \rightarrow (\mStrongPC{\PL{v}'_i \not= \Vnull} \wedge \mTypesafePC{\PL{v}'_i \not= \Vnull \rightarrow \PL{v}'_i.@\mathbf{c} = \oC_{\PL{C}_i}}), i \in \mathbb{N}^1_n
         \end{matrix}$}
       
\UnaryInfC{$\{p\} S \{q\}$}
\DisplayProof
\end{center}
where method $\PL{m}_i(\Vector{\PL{u}_i}) \{ S_i \} \in \Method_{\PL{C}_i}$ and $A \equiv \{p_1\} \PL{v}'_1.\PL{m}_1(\Vector{\PL{v}_1}) \{q_1\}, ..., \{p_n\} \PL{v}'_n.\PL{m}_n(\Vector{\PL{v}_n}) \{q_n\}$.

% \noindent RULE: REC (update)\todo{remove when done}
% %
% (partial correctness \Optional{ \StrongPC{+strong} \TypesafePC{+typesafe}} )
% %
% \begin{center}
% \AxiomC{$I_1,...,I_k$}
% \AxiomC{$A_1,...,A_k \vdash F_0,...,F_k$}
% \BinaryInfC{$\{p\} S \{q\}$}
% \DisplayProof
% \end{center}
% where
% \begin{itemize}
%  \item $A_i \equiv \{p_i[\Vector{u_i} := \Vector{v_i}][\Vector{f_i} := \Vector{z_i}]\} l_i.m_i(\Vector{v_i}) \{r_i\}$
%  \item $F_0 \equiv \{p\} S \{q\}$
%  \item $F_i \equiv \{p_i\} S_i \{q_i[\result := \mathit{return}_i]\}$ for $i \in \mathbb{N}^1_k$
%  \item $I_i \equiv q_i[\Vector{f_i} := \Vector{z_i}] \rightarrow r_i[\mathit{return}_i := \result]$
% \end{itemize}
% 
% 
% 
\noindent RULE: CNSTR
\begin{center}
\AxiomC{$\{p\} \mathbf{new}_\PL{C}.init(\Vector{e}) \{q\}$}
\UnaryInfC{$\{p\} \PL{new C}(\Vector{e}) \{q\}$}
\DisplayProof
\end{center}

\noindent AXIOM: NEW
\begin{center}
$\{p[\result := \mathbf{new}_\PL{C}]\} \mathbf{new}_\PL{C} \{p\}$
\end{center}

\ifx \version\extended
\subsection{Auxiliary Rules}
\label{sec:aux_rules}

\noindent RULE: DISJ

\begin{center}
\AxiomC{$\{p\} S \{q\}$}
\AxiomC{$\{r\} S \{q\}$}
\BinaryInfC{$\{p \vee r\} S \{q\}$}
\DisplayProof
\end{center}

\noindent RULE: CONJ

\begin{center}
\AxiomC{$\{p_1\} S \{q_1\}$}
\AxiomC{$\{p_2\} S \{q_2\}$}
\BinaryInfC{$\{p_1 \wedge p_2\} S \{q_1 \wedge q_2\}$}
\DisplayProof
\end{center}

\noindent RULE: $\exists$-INT

\begin{center}
\AxiomC{$\{p\} S \{q\}$}
\UnaryInfC{$\{\exists x. p\} S \{q\}$}
\DisplayProof
\end{center}
where $x \not\in var(\Method) \cup var(S) \cup \free(q)$.

\noindent RULE: INV

\begin{center}
\AxiomC{$\{r\} S \{q\}$}
\UnaryInfC{$\{p \wedge r\} S \{p \wedge q\}$}
\DisplayProof
\end{center}
where $\free(p) \cap (change(\Method) \cup change(S)) = \emptyset$ and $p$ does not contain quantification
over objects.
\todo{explain $change(\Method) \subseteq \Var_I$}

\noindent RULE: SUBST
\todo{think about: are the $z$ local/instance variables?}

\begin{center}
\AxiomC{$\{p\} S \{q\}$}
\UnaryInfC{$\{p[\Vector{z} := \Vector{t}]\} S \{q[\Vector{z} := \Vector{t}]\}$}
\DisplayProof
\end{center}
where $var(\Vector{z}) \cap (var(\Method) \cup var(S)) = var(\Vector{t}) \cap (change(\Method) \cup change(S)) = \emptyset$.
\fi
% weakening rule is only for partial correctness, not for strong partial correctness

%% file: appendices/app_VEx.tex
To allow for intra-procedural flow-sensitivity, all statements $S$ are converted to static single-assignment form (SSA) \todo{reference} for local- as well as instance variables of the current object. This necessitates that each occurrence of such a variable $\PL{x}$ having some number of assignments, say $n$, is replaced by one of its $k \ge n$ ``versions'' $\PL{x}_1, ..., \PL{x}_k$ such that each version has exactly one assignment dominating all its occurrences (except $\phi$-occurrences\footnote{for $k - n$ versions, so-called $\phi$-assignments are inserted at control-flow joins. The occurrences in such assignments need not be dominated.}). We maintain a mapping $\Fversion(\PL{x},L)$ from variables $\PL{x} \in \Var$ and program locations $L$ of $\pi$ to the version $\PL{x}_i$ whose assignment dominates this location.

Next, each sub-statement $S$ of the program $\pi$ is given a type variable $\llbracket S \rrbracket$. For each
method $\PL{C.m}$, additional type variables $\llbracket P^{\PL{C.m}}_i \rrbracket$ for each of its parameters, $\llbracket R^{\PL{C.m}} \rrbracket$ for its return value and $\llbracket V^{\PL{C.m}}_{\PL{u},i} \rrbracket$ for each version $\PL{u}_i$ of each local variable $\PL{u}$ are added. Instance variables $\PL{@v} \in \Var_{\PL{C}}$ are given both a global type variable $\llbracket V_{\PL{C.@v}} \rrbracket$ and  type variables $\llbracket V^{\PL{C.m}}_{\PL{C.@v},i} \rrbracket$ for each version $\PL{@v}_i$ in methods $\PL{C.m}$ using them.

The constraint generation then proceeds as follows: for every method declaration $\PL{method m}(\PL{u}_1,...,\PL{u}_n) \{ S_{\PL{C.m}} \}$ of class $\PL{C}$ we generate the following constraints:

\vspace{0.2cm}
\begin{tabular}{lcl}
$\llbracket P^{ \PL{C.m}}_i \rrbracket \constraint \llbracket \PL{u}_i \rrbracket$, & \hspace{1cm} $\llbracket S_{\PL{C.m}} \rrbracket \constraint \llbracket R^{\PL{C.m}} \rrbracket$, \hspace{1cm} & $\llbracket V_{\PL{C.@v}} \rrbracket \constraint \llbracket V^{\PL{C.m}}_{\PL{C.@v},0} \rrbracket$ \\
\hspace{0.5cm} for all $i \in \mathbb{N}^1_n$ & & \hspace{0.5cm} for all $\PL{@v}$ used in $\PL{C.m}$ \\
\end{tabular}
\vspace{0.2cm}

\noindent Additionally, we traverse the parse tree of the body $S_{\PL{C.m}}$ of the method $\PL{C.m}$ applying the rules given in Figure~\ref{fig:typing-rules}.

\begin{figure}
In the body of a method $\PL{C.m}$:

\begin{tabular}{rllrll}
$\Vnull$ & $\Longrightarrow$ & $\llbracket \Vnull \rrbracket \constraintequal \{\}$ & $\PL{u}_i := e$ & $\Longrightarrow$ & $\llbracket e \rrbracket \constraint \llbracket \PL{u}_i \rrbracket \wedge \llbracket \PL{u}_i := e \rrbracket \constraintequal \llbracket e \rrbracket$ \\

$\PL{u}_i$ & $\Longrightarrow$ & $\llbracket \PL{u}_i \rrbracket \constraintequal \llbracket V^{C.m}_{\PL{u},i} \rrbracket$ & $\PL{@v}_i := e$ & $\Longrightarrow$ & $\llbracket e \rrbracket \constraint \llbracket \PL{@v}_i \rrbracket \constraint \llbracket V_{\PL{C.@v}} \rrbracket \wedge \llbracket \PL{@v}_i := e \rrbracket \constraintequal \llbracket e \rrbracket$ \\

$\PL{@v}_i$ & $\Longrightarrow$ & $\llbracket \PL{@v}_i \rrbracket \constraintequal \llbracket V^{\PL{C.m}}_{\PL{C.@v},i} \rrbracket$ \hspace{0.5cm} & $e.\PL{m}(e_1, ... ,e_n)$ & $\Longrightarrow$ & \precision{$\llbracket e \rrbracket \constraint [\PL{m}(\llbracket e_1 \rrbracket,...,\llbracket e_n \rrbracket) \rightarrow \llbracket e.\PL{m}(...) \rrbracket]$} \\

$\this{}$ & $\Longrightarrow$ & $\llbracket \this{} \rrbracket \constraintequal \{\PL{C}\}$ & $S; e$ & $\Longrightarrow$ & $\llbracket S; e \rrbracket \constraintequal \llbracket e \rrbracket$ \\
\end{tabular}

\begin{tabular}{rll}
if $e$ then $S_1$ else $S_2$ fi & $\Longrightarrow$ & $\precision{\llbracket e \rrbracket \constraint \{\mathit{bool}\}}\wedge \llbracket S_1 \rrbracket \constraintunion \llbracket S_2 \rrbracket \constraint \llbracket \PL{if }...\PL{ fi}\rrbracket$ \\

while $e$ do $S$ od & $\Longrightarrow$ & $\precision{\llbracket e \rrbracket \constraint \{\mathit{bool}\}} \wedge \llbracket $while $...$ od$\rrbracket = \{\}$ \\

new $C'(e_1,...,e_n)$ & $\Longrightarrow$ & $\{C'\} \constraint [\PL{init}(\llbracket e_1 \rrbracket, ..., \llbracket e_n \rrbracket)\rightarrow \llbracket $new $ C'(...) \rrbracket]$ \\

$\phi(e_1,...,e_n)$ & $\Longrightarrow$ & $\llbracket e_1 \rrbracket \sqcup ... \sqcup \llbracket e_n \rrbracket \constraint \llbracket \phi(e_1,...,e_n) \rrbracket$ \\

$e \sqcap T$ & $\Longrightarrow$ & $\llbracket e \rrbracket \sqcap T \constraint \llbracket e \sqcap T \rrbracket$ \\
\end{tabular}

%\noindent where $V_1 \constraintequal V_2$ is a shorthand for $V_1 \constraint V_2 \wedge V_2 \constraint V_1$.

\caption{\label{fig:typing-rules} $\VEx$ typing rules for \textbf{dyn}}
\end{figure}
\todo{maybe only show rules for assignments, method call and conditional}
\todo{think about a a marking for precision-constraints that works when printed black-and-white}
\todo{add $e == e$ and $is\_a?$}
\todo{add $\psi(\PL{u}_1,\PL{u}_2)$}

\todo{a little too informal}
The resulting system of set inclusion constraints can be solved by propagating the type information from constructor calls forwards along the data flow and whenever a type $\PL{C}$ reaches a type variable of the form $[\PL{m}(T_1,...,T_n)\rightarrow R]$ (generated by a method call), the method $\PL{C.m}$ of arity $n$ is looked up and the following connection constraints are added

\vspace{0.2cm}
$\llbracket T_i \rrbracket \constraint \llbracket P^{C.m}_i \rrbracket$ for all $i \in \mathbb{N}^1_n,\hspace{0.5cm} \llbracket R^{C.m} \rrbracket \constraint \llbracket R \rrbracket$
\vspace{0.2cm}

Upon reching a fixpoint, the analysis provides a solution $\typingX$ (a typing) mapping each type variable (and thus every sub-statement, variable, parameter and return value) to a union type in $\mathcal{T}$. To mask the initial conversion to SSA we define $\typingX(\VarX) \defined \bigsqcup_i \typingX(\VarX_i)$ for all variables $\VarX$ tracked flow-sensitively.

A typing $\typingX$ is called \emph{consistent} iff for every constraint $\llbracket S_1 \rrbracket \sqsubseteq \llbracket S_2 \rrbracket$ it holds that $\typingX(\llbracket S_1 \rrbracket) \sqsubseteq \typingX(\llbracket S_2 \rrbracket)$. An important property of typing rules is for consistency to imply soundness. Note that the constraints {\precision{marked}} in Figure~\ref{fig:typing-rules} serve to ensure sufficient precision ($\typingX \sqsubseteq \smin_\pi$) rather than soundness. When omitting these constraints, the algorithm outputs a sound typing even when its precision is insufficient to establish type safety.

By intra-procedural path-sensitivity we mean that the algorithm maintains a set of alternatives $\mathit{path}(\locX)$ for each program location $\locX$ and for each $j \in \mathit{path}(L)$ separately derives the types of flow-sensitively tracked variables $\typingX(\Fversion(\VarX,\locX),j)$ as well as results of the previous sub-statement $\typingX(\llbracket S \rrbracket,j)$ with $L = \post{S}$. In $\VEx$, path sensitivity is only introduced by specially marked if$_{PS}$-conditionals duplicating the number of paths leading to them. Note that eliminating equivalent alternatives is important to keep the analysis terminating in the presence of loops\footnote{$\llbracket \VarX \rrbracket \in \{\ClassX_1\} \vee \llbracket \VarX \rrbracket \in \{\ClassX_2\}$ is not equivalent to $\llbracket \VarX \rrbracket \in \{\ClassX_1,\ClassX_2\}$! Otherwise, no assumption could type if $\PL{b}$ then $\VarX := ``foo''$ else $\VarX := 21$ end; $\VarX + \VarX$.}. However, as these markings are only used in the definition of trusted assumptions below, usually $|\mathit{path}(\locX)| = 1$ for all $\locX$.

% Draft for path sensitivity
% $\typingassertion(s,\post{S}) \defined \TInv_s \wedge \bigvee_{j \in \mathit{path}(\post{S})} \llbracket \result \rrbracket \in s(\llbracket S \rrbracket^j) \wedge \bigwedge_{\PL{x} \in \Var_{\PL{C.m}}} \llbracket \PL{x} \rrbracket \in s(\llbracket\Fversion(\PL{x},\post{S})\rrbracket^j)$

% Definition for join of typing assertions
% $f(\llbracket \PL{x} \rrbracket \in t) \defined h(y) = \begin{cases}
%                                                         t & \text{if } y = \PL{x} \\
%                                                         \top & \text{otherwise}
%                                                        \end{cases}
% $
% 
% $f(\neg \llbracket \PL{x} \rrbracket \in t) \defined f(\llbracket \PL{x} \rrbracket \in \top \setminus t)$
% 
% $f(\tau \wedge \tau') \defined f(\tau) \sqcap f(\tau')$
% 
% $f(\tau \vee \tau') \defined f(\tau) \sqcup f(\tau')$
% 
% 
% $(f \sqcap g)(x) \defined f(x) \sqcap g(x)$
% 
% $(f \sqcup g)(x) \defined f(x) \sqcup g(x)$
% 
% 
% \todo{check if these naming is correct}
% $\sqcup$ = join  
% 
% $\sqcap$ = meet   (what is meant by ``meet over all paths''?)
% 
% 
% $\tau \sqcup \tau' \defined $
% 
% %todo: negated typing literals
% 
% $\llbracket \PL{x} \rrbracket \in t \sqcup \llbracket \PL{x} \rrbracket \in t' \defined \llbracket \PL{x} \rrbracket \in t \sqcup t'$
% 
% $\llbracket \PL{x} \rrbracket \in t \sqcup \llbracket \PL{x'} \rrbracket \in t' \defined \llbracket \PL{x} \rrbracket \in t \wedge \llbracket \PL{x'} \rrbracket \in t'$
% 
% $\mu_1 \wedge ... \wedge \mu_n \sqcup \mu'_1 \wedge ... \wedge \mu'_{n'} \defined $

\myparagraphB{Logic to Typings}
As listed in the requirements in Section~\ref{sec:overview:semi-auto}, verifiers $\Verifier_X$ need an interface for supplying trusted assumptions. Abstractly, one defines a refinement relation $\typingX \stackrel{\TAsrtX,\locX}{\rightarrow} \typingY$ between $\Verifier_X$-typings. Here, $\typingY$ refines $\typingX$ by taking the trusted assumption $\TAsrtX$ at program location $\locX$ into account. For $\VEx$ we first extend \textbf{dyn} expressions by introducing a type filter operation $e ::= \UnionTypeX \sqcap e$ generating the (monotone) typing constraint $\UnionTypeX \sqcap \llbracket e \rrbracket \constraint \llbracket t \sqcap e \rrbracket$. By inserting type filters, it is possible to refine a typing assertion $\TAsrtX \equiv \typingassertion( \typingX,\locX)$ for a program location $\locX$ to $\TAsrtX \wedge \TAsrtX'$ for some asumption $\TAsrtX'$:
\begin{definition}[Refinement of Constraint Systems]
\renewcommand{\PL}[1]{\emph{#1}}
 Let $G$ be a constraint system generated by applying $\VEx$'s constraint generation to a program $\pi$. A \emph{conjunctive refinement step} of $G$ using the trusted assumption $\TAsrtX$ at program location $\locX$ is a quadruple $(G,\tau,\locX,G')$, written $G \Ex{\stackrel{\tau,L}{\rightarrow}} G'$ with $G'$ being the constraint system generated by applying $\VEx$'s constraint generation to the program $\programX$ with the marked conditional $\assertS(\TAsrtX)$ inserted just before program location $\locX$.
 For the definition of $\assertS(\TAsrtX)$, we assume without loss of generality $\TAsrtX \equiv \TAsrtKX_1 \vee ... \vee \TAsrtKX_n$ to be in disjunctive normal form with all conjunctions $\TAsrtKX_i$ mentioning each variable in at most one typing literal:
 
 $\assertS(\TAsrtKX \vee \TAsrtX) \defined$ \PL{if}$_{PS}$ $...$ \PL{then} $\assertS(\TAsrtKX)$ \PL{else} $\assertS(\TAsrtX)$ \PL{fi},
 
 $\assertS(\llbracket \VarX \rrbracket \in \UnionTypeX \wedge \TAsrtKX) \defined \assertS(\llbracket \VarX \rrbracket \in \UnionTypeX); \assertS(\TAsrtKX)$, \hfill $\assertS(\llbracket \VarX \rrbracket \in \UnionTypeX) \defined \VarX := \UnionTypeX \sqcap \VarX$

\end{definition}
Note that the condition does not matter as $\VEx$ regards conditionals as nondeterministic choice. In essence, if $\TAsrtX$ has $n$ disjuncts $\TAsrtKX_j$ then all paths reaching the marked conditional are split into $n$ paths and in each of them, the types $\llbracket \VarX_i \rrbracket$ of all variables $\VarX_i \in \mathit{free}(\TAsrtKX_j)$ are refined to $\llbracket \VarX_i \rrbracket \sqcap \projection_{\VarX_i}(\TAsrtKX_j)$ for program locations dominated by $L$.
\todo{is there a better way to explain this?}
\todo{maybe illustrate this with a picture}

%% file: appendices/app_typing_to_proof.tex
\todo{is it neccessary to introduce $\restricted$-Hoare logic?}

We will now show how to translate a given $\VEx$-typing $\typingX$ into a typing proof $\psi$ with $\typingX_\psi = \typingX$. Note that in contrast to type safety proofs this is possible for every sound typing $\typingX$.

% For a given typing $s$, we instantiate the global typing invariant $\TInv$ to reflect those types in
% $s$:
% 
% \todo{define $\TInv$ as a function from typings to typing assertions. Then this ``instantiation'' is unneccessary}
% \begin{definition}
%  \renewcommand{\PL}[1]{\emph{#1}}
%  The global typing invariant $\TInv_s$ is derived from $\TInv$ by instantiating all types $\BaseTypeX_{\PL{C.@v}}$ to the concrete types $s(V_{\PL{C.@v}})$ where $V_{\PL{C.@v}}$ denotes the type variable for the instance variable $\PL{@v}$ of class $\PL{C}$.
% \end{definition}

Recall from Section~\ref{sec:TI2HL} that for a $\VEx$-typing $\typingX$, the global typing invariant $\Ex{\TInv}(\typingX)$ states that the types $\typingX(\typevarIvarX)$ assigned to instance variables $\ClassX.\VarIX$ in $\typingX$ safely over-approximate the actual types of these variables for all runs of the program. Establishing $\Ex{\TInv}(\typingX)$ as an invariant of all method bodies and the main statement is thus an important step in constructing a typing proof. Fortunately, $\Ex{\TInv}(\typingX)$ can be shown to be invariant under most proof rules in our logic. The only complicated cases are the rules for assignment to instance variables and object creation. For these cases, we introduce new rules that explicitly preserve the global typing invariant:

\noindent RULE: $\restricted$-IASGN
\begin{center}
\AxiomC{$\{\Ex{\TInv}(\typingX) \wedge \TAsrtX\} \ExprX \{\Ex{\TInv}(\typingX) \wedge \TAsrtY[\mathit{this}.\VarIX := \result]\}$}
\RightLabel{with $\VarIX \in \Var_I$}
\UnaryInfC{$\{\Ex{\TInv}(\typingX) \wedge \TAsrtX\} \VarIX := \ExprX \{\Ex{\TInv}(\typingX) \wedge \TAsrtY\}$}
\DisplayProof
\end{center}
where $\TAsrtX \rightarrow \llbracket this \rrbracket \in \{C\}$, $\TAsrtY \rightarrow \llbracket \result \rrbracket \in \UnionTypeX$,
$\UnionTypeX \in 2^{\Class}$, $\UnionTypeX \sqsubseteq \typingX(\llbracket V_{C.\VarIX} \rrbracket)$.

\noindent AXIOM: $\restricted$-NEW
\begin{center}
$\{ \Ex{\TInv}(\typingX) \wedge \TAsrtX[\result := \mathbf{new}_C] \} \mathbf{new}_C \{ \Ex{\TInv}(\typingX) \wedge  \TAsrtX\}$
\end{center}

Formally deriving soundness of these rules requires intricate details of the substitutions involved and rather lengthy proofs. Intuitively, assignment to instance variables preserves the global typing invariant if the assignment is compatible with the typing ($\UnionTypeX \sqsubseteq \typingX(\llbracket V_{C.\VarIX} \rrbracket)$) and object creation does so because it initializes all instance variables to $\Vnull$ and thus the newly created object satisfies the global typing invariant. For a detailed treatment of the substitutions, the interested reader is refered to \cite{DeBoerPierik2003,AptDeBoerOlderogBook2009}.

The function $\Xi$ defined below will construct typing proofs from proof steps like
\AxiomC{$\phi_1$}
\AxiomC{$\ldots$}
\AxiomC{$\phi_n$}
\RightLabel{\scriptsize(RULE)}
\TrinaryInfC{$X$}
\DisplayProof

\begin{floatingfigure}[l]{4.5cm}
\vspace{0.1cm}
\end{floatingfigure}
\noindent where $X$ is a conclusion of the form $A \deriv \{p\} S \{q\}$, RULE is the name of the proof rule applied, and the $\phi_i$ for $i \in \mathbb{N}_n$ are subproofs establishing the premises. For reasoning about recursive method calls, the REC rule needs a set of assumptions about the methods in $\pi$.
\begin{definition}\label{def_assumptions}
 \renewcommand{\PL}[1]{\emph{#1}}
For a \textbf{dyn} program $\pi$ with classes $\Class$ each with a set of methods $\Method_{\PL{C}}$ and a typing $\typingX$ for $\pi$, the set $\mathcal{A}_{\pi,\typingX}$ of method call assumptions is
\[ \mathcal{A}_{\pi,\typingX} \defined \{ \{ \pCm_{\PL{C.m}} \} \PL{v}_0.\PL{m}(\PL{v}_1,...,\PL{v}_n) \{ \qCm_{\PL{C.m}} \} \mid \exists \PL{C} \in \Class, \PL{m} \in \Method_{\PL{C}} \} \]
%o
where $\pCm_{\PL{C.m}} \defined \TInv(\typingX) \wedge \llbracket \PL{v}_0 \rrbracket \in \{\PL{C}\} \wedge \bigwedge_{i = 1}^n \llbracket \PL{v}_i \rrbracket \in \typingX(\llbracket P^i_{\PL{C.m}} \rrbracket) \wedge \bigwedge_{i = 1}^k \llbracket \PL{v}_0.\PL{@v}_i \rrbracket \in \typingX(V_{\PL{C.@v}_i})$ and 
$\qCm_{\PL{C.m}} \defined \TInv(\typingX) \wedge \llbracket \result \rrbracket \in \typingX(\llbracket R_{\PL{C.m}} \rrbracket)$, method $\PL{m}$ is of arity $n$ and $\PL{C}$ has $k$ instance variables.
\end{definition}

\todo{the function $Xi$ is only an example for translating $V_{Ex}$ typings into Hoare logic. Stress again that this is always possible.}
\todo{think about the order of preliminaries}
We are now ready to state the definition of $\Xi$:
\todo{necessary?}
% We will in the following define a function $\Xi$ that automatically assembles Hoare logic proofs. It should thus be explicitly stated that we are operating in the very restricted domain of typing statements that only use typing assertions and that can be derived by the rules of $\restricted$-Hoare logic alone. Note also that the typing $s$ is used to ensure the rules fit together. We thus build on the work of the type safety verifier and merely translate it into the language of Hoare logic.
%
\begin{definition} \textbf{Translation $\Xi$ for Programs} \newline
 \renewcommand{\PL}[1]{\emph{#1}}
Given a \textbf{dyn} program $\pi$ with a main statement $S$ and a set of methods $\Method$ and a $\VEx$-typing $\typingX$ for $\pi$, the function $\Xi(\pi,\typingX)$ yields a typing proof for $\typingX$. It is defined as follows:
\begin{center}
$\Xi(\pi,\typingX) \defined$
\AxiomC{$\begin{matrix}
          \Xi(\mathcal{A}_{\pi,\typingX} \deriv \{true\} S \{true\},\typingX) \\
          \Xi(\mathcal{A}_{\pi,\typingX} \deriv b_{\PL{C.m}},\typingX) \text{ for } \PL{C.m}(\Vector{\PL{u}}) \{S_{\PL{C.m}}\} \in \Method
         \end{matrix}$}
\RightLabel{\scriptsize(REC)}
\UnaryInfC{$\deriv \{\mathit{true}\} S \{ \mathit{true} \}$} 
\DisplayProof
\end{center}

\noindent with $b_{\PL{C.m}} \defined \{\text{\u{p}}_{\PL{C.m}}\} \PL{begin local this}, \Vector{\PL{u}} := \Vector{\PL{v}}; S_{\PL{C.m}} \text{ end} \{\text{\u{q}}_{\PL{C.m}}\}$ where $\pCm_{\PL{C.m}}$ and $\qCm_{\PL{C.m}}$ are given in Definition~\ref{def_assumptions}.
\end{definition}
\begin{definition} \textbf{Translation $\Xi$ for Statements and Expressions} \newline
\renewcommand{\PL}[1]{\emph{#1}}
 For a given $\VEx$-typing $\typingX$ of a \textbf{dyn} statement $S$ and a Hoare logic statement $X$ of the form
 \[X \;\; \equiv \;\; A \deriv \{\TInv(\typingX) \wedge \typingassertion(\typingX,\pre{S})\} S \{\TInv(\typingX) \wedge \typingassertion(\typingX,\post{S})\}\]
 with $A$ a set of assumptions, $\Xi(X,\typingX)$ yields a typing proof for $\typingX$ with precondition $\typingassertion(\typingX,\pre{S})$ and under the assumptions $A$. % and is defined inductively over the structure of $S$ as follows.
% old: In all cases, unless stated otherwise we assume $\BaseTypeX = s(\llbracket S \rrbracket), p' \equiv p$.
% old: $p' \equiv p$ or $p' \equiv \TInv_s$

$\Xi$ is defined inductively over the structure of $S$ and in essence models the reasoning of the verifier $\Ex{V}$ as an equivalent combination of Hoare logic rule applications. The Hoare triples of above form are then assembled into a full typing proof.

% We give only the most interesting cases, the full definition can be found in Appendix~\ref{}.
\todo{some cases -- assigment to instance vars}
  \begin{itemize}
  \item[if] $S \equiv \PL{null}$, $\typingassertion(\typingX,\pre{S}) \equiv \typingassertion(\typingX,\post{S})[\result := \Vnull]$ then
  \begin{center}
   $\Xi(X,\typingX) =$
   \AxiomC{}
%    \RightLabel{\scriptsize(CONST)}
%    \UnaryInfC{$A \deriv \{\typingassertion(\typingX,\pre{S})\} \PL{null} \{ \typingassertion(\typingX,\post{S})}    
   \RightLabel{\scriptsize(CONST)\footnote{\label{footnote-label} $\TInv(\typingX)$ does not contain $\result$ and is hence invariant under the substitution}}
   \UnaryInfC{$A \deriv \{\TInv(\typingX) \wedge \typingassertion(\typingX,\pre{S})\} \PL{null} \{ \TInv(\typingX) \wedge \typingassertion(\typingX,\post{S}) \}$} 
   \DisplayProof
  \end{center}

  \item[if] $S \equiv \VarLX$, $\typingassertion(\typingX,\pre{S}) \equiv \typingassertion(\typingX,\post{S})[\result := \VarLX]$ then
  \begin{center}
   $\Xi(X,\typingX) =$
   \AxiomC{}
   \RightLabel{\scriptsize(VAR)}
   \UnaryInfC{$A \deriv \{\TInv(\typingX) \wedge \typingassertion(\typingX,\pre{S})\} \VarLX \{ \TInv(\typingX) \wedge \typingassertion(\typingX,\post{S}) \}$}
%    \RightLabel{\scriptsize(INV)}
%    \UnaryInfC{$A \deriv \{\TInv(\typingX) \wedge \typingassertion(\typingX,\pre{S})\} \VarLX \{ \TInv(\typingX) \wedge \typingassertion(\typingX,\post{S}) \}$} 
   \DisplayProof
  \end{center}

  \item[if] $S \equiv \VarIX$, $\typingassertion(\typingX,\pre{S}) \equiv \typingassertion(\typingX,\post{S})[\mathit{this}.\VarIX := \result]$ then
  \begin{center}
   $\Xi(X,\typingX) =$
   \AxiomC{}
   \RightLabel{\scriptsize(IVAR)}
   \UnaryInfC{$A \deriv \{\TInv(\typingX) \wedge \typingassertion(\typingX,\pre{S})\} \VarIX \{ \TInv(\typingX) \wedge \typingassertion(\typingX,\post{S}) \}$}
   \DisplayProof
  \end{center}

  \item[if] $S \equiv \VarLX := e$, $\typingassertion(\typingX, \pre{S}) \rightarrow \typingassertion(\typingX,\pre{e})$, $\typingassertion(\typingX,\post{e}) \rightarrow \typingassertion(\typingX,\post{S})[\VarLX := \result]$ then
  \begin{center}
   $\Xi(X,\typingX) =$
   \AxiomC{$A \deriv \{\TInv(\typingX) \wedge \typingassertion(\typingX,\pre{e})\} e \{ \TInv(\typingX) \wedge \typingassertion(\typingX,\post{e}) \}$}
   \RightLabel{\scriptsize(CONS)}
   \UnaryInfC{$A \deriv \{\TInv(\typingX) \wedge \typingassertion(\typingX,\pre{S})\} e \{ \TInv(\typingX) \wedge \typingassertion(\typingX,\post{S})[\VarLX := \result] \}$}
   \RightLabel{\scriptsize(ASGN)}
   \UnaryInfC{$A \deriv \{\TInv(\typingX) \wedge \typingassertion(\typingX,\pre{S})\} \VarLX := e \{ \TInv(\typingX) \wedge \typingassertion(\typingX,\post{S}) \}$}
   \DisplayProof
  \end{center}

  \item[if] $S \equiv \VarIX := e$, $\typingassertion(\typingX, \pre{S}) \rightarrow \typingassertion(\typingX,\pre{e})$, $\typingassertion(\typingX,\post{e}) \rightarrow \typingassertion(\typingX,\post{S})[\mathit{this}.\VarIX := \result]$ then
  \begin{center}
   $\Xi(X,\typingX) =$
   \AxiomC{$A \deriv \{\TInv(\typingX) \wedge \typingassertion(\typingX,\pre{e})\} e \{ \TInv(\typingX) \wedge \typingassertion(\typingX,\post{e}) \}$}
   \RightLabel{\scriptsize(CONS)}
   \UnaryInfC{$A \deriv \{\TInv(\typingX) \wedge \typingassertion(\typingX,\pre{S})\} e \{ \TInv(\typingX) \wedge \typingassertion(\typingX,\post{S})[\mathit{this}.\VarIX := \result] \}$}
   \RightLabel{\scriptsize($\restricted$-IASGN)}
   \UnaryInfC{$A \deriv \{\TInv(\typingX) \wedge \typingassertion(\typingX,\pre{S})\} \VarIX := e \{ \TInv(\typingX) \wedge \typingassertion(\typingX,\post{S}) \}$}
   \DisplayProof
  \end{center}

  \item[if] $S \equiv S \equiv \PL{if} e \PL{ then} S_1 \PL{else} S_2 \PL{ end}$, $\typingassertion(\typingX,\post{e}) \rightarrow \llbracket \result \rrbracket \in \{\mathit{bool}\}$, $\typingassertion(\typingX,\pre{S}) \rightarrow \typingassertion(\typingX,\pre{e})$, $\typingassertion(\typingX,\post{e}) \rightarrow \typingassertion(\typingX,\pre{S_1})$, $\typingassertion(\typingX,\post{e}) \rightarrow \typingassertion(\typingX,\pre{S_2})$, $\typingassertion(\typingX,\post{S_1}) \rightarrow \typingassertion(\typingX,\post{S})$, $\typingassertion(\typingX,\post{S_2}) \rightarrow \typingassertion(\typingX,\post{S})$ then
  \begin{center}
   $\Xi(X,\typingX) =$
   \AxiomC{$\begin{matrix}
                  A \deriv \{\TInv(\typingX) \wedge \typingassertion(\typingX,\pre{e})\} e \{ \TInv(\typingX) \wedge \typingassertion(\typingX,\post{e}) \} \\
                  A \deriv \{\TInv(\typingX) \wedge \typingassertion(\typingX,\pre{S_1})\} S_1 \{ \TInv(\typingX) \wedge \typingassertion(\typingX,\post{S_1}) \} \\
                  A \deriv \{\TInv(\typingX) \wedge \typingassertion(\typingX,\pre{S_2})\} S_2 \{ \TInv(\typingX) \wedge \typingassertion(\typingX,\post{S_2}) \}
                \end{matrix}$}

   \RightLabel{\scriptsize(CONS)}

   \UnaryInfC{$\begin{matrix}
                  A \deriv \{\TInv(\typingX) \wedge \typingassertion(\typingX,\pre{e})\} e \{ \TInv(\typingX) \wedge \typingassertion(\typingX,\post{e}) \} \\
                  A \deriv \{\TInv(\typingX) \wedge \typingassertion(\typingX,\post{e})\} S_1 \{ \TInv(\typingX) \wedge \typingassertion(\typingX,\post{S_1}) \} \\
                  A \deriv \{\TInv(\typingX) \wedge \typingassertion(\typingX,\post{e})\} S_2 \{ \TInv(\typingX) \wedge \typingassertion(\typingX,\post{S_2}) \}
                \end{matrix}$}
   \RightLabel{\scriptsize(CONS)}
   \UnaryInfC{$\begin{matrix}
                  A \deriv \{\TInv(\typingX) \wedge \typingassertion(\typingX,\pre{S})\} e \{ \TInv(\typingX) \wedge \typingassertion(\typingX,\post{S}) \} \\
                  A \deriv \{\TInv(\typingX) \wedge \typingassertion(\typingX,\pre{S})\} S_1 \{ \TInv(\typingX) \wedge \typingassertion(\typingX,\post{S}) \} \\
                  A \deriv \{\TInv(\typingX) \wedge \typingassertion(\typingX,\pre{S})\} S_2 \{ \TInv(\typingX) \wedge \typingassertion(\typingX,\post{S}) \}
                \end{matrix}$}
   \RightLabel{\scriptsize(COND)}
   \UnaryInfC{$A \deriv \{\TInv(\typingX) \wedge \typingassertion(\typingX,\pre{S})\} \PL{if } e \PL{ then} S_1 \PL{ else } S_2 \PL{ end} \{ \TInv(\typingX) \wedge \typingassertion(\typingX,\post{S}) \}$}
   \DisplayProof
  \end{center}

  \item[if] $S \equiv e_0.\MethodX(e_1,...,e_n)$, $\typingassertion(\typingX,\post{e_i}) \rightarrow \typingassertion(\typingX,\pre{e_{i+1}})$ for $i \in \mathbb{N}_n$, $\typingassertion(\typingX,\pre{S}) \rightarrow \typingassertion(\typingX,\pre{e_0})$, $\{\TInv(\typingX) \wedge \typingassertion(\typingX,\post{e_n})\} \VarLX_0.\MethodX(\VarLX_1,...,\VarLX_n) \{ \TInv(\typingX) \wedge \typingassertion(\typingX,\post{S})\} \in A$ then
  \begin{center}
   $\Xi(X,\typingX) =$
   \AxiomC{$\begin{matrix}
                  A \deriv \{\TInv(\typingX) \wedge \typingassertion(\typingX,\pre{e_i})\} e_i \{ \TInv(\typingX) \wedge \typingassertion(\typingX,\post{e_i}) \} \text{ for } i \in \mathbb{N}_n\\
                \end{matrix}$}
   \RightLabel{\scriptsize(COND + $A$)}
   \UnaryInfC{$\begin{matrix}
                  A \deriv \{\TInv(\typingX) \wedge \typingassertion(\typingX,\pre{S})\} e_i \{ \TInv(\typingX) \wedge \typingassertion(\typingX,\pre{S})[\VarLX_i := \result] \} \text{ for } i \in \mathbb{N}_n\\
                  A \deriv \{\TInv(\typingX) \wedge \typingassertion(\typingX,\pre{S})\} \VarLX_0.\MethodX(\VarLX_1,...,\VarLX_n) \{ \TInv(\typingX) \wedge \typingassertion(\typingX,\post{S}) \}
                \end{matrix}$}
   \RightLabel{\scriptsize(METH)}
   \UnaryInfC{$A \deriv \{\TInv(\typingX) \wedge \typingassertion(\typingX,\pre{S})\} e_0.\MethodX(e_1,...,e_n) \{ \TInv(\typingX) \wedge \typingassertion(\typingX,\post{S}) \}$}
   \DisplayProof
  \end{center}
  
  \end{itemize}
\end{definition}
%
%Note that $\Xi$ only translates consistent typings.
%
% The cases for assignment, conditional and method call given apply their respective proof rules, checking
% both their side conditions and
% the solution $s$ for consistency.
%the case for assignment
% Disassembling the given precondition $p$ into $\dot{p} \wedge \llbracket v \rrbracket \in s(\llbracket v \rrbracket)$
% is computable due to the restriction to typing assertions.
% The $\CheckIf(b)$-function is used in the rules for conditional and method call to allow the controlled
% passing of typing inconsistencies into the constructed proof attempt.
%
% The given translation does not allow translating arbitrary typings. The application
% conditions in each case in effect perform a type check. This is ok, as we do not want to translate
% arbitrary typings. However, in the following cases we deliberately allowed
% the typing to become inconsistent:
% In applications of the $\restricted$-METH rule, the
%        statement $\{\pCm_{{C_j}.m}\} l.m(v_1,...,v_n) \{\qCm_{{C_j}.m}\}$ for any $C_j \in \BaseTypeX_0$
%        might be missing in the assumptions $A$ and
% the premise for $B$ might admit a broader type than $\{bool\}$ in applications of the
%        $\restricted$-COND or $\restricted$-LOOP rule.
% 
% Note that these cases exactly correspond to the ways in which the TI-results
% may be inconsistent (listed in Section~\ref{sec:type_inference}) and that in all cases,
% the corresponding rule application would be labeled with ``FAIL''. We thus state the following lemma
%
\begin{lemma}
\label{lem:cons-valid}
 For every \textbf{dyn} program $\pi$ and every consistent $\VEx$-typing $\typingX$ of $\pi$, $\Xi(\pi,\typingX)$ can be constructed and is a valid typing proof for $\typingX$ in $\pi$.
\end{lemma}
\begin{proof}
 By induction over the structure of the program $\pi$ comparing the application conditions of Hoare logic rules, the typing rules for $\VEx$ and the preconditions for the respective cases in the translation $\Xi$. \hfill $\qed$
\end{proof}

Note that this implies soundness of $\VEx$.

%% file: appendices/app_proofs.tex
\paragraph{Proof for Theorem~\ref{thm:layer_assertion_homomorphism}}
\begin{proof}
 By definition of $\Theta(\sigma)$, for all variables $\VarX$ of a base type $\BaseTypeX$ in $\sigma$, $\Theta(\sigma) \models \mathit{safe}_\BaseTypeX(\VarX)$ holds and $\VarX$ can hence be safely mapped. Under the assumption that for all such variables $\VarX$ it holds that $\llbracket \VarXVirtual \rrbracket(\Theta(\sigma)) = \llbracket \VarX \rrbracket(\sigma)$, the following lemma can be established by induction over the structure of the assertion language: $\llbracket l \rrbracket(\sigma) = \llbracket \Theta(l) \rrbracket(\Theta(\sigma))$ for all logical expressions $l$ and \textbf{stat} states $\sigma$. As the assumption is guaranteed by the mapping predicates introduced by $\Upsilon_M$, the desired result can then be established by induction over the structure of the assertion language. \hfill $\qed$
\end{proof}

\paragraph{Proof for Theorem~\ref{thm:translation_stat_to_dyn}}
\begin{proof}
 By induction over the structure of the proof $\phi$, using Theorem~\ref{thm:layer_assertion_homomorphism} and the fact that the application conditions for the pure expression rules are satisfied when $S$ is a statically typed program and all assertions where translated using $\Theta$.\todo{actually do the proof!}
\end{proof}

\paragraph{Proof for Lemma~\ref{lem:extractability}}
\begin{proof}
 $\AsrtX' \equiv \AsrtX \wedge \TAsrtX$ has the described properties.
\end{proof}

\paragraph{Proof for Theorem~\ref{thm:extractability}}
\begin{proof}
 Follows from Lemma~\ref{lem:extractablity} and the definition of the most precise typing assertion.
\end{proof}

\paragraph{Proof for Theorem~\ref{thm:relative_completeness}}
\begin{proof}
 Follows from completeness of the Hoare logic, Theorem~\ref{thm:extractability} and the fact that type safety proofs must establish the absence of type errors.
\end{proof}

\paragraph{Proof for Lemma~\ref{lem:cons-valid}}
\begin{proof}
 By comparing the application conditions for cases of $\Xi$ with the typing rules given in Figure~\ref{fig:typing-rules} and the side conditions of applied Hoare-logic rules.
 For instance, in the case for assignment, consistency of the typing implies $\typingX(\llbracket e \rrbracket) = \typingX(\llbracket \PL{u}_i := e \rrbracket)$ and $\typingX(\llbracket e \rrbracket) \subseteq \typingX(\llbracket \PL{u}_i \rrbracket)$, by the definitions for SSA it follows $\typingX(\llbracket e \rrbracket) = \typingX(\llbracket \PL{u} := e \rrbracket)$ and $\typingX(\llbracket e \rrbracket) \subseteq \typingX(\llbracket \PL{u} \rrbracket)$ and that $\PL{u}$ is the only variable whose (flow-sensitive) type may have changed between $\post{e}$ and $\post{S}$ ($\result$ stays the same). With the definition of $\typingassertion(\typingX,L)$ we conclude that $\typingassertion(\typingX,\post{e}) \leftrightarrow \typingassertion(\typingX,\post{S})[\PL{u} := \result] \wedge \llbracket \PL{u} \rrbracket \in t$ for some type $t$. \hfill $\qed$
\end{proof}

\paragraph{Proof of Theorem~\ref{thm:refinement_monotonicity}}
\begin{proof}
 Let $\typingX \stackrel{\tau,L}{\rightarrow} \typingY$ be a conjunctive refinement step
 and $\PL{x} \in \Var_\PL{C.m}$. Then $\typingY(\llbracket \PL{x} \rrbracket, L) = \typingX(\llbracket \PL{x} \rrbracket, L) \sqcap \Omega^{\PL{x}}_X(\tau) \sqsubseteq \typingX(\llbracket \PL{x} \rrbracket, L)$. This difference is induced by the constraints generated from $\mathcal{R}_\tau$. Since all other constraints are identical between $\programX$ and $\programX'$ and all constraints are monotonic, $\typingY(\llbracket S \rrbracket, L) \sqsubseteq \typingX(\llbracket S \rrbracket, L)$ for all substatements $\StmtX$ of $\programX$ and consequently $\typingY \sqsubseteq \typingX$ follows by induction over the constraint system.\hfill $\qed$
\end{proof}

\paragraph{Proof for Theorem~\ref{thm:two-layered-proof-construction}}
\hspace{0.5cm}\newline
For our proof we need the following definition:
\begin{definition}[Fusion of Hoare Logic Proofs]
  Let $\phi$ be a Hoare logic proof for $\{\AsrtX\} \StmtX \{\AsrtY\}$ in some notion of correctness $X$ and $\varphi$ be a typing proof for $\{\TAsrtX\} \StmtX \{\TAsrtY\}$. Then, the fusion $\phi + \varphi$ is a two-layered proof for $\{\AsrtX \wedge \TAsrtX\} \StmtX \{\AsrtY \wedge \TAsrtY\}$ in the sense of $X$-correctness.
\end{definition}

WLOG, we assume $\phi$ and $\varphi$ to be minimal (all Hoare triples contribute to the proof's conclusion). They hence have a tree-like structure. Their fusion can then be constructed by recursion over this structure. 

Induction basis = Fusing axioms. All axioms in our Hoare logic (Appendix~\ref{app:hoare_logic}), are invariant under conjunction: if $\{\AsrtX\} \StmtX \{\AsrtY\}$ and $\{\TAsrtX\} \StmtX \{\TAsrtY\}$ can be derived using this axiom, then $\{\AsrtX \wedge \TAsrtX\} \StmtX \{\AsrtY \wedge \TAsrtY\}$ can also.

Induction step = Fusing rules. All rules in our Hoare logic have the following properties

\begin{itemize}
 \item Invariant under fusion: If
 
  \AxiomC{$\{\AsrtX_1\} \StmtX_1 \{\AsrtY_1\}, ..., \{\AsrtX_n\} \StmtX_n \{\AsrtY_n\}$}
  \RightLabel{\scriptsize(X)}
  \UnaryInfC{$\{\AsrtX\} \StmtX \{\AsrtY\}$}
  \DisplayProof
  
  and
  \AxiomC{$\{\TAsrtX_1\} \StmtX_1 \{\TAsrtY_1\}, ..., \{\TAsrtX_n\} \StmtX_n \{\TAsrtY_n\}$}
  \RightLabel{\scriptsize(X)}
  \UnaryInfC{$\{\TAsrtX\} \StmtX \{\TAsrtY\}$}
  \DisplayProof
  
  are valid rule applications, then
  
  \AxiomC{$\{\AsrtX_1 \wedge \TAsrtX_1\} \StmtX_1 \{\AsrtY_1 \wedge \TAsrtY_1\}, ..., \{\AsrtX_n \wedge \TAsrtX_n\} \StmtX_n \{\AsrtY_n \wedge \TAsrtY_n\}$}
  \RightLabel{\scriptsize(X)}
  \UnaryInfC{$\{\AsrtX \wedge \TAsrtX\} \StmtX \{\AsrtY \wedge \TAsrtY\}$}
  \DisplayProof
  
  is also.

 \item They are either syntax-directed (and hence must appear in both proofs) or have a neutral application   \AxiomC{$\{\AsrtX\} \StmtX \{\AsrtY\}$}
  \RightLabel{\scriptsize(X)}
  \UnaryInfC{$\{\AsrtX\} \StmtX \{\AsrtY\}$} 
  \DisplayProof that can be inserted into a proof to make its structure match the other one (having an application of rule $X$).
\end{itemize}

For most rules, this is obvious. For applications of CONJ and DISJ, one needs to fuse the proof with both premises. To see that the properties hold for the SUBST rule, consider that all variables occurring in typing assertions are being read in some method of the program (otherwise, typing them is useless). Hence, the side-condition of the SUBST rule does not allow them to be substituted for and all applications of this rule hence are neutral for all typing assertions.

Both proofs can hence be made structurally equivalent by inserting neutral rule applications and then fused using the invariance property. \hfill $\qed$

% Since Hoare logic proofs are syntax-directed, for every sub-statement $S$, the Hoare triple $\{\tau\} S \{\tau'\}$ in $\phi$ can be fused with every Hoare triple $\{p\} S \{q\}$ in $\varphi$ using the CONJ rule. To see that this is sufficient to generate a two-layered proof, observe that all non-syntax directed rules except SUBST are invariant under this. For instance if
% 
%   \AxiomC{$p \rightarrow p_1$, $\{p_1\} S \{q_1\}$, $q_1 \rightarrow q$}
%   \RightLabel{\scriptsize(CONS)}
%   \UnaryInfC{$\{p\} S \{q\}$} 
%   \DisplayProof
% 
% is a valid application of the CONS rule, then
% 
%   \AxiomC{$p \wedge \tau \rightarrow p_1 \wedge \tau$, $\{p_1 \wedge \tau\} S \{q_1 \wedge \tau'\}$, $q_1 \wedge \tau' \rightarrow q \wedge \tau'$}
%   \RightLabel{\scriptsize(CONS)}
%   \UnaryInfC{$\{p \wedge \tau\} S \{q \wedge \tau'\}$} 
%   \DisplayProof
% 
% also is. While this in principle also holds for the SUBST rule, it requires a few more rule applications:
% 
%   \AxiomC{$\{p \wedge \tau\} S \{q \wedge \tau'\}$}
%   \UnaryInfC{$\{p\} S \{q\}$}
%   \RightLabel{\scriptsize(SUBST)}
%   \UnaryInfC{$\{p[\Vector{z} := \Vector{t}] \wedge \tau\} S \{q[\Vector{z} := \Vector{t}] \wedge \tau'\}$} 
%   \DisplayProof
% 
% \todo{finish proof}
%   